\documentclass[12pt]{article}
\usepackage{amssymb,amsfonts,amsthm,amsmath,calligra}
\usepackage{slashed}
\usepackage{yfonts}
\usepackage{mathrsfs,pifont}

\usepackage[all]{xy}

\usepackage{amssymb,amsfonts,amsthm,amsmath}
\usepackage{slashed}
\usepackage{yfonts}
\usepackage{mathrsfs,pifont}

\usepackage{color}

\usepackage{graphicx}

\usepackage{accents}
\newlength{\dhatheight}
\newcommand{\doublehat}[1]{%
    \settoheight{\dhatheight}{\ensuremath{\widehat{#1}}}%
    \addtolength{\dhatheight}{-0.35ex}%
    \widehat{\vphantom{\rule{1pt}{\dhatheight}}%
    \smash{\widehat{#1}}}}

\def\hZ {{\textbf{Z}}}
\def\calK{{{\cal{K}}}}

\def\calA{{{\cal{A}}}}

\def\matZ{{\mathbb{Z}}}
\def\matQ{{\mathbb{Q}}}
\def\matC{{\mathbb{C}}}

\def\M{ {\textbf{M}}  }

\def\sldh{{\mathscr{U}_{\hbar}(\widehat{\frak{gl}}_2)}}

\def\gt{{U_{q}(\doublehat{\frak{gl}}_1)}}

\def\slope{s} 
\def\dz {z}   
\def\Wc {B}  
\def\tb {{\cal{V}}} 
\def\KG {{G}} 
\def\qmV {{\mathscr{V}}} 
\def\qmW {{\mathscr{W}}} 
\def\qmQ {{\mathscr{Q}}} 
\def\qmM {{\mathscr{M}}} 
\def\qm {\textsf{QM}} 
\def \vss {\widehat{{\cO}}_{\textrm{vir}}} 
\def \capping{{\textbf{J}}} 
\def \ev {{\textrm{ev}}} 
\def \shift {\mathsf{S}} 
\def \Rtot {{\mathscr{R}}} 
\def \Rwal {{\sf{R}}} 
\def \geomM {{\mathbf{M}}} 
\def \repM {{\mathbf{B}}} 
\def \wJ {{\textbf{J}}} 
\def \T {{T}} 
\def\qKZ {{\mathscr{K}}} 
\def\qA {{\mathscr{A}}} 
\def\be{\begin{eqnarray}}
\def\ee{\end{eqnarray}}

\newcommand{\C}{\mathbb{C}}

\newcommand{\Z}{\mathbb{Z}}
\newcommand{\R}{\mathbb{R}}
\newcommand{\B}{\mathbb{B}}
\newcommand{\Q}{\mathbb{Q}}
\newcommand{\N}{\mathbb{N}}
\newcommand{\bv}{\mathsf{v}}

\newcommand{\bw}{\mathsf{w}}
\newcommand{\bmu}{\boldsymbol{\mu}}
\newcommand{\bal}{\widehat{\alpha}}

\newcommand{\Uq}{\mathscr{U}_\hbar}
\newcommand{\fg}{\mathfrak{g}}
\newcommand{\fh}{\mathfrak{h}}
\newcommand{\fgh}{\widehat{\fg}}
\newcommand{\fhh}{\widehat{\fh}}
\newcommand{\fgKM}{\mathfrak{g_\textup{\tiny \sc KM}}}
\newcommand{\fghKM}{\widehat{\fg}_\textup{\tiny\sc KM}}

\newcommand{\cU}{\mathscr{U}}
\newcommand{\cS}{\mathscr{S}}
\newcommand{\cO}{\mathscr{O}}
\newcommand{\cL}{\mathscr{L}}

\newcommand{\cE}{\mathscr{E}}

\newcommand{\cM}{\mathscr{M}}
\newcommand{\cN}{\mathscr{N}}

\newcommand{\rr}{ } 

\newcommand{\bb}{} 

\newcommand{\bA}{\mathsf{A}}
\newcommand{\bT}{G}
\newcommand{\fC}{\mathfrak{C}}
\newcommand{\Topp}{T^{1/2}_\textup{opp}}

\newcommand{\Ct}{{\C^{\times}}}
\newcommand{\pt}{\textup{pt}}

\newcommand{\Hd}{{H}^{\raisebox{0.5mm}{$\scriptscriptstyle \bullet$}}}

\newcommand{\Ldm}{{\Lambda}^{\!\raisebox{0.5mm}{$\scriptscriptstyle \bullet$}\,}_{\scriptscriptstyle{-}}}

\DeclareMathOperator{\Hom}{\mathscr{H}\text{\kern -3pt {\calligra\large om}}\,}

\newcommand{\bkap}{\boldsymbol\kappa}

\DeclareMathOperator{\Aut}{Aut}
\DeclareMathOperator{\Pic}{Pic}
\DeclareMathOperator{\Hilb}{Hilb}
\DeclareMathOperator{\Stab}{Stab}
\DeclareMathOperator{\Lie}{Lie}
\DeclareMathOperator{\Attr}{Attr}
\DeclareMathOperator{\diag}{diag}
\DeclareMathOperator{\rk}{rk}
\DeclareMathOperator{\Coh}{Coh}
\DeclareMathOperator{\codim}{codim}

\newtheorem{Corollary}{Corollary}
\newtheorem{Theorem}{Theorem}
\newtheorem{Lemma}{Lemma}
\newtheorem{Proposition}{Proposition}

\theoremstyle{definition}

\newtheorem{Remark}{Remark}

\begin{document}

\title{Quantum difference equation for Nakajima varieties}
\author{A. Okounkov and A.  Smirnov}
\date{}
\maketitle

\begin{abstract}
For an arbitrary Nakajima quiver variety $X$, we construct
an analog of the quantum dynamical Weyl group acting
in its equivariant K-theory. The correct generalization of
the Weyl group here is the fundamental groupoid of a certain
periodic locally finite hyperplane arrangement in $\Pic(X)\otimes
\C$. We identify the lattice part of this groupoid with the
operators of quantum
difference equation for $X$. The cases of quivers of finite
and affine type are illustrated by explicit examples.
\end{abstract}

\setcounter{tocdepth}{2}
\tableofcontents

\section{Introduction}

\subsection{The quantum differential equation}

\subsubsection{}

This paper is about enumerative $K$-theory of rational curves in
Nakajima quiver varieties. The cohomological
version of the questions that we answer here may be
asked very generally, for example one may replace a Nakajima
variety $X$ by a general smooth quasiprojective variety over $\C$
as long as rational curves in $X$ satisfy certain properness
conditions.

Consider the cone of effective curves in $H_2(X,\Z)$ and its
semigroup algebra spanned by monomials $z^d$, where
$d\in H_2(X,\Z)_\textup{effective}$. It has  a natural
completion which we denote $\C[[z^d]]$. The cup product in
$\Hd(X,\C)$ has an associative supercommutative deformation
\begin{equation}
\alpha \star \beta = \alpha \cup \beta + O(z) \,,\label{star_p}
\end{equation}
parametrized by $\C[[z^d]]$,
in which one counts not only triple intersections of cycles but
also rational curves meeting three given cycles, see \cite{CK}
for an introduction. The corresponding algebra is known as
the quantum cohomology of $X$.
The construction works equivariantly
with respect to $\Aut(X)$; in what follows, it will be important
to work equivariantly with respect to a torus $\bT \subset \Aut(X)$.

Associated to \eqref{star_p}  is a remarkable
\emph{flat} connection on the trivial $\Hd_\bT(X,\C)$-bundle over
$\textup{Spec}\, \C[[z^d]]$ known as the
quantum connection, the Dubrovin connection, or the
quantum differential equation. It has the form
\begin{equation}
\frac{d}{d\lambda} \Psi(z) = \lambda \star \Psi(z) \,, \quad \Psi(z)
\in \Hd(X)\,,
\quad \frac{d}{d\lambda} z^d = (\lambda,d) \, z^d \,,\label{Dubr}
\end{equation}
where $\lambda\in H^2(X,\C)$. Flat sections of this connection
play a very important enumerative role.

\subsubsection{}

For Nakajima varieties, the formal series in $z$ in \eqref{star_p}
converge to rational functions, and the connection extends as a connection with
regular singularities to a certain toric compactification
$$
\textup{K\"ahler moduli space} \supset \Pic(X)
\otimes \Ct \owns z \,.
$$
In fact, the following representation-theoretic
interpretation of this connection was proven in \cite{MO}.

Recall that Nakajima quiver varieties \cite{NakALE,NakQv}
play a central role in
geometric representation theory and very interesting algebras
act by correspondences between Nakajima varieties. In particular,
quantum loop algebras $\cU_\hbar(\fghKM)$
associated to a Kac-Moody Lie algebra $\fgKM$ were
realized geometrically by Nakajima in equivariant K-theory of his quiver
varieties, see \cite{Nakfd}.
 Parallel results for Yangians
$Y(\fgKM)$ in cohomology were
proven by Varagnolo in \cite{Varan}.

A representation-theoretic description of the quantum
differential equation requires a certain larger Lie algebra
$\fg \supset \fgKM$. It coincides with the Kac-Moody
Lie algebra for quivers of finite ADE type and otherwise can be
significantly larger. This Lie algebra, together with the
corresponding Yangian $Y(\fg)$, was constructed in \cite{MO}.
This construction will be recalled in Section \ref{frt} below, in the
generality of quantum loop algebras.

The Lie algebra $\fg$ has a root decomposition
$$
\fg = \fh \oplus \bigoplus_{\alpha} \fg_\alpha
$$
in which $\fh = \Pic(X) \otimes \C \oplus \textup{center}$
and $\alpha \in \pm H_2(X,\Z)_\textup{effective}$. The root
subspaces are finite-dimensional and
$\fg_{-\alpha} = \fg_\alpha^*$ with respect to a nondegenerate
symmetric invariant form.

The main
result of \cite{MO} reads
\begin{equation}
  \label{eq1}
c_1(\lambda) \star_\textup{\tiny modif}  =
c_1(\lambda) \cup - \hbar \sum_{\theta\cdot \alpha >0}
(\lambda, \alpha) \, \frac{z^\alpha}{1-z^\alpha} \, e_\alpha
e_{-\alpha} + \dots
\end{equation}
where
$$
\lambda \in \Pic(X) \otimes \C \subset \fh
$$
and
the subscript in $\star_\textup{\tiny modif}$ means a
shift of the form $z^d \mapsto (-1)^{(d,\bkap)} z^d$ for a
certain canonically defined $\bkap\in H^2(X,\Z/2)$. We
will
see a parallel shift in our formulas below (see the footnote after Theorem \ref{okth}).
Further in \eqref{eq1},
$$
\hbar\in H^2_\bT(\pt) = \left(\Lie \bT\right)^*
$$
 is the equivariant weight of
the symplectic form and the pairing $\theta\cdot \alpha$ with
the stability parameter $\theta\in H^2(X,\R)$ selects the
effective representative from each $\pm \alpha$ pair.
The abbreviation
$$
e_\alpha e_{-\alpha} \in \fg_\alpha \fg_{-\alpha} \subset
\cU(\fg)
$$
stands for the image of the canonical element of
$\fg_\alpha \otimes \fg_{-\alpha}$ under multiplication.
Finally, the dots in \eqref{eq1} stand for the a multiple of the
identity. Such normalization ambiguity is typical, and is resolved
e.g.\ by the requirement that the purely quantum part of \eqref{eq1}
annihilates $1\in H^0(X)$. We will see a similar multiplicative
scalar ambiguity in our main formula.

The poles in \eqref{eq1} are contained in
\begin{equation}
\{z^\alpha = 1\,, 0 < \alpha \le \bv \} \label{sing_QDE}\,,
\end{equation}
where $\bv$ is the dimension vector for a given quiver
variety. The condition $\alpha \le \bv$ is necessary for
$\fg_\alpha \Hd(X) \ne 0$ and hence for the occurrence of
the corresponding pole in \eqref{eq1}. The singularities
\eqref{sing_QDE}
 lift to a periodic locally finite arrangement of hyperplanes
\begin{equation}
\{(\lambda,\alpha) \in \Z \,, 0<\alpha\le \bv \} \label{hyper}
\end{equation}
on the universal cover $H^2(X,\C)$ of the
K\"ahler torus $\Pic(X) \otimes \Ct$. These \emph{
affine root hyperplanes}
will play an important role below.

\subsubsection{}

The Yangian $Y(\fg)$ is a certain Hopf algebra deformation of
the algebra $\cU(\fg[t])$ of polynomial loops and one of
its basic features is that the operator $c_1(\lambda)$ is a
deformation of $t\lambda \in \fh[t]$. Thus \eqref{eq1}
becomes an instance of the trigonometric
\emph{Casimir connection}, studied in \cite{Tol1}
for Yangians of
finite-dimensional semisimple Lie algebras, see also
the work \cite{TV1,TV2} by Tarasov and Varchenko.

In fact, the program of constructing the general Yangians
$Y(\fg)$ and
identifying their Casimir connections with the quantum connection
for Nakajima varieties
was born out of conjectures made by Nekrasov and Shatashvili
on one hand \cite{NS1,NS2}
, and Bezrukavnikov and his collaborators ---
on the other.

Already back then it was predicted by Etingof that the
correct K-theoretic version of the quantum connection should
be identified with a similar generalization of
 dynamical difference equations studied by Tarasov, Varchenko,
Etingof, and others (see e.g.\ \cite{TV3,EV}
) for finite-dimensional
Lie algebras $\fg$. In particular, Balagovic proved \cite{Bal}
 that for
a finite-dimensional $\fg$, the dynamical equations
degenerate to the Casimir connection in the appropriate
limit. While both our methods and objects of study
differ significantly from the
above cited works, it is fundamentally this vision of Etingof
that is realized in the present paper.

\subsubsection{}
For quivers of affine ADE type, Nakajima varieties are moduli
of framed coherent sheaves on the corresponding surfaces.
In particular, the Hilbert schemes
$\Hilb(S,\textup{points})$, where $S$ is an
ADE  surface, are
Nakajima varieties. Quantum differential equations for
those were determined earlier in \cite{OP,ObloMal},  and
play a key role in enumerative geometry of curves in
threefolds. Such enumerative theories exist in
different flavors known as the Gromov-Witten and the
Donaldson-Thomas theories\footnote{Here the threefold need
not be Calabi-Yau, to point out a frequent misconception.
For example, the equivariant
Donaldson-Thomas theory of toric varieties is a very rich
subject with many applications in mathematical
physics.}. A highly nontrivial
equivalence between the two was conjectured in \cite{MNOP1,MNOP2}
and its proof for toric varieties given in \cite{MOOP}
rests
on reconstructing both from the quantum difference
equation for the Hilbert schemes of points in $A_n$ surfaces.

In fact, it may be accurate to say that the GW/DT correspondence
in the generality known today, see especially \cite{PP1}
 for
state of the art results, is
proven by breaking the threefolds in pieces until we get to
an ADE surface fibration, for which the computations on
both sides can be equated to a computation in quantum
cohomology of $\Hilb(S,\textup{points})$.  It is not
surprising that such a connection exists, because a curve
$$
C \to \Hilb(S,\textup{points})
$$
defines a subscheme of $C\times S$. However, it is
very important for $S$ to be a symplectic surface for this
correspondence to remain precise enumeratively, and not be corrected
by contributions of nonmatching strata in different moduli
spaces.

As a particular case of our general result, we compute the
quantum difference connection in the quantum K-theory
of $\Hilb(S,\textup{points})$. This has an entirely parallel
use in K-theoretic Donaldson-Thomas theory of threefolds, see \cite{pcmilect}.
There is a great interest in this theory, for instance because
of its conjectural connection to a certain curve-counting in
Calabi-Yau 5-folds, which is expected to be an
algebro-geometric version of computing the contribution of
membranes to the index of M-theory, see \cite{NO}.

\subsubsection{}
Another reason why quantum differential equations are important
is because the conjectures of Bezrukavnikov and his
collaborators relate them to representation theory of
\emph{quantizations} of $X$, see for example \cite{PE} and
also e.g.\ \cite{BL}
for subsequent developments.

Much technical and conceptual progress in
representation theory has been achieved by treating
algebras of interest, such as e.g.\ universal enveloping
algebras of semisimple Lie algebras, as quantizations of
algebraic symplectic varieties, see e.g.\ \cite{BK,BM,Dkad,BF},  especially
in prime characteristic. By construction, Nakajima
varieties are algebraic symplectic reductions of linear
symplectic representations, and hence come
with a natural family of quantizations $\widehat{X}_\lambda$.
Here $\lambda$ is
a parameter of the quantization, which is of the same
nature as commutative deformations of $X$, e.g.
the central character in the case
$$
\cU(\fg)\big/ \textup{central character} =
\textup{Quantization of $T^*G/B$} \,.
$$
For example, the Hilbert
scheme of $n$ points in the plane yields the spherical
subalgebra of Cherednik's double affine Hecke algebra
of $\mathfrak{gl}(n)$ --- a structure of great depth and
importance in applications.

Using quantization in characteristic $p\gg 0$, one constructs
an action of the fundamental group of the complement of
a certain periodic locally finite arrangement of rational
hyperplanes in $H^2(X,\C)$ by autoequivalences of
$D^b_\bT(\Coh X)$. It is known in special cases
and conjectured in general that these hyperplanes
coincide with \eqref{hyper} and, moreover, one conjectures
a precise identification of the resulting action on $K_\bT(X)$ with
the monodromy of the quantum differential equation.
This can be verified for the Hilbert schemes of points and other
Nakajima varieties whose fixed loci under a torus action consists of isolated points
\cite{OP,BO}
and it is quite possible that similar arguments can be
made to work for general Nakajima varieties. There are
parallel links between the singularities of \eqref{eq1} and
representation theory of $\widehat{X}_\lambda$ for
special values of $\lambda$
in characteristic $0$, see \cite{PE}.

\subsubsection{}

An important structure which emerges from the quantization
viewpoint is an association of a $t$-structure on $D^b_\bT(\Coh X)$
to each alcove of the complement of \eqref{hyper} in
$H^2(X,\R)$. The abelian hearts of the corresponding
$t$-structures are identified with $\widehat{X}_\lambda$-modules
for the corresponding range of parameters $\lambda$. In this way,
the action of the fundamental group by derived autoequivalences
of $\Coh X$ fits into an action of the fundamental
groupoid
$$
\B= \pi_1\left(H^2(X,\C) \setminus \textup{affine root hyperplanes}\right)
$$
by derived equivalences between the categories of
$\widehat{X}_\lambda$-modules.  In particular, $\B$ acts on
the common K-theory $K_\bT(X)$ of all these categories.

The main object constructed in this paper is a
\emph{dynamical} extension of the action of
$\B$ on $K_\bT(X)$. By definition, this means that
the operators of $\B$ depend on the K\"ahler
variables $z$ and the braid relations are understood
accordingly.

To be precise, in this paper we construct
a dynamical action of $\B$ and we \emph{prove} its
relation to the quantum difference equation. The
connection with quantization
in characteristic $p\gg 0$ is not considered in this
paper, see \cite{BO}.
Similarly, a categorical lift of the dynamical action at this
point remains an open problem. It is possible that it easier
to categorify the \emph{monodromy} of the quantum
difference equation, which can be characterized in terms
of an action of an elliptic quantum group on the elliptic
cohomology of Nakajima varieties, see \cite{AO}.

\subsection{Quantum difference equations}

\subsubsection{}
The quantum difference equation is a flat
$q$-difference connection
$$
\Psi(q^\cL z) = \M_\cL(z) \Psi(z)
$$
on functions of $z$ with
values in $K_\bT(X)$. It shifts the argument by
$$
z \mapsto q^{\cL} z\,,
$$
where $\cL\in \Pic(X)$ is a line bundle on $X$ or, equivalently,
a cocharacter of the K\"ahler torus $\Pic(X) \otimes \Ct$.
See \cite{pcmilect} for an introductory exposition of their
construction and enumerative significance; these are
briefly recalled in Section \ref{s_quant_K}.

In particular, in \cite{pcmilect} it is shown that these equations commute
with the quantum Knizhnik-Zamolodchikov equations for
the $\cU_\hbar(\fgh)$-action on $K_\bT(X)$. This
commutation property will be the key ingredient in
determining the quantum difference equation.

\subsubsection{}
The arrangement \eqref{hyper} is periodic under the
action of the lattice $\Pic(X)$ and hence there is a copy of
this lattice in the fundamental groupoid. Our main
result is the identification of this lattice with the
operators of the quantum difference equation.

Concretely, this means the following formula for the
quantum difference equation. Let
$$
\nabla \subset \Pic(X)\otimes \R \setminus \{
\textup{affine root hyperplanes}\}
$$
be the unique alcove contained in minus the ample cone and
whose closure contains the origin. Let $\cL$ be an ample
line bundle and choose a path connecting $\nabla$ to the
alcove $\nabla-\cL$. Let $w_1,w_2,\dots$ be the ordered
list of affine root hyperplanes that this path crosses.

Each $w$ determines a set of affine roots that vanish on it and the corresponding rank $1$ subalgebra
$$
\fg_{w} \subset \fgh = \fg \otimes \C[t^{\pm1}]  \,.
$$
While there is no canonical root subalgebra
$\cU_\hbar(\fg_w) \subset \cU_\hbar(\fgh)$ in the quantized
loop algebra, the choice of a path as above is precisely
the additional data needed to fix such $\cU_\hbar(\fg_w)$.

Each $\cU_\hbar(\fg_w)$ is a triangular Hopf algebra and
to any such one can associate a universal element
$\repM_w(\lambda)$, $\lambda\in \fh_w$, in its
completion. It reduces to the dynamical operator
of Etingof and Varchenko when $\fg_w\cong
\mathfrak{sl}_2$.  When $\fg_w$ is a Heisenberg algebra,
which happens in the case of Hilbert schemes of points
in ADE surfaces, there is an equally explicit formula
for the element $\repM_w(\lambda)$, see Sections \ref{s_ExamplesG} and \ref{apa}.

Our main result, Theorem \ref{mainth}, says that
$$
\M_\cL = \textup{const} \, \cL \, \cdots \repM_{w_3} \repM_{w_2}
\repM_{w_1}
$$
where $\cL$ is the operator of tensor product by $\cL$ in
$K_\bT(X)$. By the basic property of the fundamental
groupoid, the result is independent of the choice of the path.

\subsubsection{}
Intertwining operators between Verma modules, which
are the main technical tool of \cite{EV}, are only available for real
roots and $\fg_w\cong\mathfrak{sl}_2$.  Outside quivers
of finite ADE type, these do not generate a large enough
dynamical Weyl group. It is therefore important to use
an abstract formula for the operator $\repM_w(\lambda)$.

Such a general formula is given by
\begin{equation}
\repM_w(\lambda)=
 \left.\textbf{m} \Big(1 \otimes S_{w}( \, \wJ^{-}_{w}(\lambda)^{-1}\, )
   \Big)\right|_{\lambda\rightarrow\lambda+\textup{shift}}
 \,,
 \label{gen_form_B}
\end{equation}
where $\wJ^{-}$ lies in a completion of
the tensor square of $\cU_\hbar(\fg_w)$
and is a fundamental solution of a qKZ-like equation known as
the ABRR equation in honor of D.~Arnaudon, E.~Buffenoir,
E.~Ragoucy, and Ph.~Roche \cite{abrr}.  One then applies the
antipode $S_w$ of $\cU_\hbar(\fg_w)$ in one of the tensor
factors and the multiplication map
$$
\textbf{m}: \cU_\hbar(\fg_w)^{\otimes 2} \to \cU_\hbar(\fg_w)
$$
to get an element in the completion of $\cU_\hbar(\fg_w)$.

One makes $\repM_w$ a function of $\lambda\in \fhh$
via the natural surjection
\begin{equation}
\fhh \to \fh_w \to 0 \,,\label{surjhh}
\end{equation}
where $\fh_w \cong \C$ is the Cartan subalgebra of $\fg_w$ and
$\fhh$ is the Cartan subalgebra of $\fgh$ that includes
$\fh$ and the infinitesimal
loop rotation  $t\frac{d}{dt}$. In particular, the operator
$\repM_w(\lambda)$ depends on $q$ via
$$
q \frac{d}{dq} \mapsto t\frac{d}{dt}
\,.
$$
The shift in \eqref{gen_form_B} includes the shift by $\hbar^\kappa$,
where
$$
2 \kappa =  \bw - C \bv
$$
 is the weight of the component with dimension
vector $\bv$ with respect to the geometric action of the quantum
loop algebra. Here $C$ is the Cartan matrix of the quiver.
The shifts by $\hbar^\kappa$ in all formulas can be traced to the
$\hbar^{\codim/4}$ prefactor in $R$-matrices, see Section
\ref{s_Omega_Nakajima}.

For quivers of finite or affine type, all root subalgebras are
either $\mathfrak{sl}_2$ or Heisenberg algebras, and the
general formula for $\repM_w(\lambda)$ may be converted
into something very explicit. We consider these examples in
Sections \ref{s_ExamplesG} and~\ref{apa}.

\subsubsection{}
The main result of this paper is a description of the quantum
difference equations that arise in the enumerative K-theory
of \emph{quasimaps} to Nakajima varieties, see \cite{pcmilect} for
an introduction. This is the natural generality in which our
methods of geometric representation theory work.

There exist both more general and more special problems
in enumerative K-theory. A very general study of K-theoretic
questions using the moduli spaces of \emph{stable maps}
was initiated many years ago by Givental. In that theory,
there exist difference equations as shown by Givental and
Tonita \cite{GV}.
The general theory lacks certain crucial self-duality
properties that are exploited in the construction of
the quantum Knizhnik-Zamolodchikov equations, see the
discussion in \cite{pcmilect}, and it
remains to be seen how much progress one can make in
the study of the difference equations of \cite{GV}.

On the other side, there exist quantum K-theory of homogeneous
spaces, initiated by Givental and Lee \cite{GL}
who discovered, in particular, its connection to the difference Toda
equations. One expects this theory to extend to symplectic
resolutions $T^*G/P$, with a connection to Macdonald theory
similar to \cite{BMO}. For $G=GL(n)$, these were studied in \cite{GL}.
In this case, $T^*G/P$ is a Nakajima variety for a linear
quiver and so is covered by our result. The relation of the
quantum dynamical Weyl group to Macdonald operators was
already explicitly present in the original work of Etingof, Tarasov,
and Varchenko.

\subsection{Other directions} 
{\rr Substantial progress has been made since the first release of this paper in 2016. The construction of stable envelope, which is an important tool of this paper, was generalized to elliptic cohomology setting in \cite{AO}. Explicit combinatorial formulas for the elliptic, K-theoretic and cohomological stable envelopes are now available for many classes of varieties \cite{SmirnovEllipticHilbert,DinkinsAffineQuivers, RimanyiShouBowVarieties}. The class of varieties for which the stable envelope exists has been extended in ~\cite{OkounkovInductive}, see also \cite{RimanyiRozansky} for a super-algebra generalization.
	
An important new feature of the elliptic stable envelope is that, in addition to the torus equivariant parameters, it depends on the K\"ahler parameters.  This makes the elliptic stable envelope a natural object in the study of the so called {\it three-dimensional mirror symmetry} which, among other things, interchanges the equivariant and the K\"ahler parameters. Three-dimensional mirror symmetry of the elliptic stable envelope has been investigated and proven for many examples of symplectic varieties, see \cite{RSVZid,RSVZj,SmirnovZhuHypertoric, RimanyiWeber, Dinkins3Dmir}.  
	
The elliptic stable envelope provides the transition matrices between
various bases of solutions of the quantum difference equations which we study in this paper, see \cite{AO}. In particular, one can use the elliptic stable envelope to describe the monodromy of these equations and to obtain integral representations for their solutions \cite{OkounkovNonabelian,AgannagicOkounkovBethe}. These results, combined with the three-dimensional mirror symmetry, lead to a new geometric descriptions of many constructions of our paper. 
As an example, the dynamical braid group generators~(\ref{gen_form_B}), playing the most fundamental role in this paper, can be identified with  K-theoretic $R$-matrices of certain subvarieties of the $3D$-mirror variety \cite{SmirnovQDE2}, see also \cite{SmirnovKononov1,SmirnovKononov2} for similar applications.
   
The $q\to 1$ limit of the quantum difference equations provides a natural description of the quantum K-theory ring of corresponding varieties.  Our results can be used to relate quantum K-theory rings to known integrable systems and give a proof of various predictions from theoretical physics \cite{PushkarZeitlinBaxter,KPZS,KorZetIInstantons}. 
	
}

\subsection{Acknowledgements}

During our work on this project, we greatly benefited
from interaction with M.~Aganagic,
R.~Bezrukavnikov, H. Dinkins, S.~Gautam, D.~Maulik,
 S.~Shakirov, V.~Toledano Laredo and A.~Oblomkov.
We are particularly grateful to Pavel Etingof for his
inspiration and guidance.

A.O.\ thanks the Simons foundation for being financially
supported as a Simons investigator.

The work of A.S. was partially supported by NSF grant DMS-2054527 and by the RSF grant 19-11-00062.

\section{Equivariant K-theory of
  Nakajima varieties and $R$-matrices  \label{sec1}}

\subsection{Stable envelopes in K-theory}

\subsubsection{}
\label{defsect}
Let $X$ be an algebraic symplectic variety and $G$ a
reductive group acting on $X$.
Since the algebraic symplectic form $\omega$
on $X$ is
unique up to a multiple, the group $G$ scales $\omega$ by a character
$\hbar$. Replacing $G$ by its double cover if necessary, we can
assume that $\hbar^{1/2}$ exists.

Let
$\bA\subset G$ be a torus in the center of $G$ and in the
kernel of $\hbar$. By definition, the K-theoretic
stable envelope is a K-theory class on the product~\cite{pcmilect}:
$$
\Stab \subset K_{G}(X \times X^\bA)\,,
$$
uniquely defined by certain support, degree, and normalization conditions.  {\rr The corresponding conditions are summarized in the Theorem \ref{themstab} below. } The stable envelope provides a wrong way map
$$
\Stab: K_{G}(X^\bA) \to  K_{G}(X)\,,
$$
which we denote by the same symbol.

\subsubsection{}

The construction of stable envelopes requires additional data,
namely the choice of:
\begin{itemize}
\item a cone $\fC \subset \Lie(\bA)$,
which divides the normal directions
to $X^\bA$ into attracting and repelling ones and determines the support
of $\Stab$,
\item a polarization $T^{1/2} \in K_G(X)$, which is a choice of a half
of the tangent bundle $TX\in K_G(X)$, that is, a solution of
 \begin{equation}
T^{1/2} + \hbar^{-1} \otimes \left(T^{1/2}\right)^\vee =
TX
\label{def_polar}
\end{equation}
in $K_G(X)$,
\item a slope $\slope \in \Pic(X)\otimes_{\Z} \Q$, which should be suitably
generic, see below.
\end{itemize}

\noindent
Of these pieces of data, the cone $\fC$ is exactly the same as
in cohomology \cite{MO}. The polarization reduces in cohomology to a
certain sign, while the slope parameter is genuinely
$K$-theoretic.

We recall from \cite{MO}, Section 2.2.7, that a Nakajima variety,
like any symplectic reduction of a cotangent bundle, has
natural polarizations. For any polarization $T^{1/2}$, there is the
opposite polarization
\begin{equation}
\Topp =\hbar^{-1} \otimes \left(T^{1/2}\right)^\vee \label{Topp} \,.
\end{equation}

\subsubsection{}
Let $\cN$ be the normal bundle to $X^\bA$ in $X$. The
$\bA$-weights $v$ appearing in $\cN$ define hyperplanes
$\{v=0\}$ in $\Lie \bA$. By definition, a cone
$$
\fC \subset \Lie \bA \setminus \bigcup_{v} \{v=0\}
$$
is one of the chambers of  the complement. We write $v>0$ if
$v$ is positive on $\fC$. A choice of $\fC$ thus
determines the decomposition
$$
\cN = \cN_+ \oplus \cN_{-}
$$
into attracting and repelling directions, with the
corresponding attracting manifold
$$
\Attr = \left\{(x,y), \lim_{a\to 0} a\cdot x = y \right\} \subset X \times X^\bA
$$
where $a\to 0$ means that $v(a)\to 0 $ for all $v>0$.

We define the full attracting set $\Attr^{\,f} \subset X\times X^{\bA}$ as the minimal closed set which contains the diagonal $X^{\bA} \times X^{\bA}$ and is invariant under taking $\Attr(\cdot)$. In other words, the components of  $\Attr^{\,f}$ are obtained from the components of the diagonal  $X^{\bA}\times X^{\bA}$ iterating taking $\Attr(\cdot)$ and the closure. The stable envelope is supported at the full attracting set:
\be \label{supcond}
\textrm{supp}(\Stab) \subset \Attr^{\,f}.
\ee

\subsubsection{}
Let $F$ be a component of $X^\bA$.
By Koszul resolution,
$$
\cO_{\Attr}\Big|_{F \times F}  =
\cO_{\diag F} \otimes \Ldm \cN_{-}^\vee \,,
$$
where the subscript in $\Ldm$ indicates an alternating
sum of exterior powers. We require
\begin{equation*}
  \label{Stab_norm_}
  \Stab \Big|_{F \times F} = \pm \, \, \textup{line bundle} \otimes \cO_{\Attr}\Big|_{F \times F}
\end{equation*}
where the sign and the line bundle are determined
by the choice of polarization.

Concretely, let
$$
T^{1/2}\big|_F =   T^{1/2}_{0}
\oplus T^{1/2}_{\ne 0}
$$
be the splitting of the polarization into trivial and nontrivial
$\bA$-characters. We have
$$
\cN_{-} \ominus  T^{1/2}_{\ne 0} = \hbar^{-1}
\left( T^{1/2}_{>0} \right)^\vee \ominus T^{1/2}_{>0} \,,
$$
and therefore the determinant of this virtual vector bundle is
a square (recall that we replace $G$ by its double cover if the
character $\hbar$ is not a square). We set
\begin{equation}
  \label{Stab_norm}
  \Stab \Big|_{F \times F} = (-1)^{\rk T^{1/2}_{>0}}
\left( \frac{\det \cN_{-}}{\det T^{1/2}_{\ne 0}}\right)^{1/2}
\otimes \cO_{\Attr}\Big|_{F \times F} \,.
\end{equation}

\subsubsection{}
The key property of stable envelopes are degree bounds
satisfied by $\Stab\big|_{F_2 \times F_1}$, where $F_1$ and $F_2$
are two different components of $X^\bA$. Note that
because of the support condition, this restriction vanishes
unless $F_2 < F_1$ in the partial ordering defined by
the closures of attracting manifolds, that is, by
$$
\exists x, \quad \lim_{a \to 0} a^{\pm 1} x  \in F_\pm
\Rightarrow F_+ > F_-  \,.
$$
Recall that in cohomology the degree bound reads
\begin{equation}
\deg_\bA \Stab\Big|_{F_2 \times F_1} <
\deg_\bA \Stab\Big|_{F_2 \times F_2} \,,\label{degH}
\end{equation}
where $\deg_\bA$ for an element of
$$
\Hd_G(X^\bA,\Q) \cong \Hd_{G/\bA}(X^\bA,\Q) \otimes \Q[\Lie \bA]
$$
is its degree in the variables $\Lie \bA$.

\subsubsection{}
Now in K-theory the degree $\deg_\bA f$ of a Laurent polynomial
$$
f = \sum_{\mu\in \bA^\wedge} f_\mu \, a^\mu \in \Z[\bA] = K_\bA(\pt)
$$
is its Newton polygon
$$
\deg_\bA f  = \textup{Convex hull} \,
\left(\{ \mu, f_\mu \ne 0\} \right) \subset \bA^\wedge \otimes_\Z \Q
\,,
$$
with the natural partial ordering on polygons defined by inclusion.

Such a definition has a caveat, in that the degree of an invertible
function $a^\mu$ should really be zero, and so the Newton polygons
should really be considered up to translation by the lattice
$\bA^\wedge$. If we want to compare two Newton polygons by
inclusion, a possibility of inclusion after a shift appears, and this
is where the slope parameter $\slope$ comes in.

The K-theoretic analog of \eqref{degH} is the following condition
\be  \label{degcond}
\deg_\bA \Stab_\slope\Big|_{F_2 \times F_1} \otimes \slope\Big|_{F_1}
\,\, \subset \,\,
\deg_\bA \Stab_\slope \Big|_{F_2 \times F_2} \otimes \slope\Big|_{F_2} \,,\label{degK}
\ee
where the weight of a fractional line bundle
$\slope \in \Pic(X) \otimes_\Z \Q$ is a fractional weight, that is,
an element of $\bA^\wedge \otimes_\Z \Q$. Note that
\eqref{degK} is independent of the $\bA$-linearization of $\slope$.
The dependence of the stable envelope $\Stab_\slope$ on the
slope $\slope$ is indicated for emphasis in the LHS of \eqref{degK}.
The degree of $\Stab \Big|_{F_2 \times F_2}$ is given by
\eqref{Stab_norm} and is independent of $\slope$.

\begin{Remark}\label{Rem1}
  Observe that for a sufficiently generic $\slope$ the inclusion in
\eqref{degK} is necessarily strict, as the inclusion between
fractional shifts of integral polytopes.
\end{Remark}

\subsubsection{}
{\rr
Let us summarize the above definitions in the following result:

\begin{Theorem} \label{themstab}
Let $X$ be a Nakajima variety, then for an arbitrary choice of  chamber 
$\fC \subset \Lie(\bA)$, polarization $T^{1/2} \subset K_{\bT}(X)$ and generic slope $s \in \Pic(X) \otimes_{\matZ} \matQ$ there exists a unique K-theory class $\Stab_{\fC,T^{1/2},s} \in K_{\bT}(X\times X^{\bA})$ which satisfies 

1) support condition (\ref{supcond})

2) degree condition (\ref{degcond}) 

3) normalization condition (\ref{Stab_norm})
\end{Theorem}

\begin{Remark}
Stronger results were obtained since the first release of this paper. 
For the Nakajima varieties a version of above theorem for the {\it elliptic stable envelope} was proven in~\cite{AO}. The existence and uniqueness
of the elliptic stable envelope then implies Theorem \ref{themstab} in the K-theoretic degeneration of elliptic cohomology, see Section 4.5 in \cite{AO}. The existence of the stable envelope under weaker conditions on $X$ was also proved in \cite{OkounkovInductive}. In particular, the existence of polarization of $X$ is replaced in \cite{OkounkovInductive}  by a weaker condition of existence of {\it attracting line bundles}. With these new tools, many constructions of this paper translate to a setting more general than Nakajima varieties.
\end{Remark}

Uniqueness of stable envelopes implies the following
transformation law under duality on $X \times X^\bA$
\begin{equation}
\left(\Stab_{\fC,\, T^{1/2}, \, \slope}\right)^\vee  =
{\rr \hbar^{-\codim(X^{\bA})/4}}
\Stab_{\fC,\,\Topp\,,
	- \slope} \,. \label{Stab_dual}
\end{equation}
Here $\Topp$ is
the opposite polarization \eqref{Topp}.
}

\subsubsection{}
To keep track of the weights of the line bundles $\slope$
restricted to components of the fixed locus, it is convenient
to introduce a locally constant map (a form of moment map)
\begin{equation}
\bmu: X^\bA \to H_2(X,\Z) \otimes \bA^\wedge \,,\label{defmu}
\end{equation}
defined up to an overall translation, such that
$$
\bmu(F_1) - \bmu(F_2) = [C] \otimes v
$$
if there is an irreducible $\bA$ invariant curve $C$ joining $F_1$ and
$F_2$ with tangent weight $v$ at $F_1$. For any $\slope$, we then
have
$$
\textup{weight} \, \slope\big|_{F_1} - \textup{weight} \,
\slope\big|_{F_2} = (\slope,C) \, v \,.
$$
By construction
\begin{equation}
\Stab\Big|_{F_2 \times F_1} \ne 0 \Rightarrow
\bmu(F_1) - \bmu(F_2) \in
H_2(X,\Z)_\textup{eff} \otimes
\bA^\wedge_{>0} \label{supp_bmu} \,,
\end{equation}
where $\bA^\wedge_{>0}$ is the cone of weights positive on $\fC$.


\subsection{Slope $R$-matrices \label{rtotsec}}

\subsubsection{}
Following the sign conventions set in Section 3.1.3 of \cite{MO}, we
define the transposition
$$
K(X \times Y) \owns \cE \mapsto \cE^\tau \in K(Y \times X)
$$
as a permutation of factors together with a sign
$(-1)^{(\dim X - \dim Y)/2}$.

The following is an analog of Theorem 4.4.1 in  \cite{MO}

\begin{Proposition}\label{p_Stab_tau}
  \begin{equation}
    \label{StabStab}
    \Stab^\tau_{-\fC,\, \Topp,\, -\slope}
\circ \Stab_{\fC,\, T^{1/2},\,\slope}  = 1 \,.
  \end{equation}
\end{Proposition}

\noindent
Here we do not distinguish between the structure sheaf of the
diagonal and the identity operator by which it acts on the K-theory.

\begin{proof}
Since the support of stable envelopes is the same as in cohomology,
the convolution \eqref{StabStab} is an integral K-theory class on
$X^\bA \times X^\bA$.

Denoting by $\cS$ and $\cS'$ the two
stable envelopes in \eqref{StabStab}, we have
\begin{equation}
\left(\cS'^\tau \circ \cS\right)_{F_3 \times F_1}  =
\sum_{F_1 \ge F_2 \ge F_3}  (-1)^{\frac{\codim F_3}{2}}
\frac{\cS'\big|_{F_2 \times F_3} \otimes \cS \big|_{F_2 \times F_1}}
{\Ldm \cN_{F_2}^\vee}
\label{sumSS}
\end{equation}
by equivariant localization and the support condition, where
$F_i$ are components of the fixed point locus $X^\bA$.

Since the convolution \eqref{sumSS} is integral, its Newton polygon
may be estimated directly from \eqref{sumSS}.  We denote by
$$
\mu = \langle\bmu(F_3)-\bmu(F_1),\slope\rangle \in \bA^\wedge \otimes \Q
$$
the difference of weights of $\slope$ at $F_3$ and $F_1$. We have
$\mu \notin  \bA^\wedge$ for generic $\slope$ unless $F_3 = F_1$
because an ample line bundle will pair nonzero with
$\bmu(F_3)-\bmu(F_1)$.

The
degree bound \eqref{degK} implies each term is $O(|a|^\mu)$ as $a\in
\bA$ goes to infinity in any direction. Since this number is
fractional for $F_3 \ne F_1$ while the asymptotics are integral, it
follows that terms with $F_1 \ne F_3$ in \eqref{sumSS} vanish.

The remaining terms with $F_1 = F_2 =F_3$ are easily seen to give
the identity operator.
\end{proof}

\subsubsection{\label{rootsec}}
In the same way, stable envelopes may be defined for
real slopes $\slope \in H^2(X,\R)$. They depend on the slope
in a locally constant way and change as $\slope$ crosses
certain rational hyperplanes
\begin{equation}
w\overset{\textup{def}}{=} \{s \in H^2(X,\R) : (\slope,\alpha) + n =0 \} \,,
\label{wallL}
\end{equation}
which we will call \emph{walls}. Here
\be
\label{roots}
\bal = (\alpha,n) \in H_2(X,\Z) \oplus \Z
\ee
is an integral affine function on $H^2(X)$, which we call
an \emph{affine root} of $X$. The connected components of
the complements to the walls in $H^2(X,\R)$ are called
\emph{alcoves}.

Below we will see that $\pm \alpha$ is an effective curve class
for any affine root $\bal$. If $n\ne 0$, we set
$$
\bal' = \tfrac{1}{n} \alpha \in H_2(X,\Q) \,.
$$
This depends only on the wall and not on the particular
normalization of its equation.

\subsubsection{ \label{rootRsec}}

Let us consider two slopes $\slope$ and $\slope'$ separated by a single wall $w$.
To examine the change in stable envelopes across the wall, we define the \textit{wall $R$-matrix}:
\be
\label{rootR}
R^{\fC}_{w} = \Stab^{-1}_{\fC,\, T^{1/2},\,\slope'}
\circ \Stab_{\fC,\, T^{1/2},\,\slope} .
\ee
To distinguish $R^{\fC}_{w}$ from its inverse, we assume
$$
\langle \slope' - \slope,\alpha\rangle > 0 \,.
$$
for the positive root $\alpha$ defining the corresponding wall. If we cross the wall
from $\slope$  to $\slope'$ we say that it is \emph{crossed in the positive direction.}

\begin{Theorem} \label{udepth}
We have
$$
R^{\fC}_{w}\Big|_{F_3\times F_1} = 0
$$
unless
\begin{equation}
\bmu(F_1) - \bmu(F_3) =  \bal' \otimes \mu
\label{bmu21}
\end{equation}
where $\bal' \in H_2(X,\Q)_\textup{eff}$ and
$\mu$ is an integral weight of $\bA$ positive on $\fC$. In
this case
$$
\deg_\bA R^{\fC}_{w}\Big|_{F_3\times F_1}  = \mu \,.
$$
If $n=0$ the condition \eqref{bmu21} means $\mu=0$ and
that $\bmu(F_1) - \bmu(F_3)$ is proportional to $\alpha$.
\end{Theorem}

\noindent
As a corollary of the proof, we will see that
$$
R^{\fC}_{w}\Big|_{F_1 \times F_1} = 1 \,.
$$

\begin{proof}
As in the proof of Proposition \ref{p_Stab_tau}, we see that
$R^{\fC}_{w}$ is an integral K-theory class and we compute its
restriction to $F_3 \times F_1$ by localization as in \eqref{sumSS}.

Consider the localization term corresponding to a component
$F_2$ of $X^\bA$. The slope-dependent part of its degree is
\begin{align}
  \langle\bmu(F_2)&-\bmu(F_1),\slope\rangle + \langle\bmu(F_3)-\bmu(F_2),\slope'\rangle
\hfill
\notag  \\
&=  \langle\bmu(F_3)-\bmu(F_1),\slope'\rangle +
 \langle\bmu(F_2)-\bmu(F_1),\slope - \slope'\rangle  \label{degF2_1}\\
&=  \langle\bmu(F_3)-\bmu(F_1),\slope\rangle +
 \langle\bmu(F_2)-\bmu(F_3),\slope - \slope'\rangle  \label{degF2_2}\,.
\end{align}
Since the ample cone is open, we may assume that
$\pm(\slope-\slope')$ is
ample. If $\slope> \slope'$,  the second summand in \eqref{degF2_1} is a
negative weight, while the second summand in \eqref{degF2_2} is
a positive weight. If $\slope<\slope'$, these conclusions
are reversed. But in either case,
$$
R^{\fC}_{w}\Big|_{F_3\times F_1} = O(|a|^\mu)
$$
as $a\to 0$ or $a\to \infty$, where
$$
\mu = \langle \bmu(F_3)-\bmu(F_1),x \rangle
$$
for $x\in w$ and $a\to 0$ as before
means that $v(a)\to 0$ for every positive weight $v$.
Since this is a Laurent polynomial in $a$, this means vanishing
unless $\mu$ is an integral weight and $R^{\fC}_{w}\Big|_{F_3\times F_1}$
is a monomial.

For generic $\slope$ on the
hyperplane \eqref{wallL} the weight $\mu$ is integral
only if
$$
\bmu(F_1)-\bmu(F_3) \in \Q \, \alpha \otimes \bA^\wedge \,.
$$
From \eqref{supp_bmu} and since $(x,\bal')=-1$
for $n\ne 0$
by construction, we conclude \eqref{bmu21}. If $n=0$ we have
$(x,\alpha)=0$ and hence $\mu=0$.
\end{proof}

\subsection{Root subalgebras}
\label{sec:root-subalgebras}

\subsubsection{ \label{secNak}}
We recall that Nakajima varieties depend on a quiver with a
vertex set $I$, two dimension vectors $\bv,\bw \in \N^I$, and
a stability parameter $\theta\in \R^I$. The complex
deformation parameter $\zeta\in \C^I$, which is the
value of the complex moment map in symplectic reduction,
will always be set to zero in this paper. We fix $\theta$ and
denote
$$
\cM(\bw) = \bigsqcup_\bv \cM_\theta(\bv,\bw) \,.
$$
We take the canonical
polarization from Section 2.2.7 in \cite{MO} as polarization
$T^{1/2}$ of Nakajima varieties.

\subsubsection{ \label{rmatdef}}
Let $W$ be a framing space defining a Nakajima variety with dimension $\bw$. Let us consider its arbitrary
decomposition into a direct sum of subspaces $W=W' \oplus W''$ with dimensions $\bw'$ and $\bw''$. Assume that a torus
$\bA=\C^{\times}$ acts on $W$ scaling  $W'$ with a character $a'$ and $W''$ with a character $a''$. In this situation we
say that $\bA$ \emph{splits the framing} $\bw=a' \bw' +a'' \bw''$.

This action induces an action of $\bA$ on the Nakajima variety $\cM(\bw)$. The basic property of the Nakajima varieties is that the set of the $\bA$ fixed points is the product of Nakajima varieties for the same quiver but different framings:
$$
\cM(\bw)^{\bA} = \cM(\bw') \times \cM(\bw''),
$$
such that after localization:
$$
K_G(\cM(\bw)^{\bA})= K_G(\cM(\bw')) \otimes K_G(\cM(\bw'')).
$$
{\rr One checks that the $\bA$ characters appearing in the normal bundle to $\cM(\bw)^{\bA}$ are of the form $u^{\pm 1}$, where $u=a'/a''$.} Thus, we only have two chambers 
which correspond to $u \to 0$ and $u \to \infty$. We denote them by $+$ and
$-$ respectively. For a slope $\slope$ these give the stable maps:
$$
\Stab_{\pm,\slope}: K_G(\cM(\bw)) \otimes K_G(\cM(\bw')) \to
K_G(\cM(\bw+\bw'))
$$
for any $G$ that commutes with $\bA$. To examine the change of the stable map under the change of the chamber
we introduce the following \textit{total $R$-matrix with slope  $\slope$}:
\begin{align}
\Rtot^{\slope}(u) &= \Stab_{-,\slope}^{-1} \circ \Stab_{+,\,\slope} \,,\label{RL}
\end{align}
One checks that it depends only on the ratio $u$.
Just like the cohomological $R$-matrices, $\Rtot^{\slope}(u)$ acts in a localization
 $K_G(\cM(\bw)) \otimes K_G(\cM(\bw'))$.  However,
the coefficients of
the $u\to 0$ or $u\to \infty$ expansion of $\Rtot^{\slope}(u)$ are operators in
nonlocalized K-theory. The variable $u$ is traditionally  called
the \emph{spectral parameter}. {\rr The operators $\Rtot^{\slope}(u)$ satisfy the Yang-Baxter
 and unitarity equations for arbitrary slope $\slope$, see (9.2.20) and (9.2.22) in \cite{pcmilect} for explicit form of these equations.}

\subsubsection{\label{Rfactor}}
{\rr Let $F_1\neq F_2$ be two components of $X^{\bA}$. Let us consider the degree condition (\ref{degcond}) as the slope $s$ approaches infinity along the ample or anti-ample direction in $\Pic(X)\otimes_{\matZ} \matQ$. By (\ref{degcond}) in this limit the $\bA$ characters appearing in the restriction $\left.\Stab\right|_{F_{1}\times F_{2}}$ approach infinity in $\Lie(\bA)$. Thus, in the suitable topology of power series we have
$$
\lim_{s \to \pm \infty} \left.\Stab_{\pm,s}\right|_{F_{1}\times F_{2}} =0.
$$
Therefore, we can characterize $\Stab_{\pm, \infty}$ as the classes with restrictions (\ref{Stab_norm}) near diagonal and vanishing at non-diagonal terms of 
$X^{\bA} \times X^{\bA}$. Explicitly:
$$
\Stab_{\pm, +\infty}=\Stab_{\pm, -\infty} =  i_{*}\Big( (-1)^{\rk T^{1/2}_{>0}} \Big(  \dfrac{\det(\cN_{\mp})}{\det(T^{1/2}_{\neq 0})}\Big)^{1/2} \, S^{\bullet}(\cN^{\vee}_{\pm})\Big) \in K_{\bT}(X\times X^{\bA})_{loc}
$$ 
where $i$ is the inclusion of diagonal $X^{\bA}\times X^{\bA} \rightarrow X\times X^{\bA}$.
}

We can include the given slope $\slope$ into an doubly infinite sequence
\begin{equation}
\dots  \slope_{-2}, \slope_{-1}, \slope_0 = \slope, \slope_1, \slope_2,
\dots\label{seqL}
\end{equation}
such that
$$
\slope_i \to \pm \infty \,, \quad i \to \pm \infty \,,
$$
where $\slope_i \to +\infty$ means that $\slope_i$ goes to infinity inside
the ample cone of $X$. We can assume that $\slope_i$ and $\slope_{i+1}$
are separated by exactly one wall $w_i$ and that the
sequence $\{\slope_i\}$ crosses each wall once. We can write the following obvious identity:
$$
\begin{array}{l}
\Stab_{+,\,\slope}=\\
\\\Stab_{+,\,+\infty} \cdots \Stab_{+,\,\slope_2} \Stab_{+,\,\slope_2}^{-1} \Stab_{+,\,\slope_1} \Stab_{+,\,\slope_1}^{-1} \Stab_{+,\,\slope}=\\
\\
\Stab_{+,\,+\infty}\cdots R^{+}_{w_2} R^{+}_{w_1} R^{+}_{{\rr w_0}}.
\end{array}
$$
Similarly for the negative chamber:
$$
\begin{array}{l}
\Stab_{-,\,\slope}=\\
\\\Stab_{-,\,-\infty} \cdots \Stab_{-,\,\slope_{\rr-2}} \Stab_{-,\,\slope_{\rr-2}}^{-1} \Stab_{-,\,\slope_{\rr-1}} \Stab_{-,\,\slope_{\rr-1}}^{-1} \Stab_{-,\,\slope}=\\
\\
\Stab_{-,\,-\infty}\cdots (R^{-}_{w_{\rr-3}})^{-1} (R^{-}_{w_{\rr -2}})^{-1} (R^{-}_{w_{\rr -1}})^{-1}.
\end{array}
$$
In the last case we cross the walls in the negative direction and  by our convention from Section \ref{rootRsec}
the corresponding contribution is given by the inverse of the wall $R$-matrix.

{\rr From definitions we find 
$$
\left.\Stab_{+,\infty}(\gamma)\right|_{F}= (-1)^{\rk T^{1/2}_{>0}}
\left( \frac{\det \cN_{-}}{\det T^{1/2}_{\ne 0}}\right)^{1/2}\, \Lambda^{\bullet}_{-} (\cN^{\vee}_{-}) \otimes \gamma
$$ 
$$
\left.\Stab_{-,-\infty}(\gamma)\right|_{F}= (-1)^{\rk T^{1/2}_{<0}}
\left( \frac{\det \cN_{+}}{\det T^{1/2}_{\ne 0}}\right)^{1/2}\, \Lambda^{\bullet}_{-} (\cN^{\vee}_{+}) \otimes \gamma
$$
for any $\gamma$ supported at a component $F\subset X^{\bA}$. The restriction of $\Stab_{\pm,\mp\infty}(\gamma)$ to other components of $X^{\bA}$ vanish and thus $R_{\infty}:=\Stab_{-,-\infty}^{-1} \circ
\Stab_{+,\infty}$ is diagonal in the basis of fixed components with the following matrix elements:
\be
\label{infr}
R_\infty\big|_{F\times F} = (-1)^{\codim(F)/2}
\frac{\prod_{v<0} (v^{1/2} - v^{-1/2})}
{\prod_{v>0} (v^{1/2} - v^{-1/2})}
\ee
where $v$ are the Chern roots of $\cN_{F}$.} All together this gives the following factorization of the total $R$-matrix:
\begin{equation}
\Rtot^{\slope}(u) \stackrel{def}{=}\Stab_{-,\,\slope}^{-1} \Stab_{+,\,\slope}= \overleftarrow{\prod_{\rr i < 0}} \, R^{-}_{w_{i}}  \,\,
R_\infty  \,\, \overleftarrow{\prod_{\rr i \geq 0}} R^{+}_{w_i}
\,,\label{Rfac}
\end{equation}
{\rr where $\overleftarrow{\prod}_{i}$ stands for the product of matrices ordered from right to left as the index $i$ increases}.
The factorization \eqref{Rfac} converges in the topology of the formal power series around $u=\infty$ as will be explained in the next section.

\noindent
Similarly we can factorize the total $R$-matrix into infinite product near $u=0$. We consider:
{\rr 
$$
\begin{array}{l}
\Stab_{+,\,\slope}=\\
\\\Stab_{+,\,-\infty} \cdots \Stab_{ +,\,\slope_{-2}} \Stab_{+,\,\slope_{-2}}^{-1} \Stab_{+,\,\slope_{-1}} \Stab_{+,\,\slope_{-1}}^{-1} \Stab_{+,\,\slope}=\\
\\
\Stab_{+,\,-\infty}\cdots (R^{+}_{w_{-3}})^{-1} (R^{+}_{w_{-2}})^{-1} (R^{+}_{w_{-1}})^{-1}.
\end{array}
$$
}
and
{\rr 
$$
\begin{array}{l}
\Stab_{-,\,\slope}=\\
\\\Stab_{-,\,+\infty} \cdots \Stab_{-,\,\slope_2} \Stab_{-,\,\slope_2}^{-1} \Stab_{-,\,\slope_1} \Stab_{-,\,\slope_1}^{-1} \Stab_{-,\,\slope}=\\
\\
\Stab_{-,\,+\infty}\cdots R^{-}_{w_2} R^{-}_{w_1} R^{-}_{w_0}.
\end{array}
$$}
This gives another factorization:
{\rr \begin{equation}
\Rtot^{\slope}(u) \stackrel{def}{=}\Stab_{-,\,\slope}^{-1} \Stab_{+,\,\slope}=\overrightarrow{\prod_{i\geq 0}} \, (R^{-}_{w_i})^{-1}  \,\,
R_\infty  \,\, \overrightarrow{\prod_{i < 0}} (R^{+}_{w_i})^{-1}
\,,\label{RfacInf}
\end{equation}
with the same $R_\infty$ given by (\ref{infr}). }
We will call these formulas Koroshkin-Tolstoy (KT)
factorizations of total $R$-matrices. {\rr An explicit example of KT factorization for the simplest quiver variety $X=T^{*}\mathbb{P}^1$ can be found in Section \ref{ktexamtp}.} For the quivers of finite type this formula reproduces the factorization of quantum $R$-matrices considered in \cite{KhorTol}.

\subsubsection{\label{triangularity}}
Recall that the partial ordering on the components of the fixed point set coincides with ``ample partial ordering''. If $\theta \in \Pic(X)$ is a choice of ample line bundle, and
$\sigma \in \fC$ is a character of $\bA$ then:
$$
 F_{2} \unlhd F_{1} \ \ \ \Leftrightarrow \ \ \  \langle\theta_{F_1},\sigma \rangle \leq  \langle \theta_{F_2},\sigma \rangle
$$
The choice of the stability parameter $\theta \in \Z^{I}$ for a Nakajima variety
defines a certain ample line bundle.  If the fixed components have the form $F=\cM(\bv,\bw)\times\cM(\bv',\bw')$ then the function defining the ordering takes the following explicit form:
$$
 \langle\theta_{F},\sigma \rangle = \langle \bv, \theta \rangle \sigma + \langle \bv', \theta \rangle \sigma'
$$
All the operators $A$ acting in $K$-theory which we consider in this paper will preserve the total weight, i.e., $A=\bigoplus\limits_{\alpha} \, A_{\alpha}$ with:
$$
A_{\alpha} : K_{\KG}(F_1 ) \longrightarrow K_{\KG}(F_2)
$$
and $F_1=\cM(\bv,\bw)\times\cM(\bv',\bw')$, $F_2=\cM(\bv+\alpha,\bw)\times\cM(\bv'-\alpha,\bw')$. Therefore the difference of ordering function takes the form:
\be
\label{weightdiff}
\langle\theta_{F_2},\sigma \rangle - \langle\theta_{F_1},\sigma \rangle =  \langle \alpha, \theta \rangle(\sigma -\sigma')
\ee
In the present text we will always assume that the fixed components are ordered using the positive chamber $\sigma -\sigma'>0$.  Thus the sign of the difference (\ref{weightdiff}) is given by a sign of $ \langle \alpha, \theta \rangle$.

We will use the following terminology: an operator $A=\bigoplus\limits_{\alpha} A_{\alpha}$ with $A_{\alpha}$ as above
is \textsl{upper-triangular} if $\langle \alpha, \theta \rangle>0$ and \textit{lower-triangular} if $\langle \alpha, \theta \rangle<0$ for all $\alpha\neq 0$. We say that $A$ is \textit{strictly upper-triangular} or \textit{strictly lower-triangular} if in addition
$A_{0}=1$.   For example, the wall $R$-matrices $R^{+}_{w}$ and $R^{-}_{w}$ are strictly upper and strictly lower triangular respectively.  In particular, the Khoroshkin-Tolstoy factorization (\ref{Rfac}) gives a LU decomposition of the total $R$-matrix.

\subsubsection{\label{trian}}
 Let $\cL_w $ be a line bundle on the wall $w$.  The wall $R$-matrices $R^{\pm}_{w}$ are triangular with monomial in spectral parameter $u$ matrix elements:
\begin{equation}
R^{\pm}_{w}\big|_{F_2 \times F_1}  =
\begin{cases}
1\,, \quad & F_1 = F_2 \,, \\
\propto u^{\langle \bmu(F_2)-\bmu(F_1),\cL_w \rangle}\,, \quad
& F_1
\gtrless F_2 \,, \\
0\,, \quad & \textup{otherwise} \,.
\end{cases}\label{degRbal}
\end{equation}
The condition \eqref{supp_bmu} means
$$
R^{\pm}_{w}\to 1\,, \quad  w \to \pm \infty\,,
$$
in the topology of formal power series.

\subsubsection{}

{\rr From (\ref{infr}) we obtain}
\begin{equation}
\lim\limits_{u \rightarrow 0 } R_\infty = \hbar^{-\Omega}\, \ \ \ \lim\limits_{u \rightarrow \infty } R_\infty = \hbar^{\Omega}
\label{Rinfasy}
\end{equation}
where $\Omega$ is the codimension function:
\be \label{codimfun}
\Omega(\gamma)=\dfrac{\codim(F)}{4} \, \gamma
\ee
for a class $\gamma$ supported on the fixed set component $F\in \cM(\bw)^{\bA}$.

\subsubsection{}\label{s_Omega_Nakajima}

For Nakajima varieties, the codimension function
\eqref{codimfun} has the following description.
For a torus $\bA$ splitting the framing $\bw=a' \bw'+a'' \bw''$,
every component $F \in \cM(\bv,\bw)^{\bA}$ is of the
form
\be \label{component}
F = \cM(\bv',\bw') \times \cM(\bv'',\bw'')
\ee
for some dimension vectors $\bv',\bv''$. We have, see e.g.\
Section 2.4.2 in \cite{MO},
\be
\label{codfun}
\Omega=\dfrac{\codim F}{4}  = \dfrac{1}{2} (\bw',\bv'') + \dfrac{1}{2}(\bw'',\bv') - \dfrac{1}{2} (\bv', C \bv'')\,,
\ee
where $C$ is the Cartan matrix of the quiver, see e.g.
Section 2.2.5 of \cite{MO}. The map $\bmu$ has the form:
\begin{equation}
\bmu(F) = \bv' \otimes 1 \label{bmuF}
\end{equation}
where $1\in \bA^\wedge$ is the weight of $u$, see e.g.\ Section 3.2.8
in \cite{MO}.

\subsubsection{\label{walyb}}
Let $w_k$  be a wall labeled by $k\in \matZ$ as in Section \ref{Rfactor} and let $\cL_{w_k}$ be a fractional line bundle at $w_k$. We denote
$$
\widetilde{\Rtot^{\slope}(u)}_k  = U_k^{-1}\,  \Rtot^{\slope}(u) \, U_k
$$
where  $U_k$ is a block diagonal matrix with the block corresponding to component \eqref{component} given by:
$$
U_k\Big|_{F} = (a')^{\langle \bv',\cL_{w_k}\rangle}
(a'')^{\langle \bv'',\cL_{w_k}\rangle}
$$
Similarly we define:
$$
\widetilde{R^{\pm}_{w_i,k}}=U_{k}^{-1} \, {R^{\pm}_{w_i}}\,U_{k}.
$$
If  $F_1=\cM(\bv'_1,\bw'_1) \times \cM(\bv''_1,\bw''_1)$ and 
$F_2=\cM(\bv'_2,\bw'_2) \times \cM(\bv''_2,\bw''_2)$ are two components of $\cM(\bv,\bw)^{\bA}$, then we have
$$
\widetilde{R^{\pm}_{w_i,k}}\Big|_{F_1\times F_2}=R^{\pm}_{w_i}\Big|_{F_1\times F_2}\, \dfrac{(a')^{\langle \bv'_2, \cL_{w_k}\rangle} (a'')^{\langle \bv''_2, \cL_{w_k}\rangle} }{(a')^{\langle \bv'_1, \cL_{w_k}\rangle} (a'')^{\langle \bv''_1, \cL_{w_k}\rangle}}.
$$ 
Noting that $\bv'_1+\bv''_1=\bv'_2+\bv''_2$ we can rewrite this as
$$
\widetilde{R^{\pm}_{w_i,k}}\Big|_{F_1\times F_2}=R^{\pm}_{w_i}\Big|_{F_1\times F_2}\, u^{\langle \bv'_2-\bv'_1,\cL_{w_k}\rangle }
$$
where $u=a'/a''$. From \eqref{degRbal} and \eqref{bmuF} we then find
\begin{equation}
\widetilde{R^{\pm}_{w_i,k}}\big|_{F_1\times F_2}  =
\begin{cases}
1\,, \quad & F_1 = F_2 \,, \\
\propto u^{\langle \bv_2'-\bv_1'
	, \cL_{w_k}-\cL_{w_i} \rangle}\,, \quad
& F_1
\gtrless F_2 \,, \\
0\,, \quad & \textup{otherwise} \,.
\end{cases}\label{degRbalU}
\end{equation}
By construction of the sequence \eqref{seqL} we have:
$$
\langle \alpha, \cL_{w_k}-\cL_{w_i} \rangle 
\gtrless 0 \,, \quad i \lessgtr k \,, \ \ \ \langle \alpha, \cL_{w_k}-\cL_{w_i} \rangle 
= 0 \,, \quad i = k.
$$
when  $\alpha$ is effective. From block-triangularity of   $\widetilde{R^{+}_{w_i,k}}$  we see that in the limit ${u\to \infty}$ all non-diagonal matrix elements vanish for $k<i$. The matrix elements of  $\widetilde{R^{+}_{w_i,k}}$ do not depend on $u$ if $i=k$. In summary we can write it as
\be \label{lim1}
U_k \Big( \lim_{u\to \infty} \widetilde{R^{+}_{w_i,k}}\Big) U_{k}^{-1}= \left\{\begin{array}{ll}
\textrm{DNE}, & k>i\\
R^{+}_{w_i}, & k=i\\
1, & k<i
\end{array}\right.
\ee
where $\textrm{DNE}$ means that the corresponding limit may be undefined in this case. The matrices $\widetilde{R^{-}_{w_i,k}}$ are lower-triangular and similar consideration gives:
\be \label{lim2}
U_k \Big( \lim_{u\to \infty} \widetilde{R^{-}_{w_i,k}}\Big) U_{k}^{-1}= \left\{\begin{array}{ll}
1, & k>i\\
R^{-}_{w_i}, & k=i\\
\textrm{DNE}, & k<i
\end{array}\right.
\ee
Conjugating KT factorization around $u=\infty$ (\ref{Rfac}) by $U_k$  we obtain: 
$$
\widetilde{\Rtot^{\slope}(u)}_k= \overleftarrow{\prod_{\rr i < 0}} \, \widetilde{R^{-}_{w_{i},k}}  \,\,
R_\infty  \,\, \overleftarrow{\prod_{\rr i \geq 0}} \widetilde{R^{+}_{w_i,k}}
$$
From \eqref{lim1}, \eqref{lim2} and \eqref{Rinfasy} we see that $k=0$ and $k=-1$ are the only two choices for which the limit $u\to \infty$ of all factors in this product is well defined.
For these values the limit equals:
\begin{equation}
 U_{k} \, \left(\lim_{u\to \infty} \widetilde{\Rtot^{\slope}(u)_k} \right) \, U^{-1}_{k} =
 \begin{cases}
   \hbar^{\Omega}\,R^{+}_{w_{0}}  \,, \quad & k=0 \,,\\
  R^{-}_{w_{-1}}\,\hbar^{\Omega} \,, \quad & k=-1 \,.
 \end{cases} \label{limRU}
\end{equation}
Arguing similarly for KT-factorization near $u=0$ \eqref{RfacInf} we find:
\begin{equation}
 U_k \, \left(\lim_{u\to 0} \widetilde{\Rtot^{\slope}(u)_k} \right) \, U^{-1}_{k} =
 \begin{cases}
 (R^{-}_{w_0})^{-1}   \hbar^{-\Omega }  \,  \,, \quad & k=0 \,,\\
 \hbar^{-\Omega} \, (R^{+}_{w_{-1}})^{-1} \,, \quad & k=-1 \,.
 \end{cases} \label{limRU2}
\end{equation}
In summary, we see that the wall $R$-matrices for $w_0,w_{-1}$ which are the walls immediately before and after the slope $s$ in \eqref{seqL} can be obtained as limits of $\widetilde{\Rtot^{\slope}(u)}$. As $\widetilde{\Rtot^{\slope}(u)}$ solves the quantum Yang-Baxter  equation for any $s$, the same is true for their limits. We thus obtain:
\begin{Theorem}
\label{thmR}
The wall $R$-matrices multiplied by $\hbar^{\Omega}$:
$$
\hbar^{\Omega} R^{\pm}_{w}, \ \   R^{\pm}_{w} \hbar^{\Omega}
$$
satisfy the quantum Yang-Baxter equation for any wall $w$.
\end{Theorem}
\noindent
In what follows we denote
\be \label{defwall}
\Rwal^{\pm}_{w} = \hbar^{\Omega} R^{\pm}_{w}.
\ee

\subsubsection{\label{unitsectio}}

{\bb

Let us show that K-theoretic $R$-matrices (\ref{RL}) are unitary for an arbitrary slope~$s$. The derivation follows the same steps as in cohomology and we refer to Section 4.5 of \cite{MO} for more details. 
	
Let  $\bA= \mathbb{C}^{\times}$	and let us consider the action of $\bA$ on a Nakajima variety $\cM(\bv,\bw)$ corresponding to the  splitting of the framing $\bw= a \bw' + \bw''$.  Let $F=\cM(\bv',\bw')\times \cM(\bv'',\bw'')$ be a  component of $\cM(\bv,\bw)^{\bA}$. 

Similarly, let us consider the $\bA$-action on $\cM(\bv,\bw)$ corresponding the splitting $\bw=   a \bw''+\bw'$. We denote by 
$F_{21}=\cM(\bv'',\bw'')\times \cM(\bv',\bw')$ the $\bA$-fixed component corresponding to $F$ under this action.

In the second case the $\bA$-action on $\cM(\bv,\bw)$ is the opposite of the $\bA$-action in the first case. This means that the original action of $\bA$ is precomposed with the automorphism
$$
\phi: \bA \to \bA, \ \ \  \phi: a\mapsto a^{-1}.
$$
We note that the correspondence $\Stab_{\fC,T^{1/2},s}$ is exactly the correspondence $\Stab_{-\fC,T^{1/2},s}$ for the opposite action.
Note also that the $\bT$-characters of the normal bundles in these cases are related by:
$$
\left.\cN_{-}\right|_{F}= \phi^{*}(\left.\cN_{+}\right|_{F_{21}}). 
$$
By uniqueness of the stable envelopes we obtain:
\be \label{unitstab}
\left.\Stab_{\fC,T^{1/2},s}\right|_{F\times F'} = \left.\Big(\left.\Stab_{-\fC,T^{1/2},s}\right|_{F_{21}\times F'_{21}}\Big)\right|_{a\to a^{-1}}.
\ee

For an operator $A \in \textrm{End}(K_{\bT}(\cM(\bw'))\otimes K_{\bT}(\cM(\bw'')))$ we denote by 
$
A_{21} \in \textrm{End}(K_{\bT}(\cM(\bw''))\otimes K_{\bT}(\cM(\bw')))$ the operator corresponding to the permuted matrix elements:
$$
(A_{21})_{F,F'}=A_{F_{21},F'_{21}}.
$$

With these notations from (\ref{unitstab}) we obtain the following result:
\begin{Proposition}
K-theoretic $R$-matrices (\ref{RL}) of Nakajima varieties satisfy the unitary condition:	
\be \label{unitarity}
\Rtot^{\slope}(u)=\Rtot^{\slope}(u^{-1})^{-1}_{21}.
\ee

\begin{Remark}
	For an explicit example of identity (\ref{unitstab}) we refer to (\ref{restr1}) and (\ref{restr2}) describing the stable envelopes for $X=T^{*}\mathbb{P}^1$. In notations of this example $X^{\bA}=\{p_1,p_2\}$ with $(p_{1})_{21}=p_2$, $(p_{2})_{21}=p_1$.
	We also encourage the reader to check that (\ref{unitarity}) holds for matrix (\ref{RtotTP}). 
\end{Remark}

\end{Proposition}	

}
\section{Construction of quantum groups  \label{frt}}
As we explain in Section \ref{secNak} the equivariant K-theory of a Nakajima variety provides a set of vector spaces
$K_{\KG}( \cM(\bw) )$ labeled by a dimension vector $\bw\in \Z^{|I|}$. For any splitting of the framing $\bw=u \bw' + \bw''$ our construction gives an $R$-matrix which acts in $ K_{\KG}( \cM(\bw') )\otimes K_{\KG}( \cM(\bw'') )$ and satisfies the quantum Yang-Baxter equation. This is a well known set up for the Faddeev-Reshetikhin-Takhtajan  formalism \cite{FRT}. Using these data the FRT construction provides a triangular Hopf algebra $\Uq(\fgh_Q)$ acting in
$K_{\KG}( \cM(\bw) )$ for all $\bw$.

Similarly, applying the FRT construction to the wall $R$-matrices $\Rwal^{\pm}_{w}$ one constructs a set of triangular Hopf algebras $\Uq(\fg_w)$ which are, in fact, subalgebras of $\Uq(\fgh_Q)$.

The aim of this section is to review the FRT method and to explain the interaction between Hopf structures
of different wall subalgebras $\Uq(\fg_w)$.

\subsection{Quiver algebra $\Uq(\fgh_Q)$}
\subsubsection{}
For a splitting $\bw=u_1 \bw_1+\cdots+u_n \bw_n$ and a slope $\slope \subset H^2( \cM (\bw), \R )$ the construction of Section \ref{rmatdef}  provides a set of $R$-matrices
$$
\Rtot^{\slope}_{V_i,V_j}(u_i/u_j) \subset \textrm{End}\Big( V_{1}\otimes \cdot\cdot\cdot \otimes V_{n} \Big) \otimes \C[u_1^{\pm 1},...,u_n^{\pm 1}], \ \
$$
with $V_{k}=K_{G}(\cM (\bw_k))$  satisfying the Yang-Baxter equation. We denote
$$
V_{i}(u)\stackrel{def}{=}V_{i} \otimes \C[u^{\pm 1}]
$$
and more generally
$$
V_{i_1}(u_1)\otimes\cdot\cdot\cdot \otimes V_{i_n}(u_n) \stackrel{def}{=}V_{i_1}\otimes\cdot\cdot\cdot \otimes V_{i_n} \otimes \C[u_1^{\pm 1},...,u_n^{\pm 1}]
$$
\subsubsection{}
We have a set of vector spaces $\frak{V}$ such that for any pair $V_i,V_j \in \frak{V}$ we have an $R$-matrix $\Rtot^{\slope}_{V_{i},V_{j}}( u_i/u_j )$.

First, we note that this set is closed with respect to the tensor product. The $R$-matrix for the tensor products has the following form:
\be
\label{facproc}
\Rtot^{\slope}_{ \bigotimes\limits_{i \in I}^{\leftarrow} \, V_{i}(u_i), \bigotimes\limits_{j \in J}^{\leftarrow} \, V_{i}(u_i) }= \prod\limits_{i \in I}^{\rightarrow} \prod\limits_{j \in J}^{\leftarrow}\, \Rtot^{\slope}_{V_{i},V_{j}}( u_i/u_j ).
\ee
Second, following \cite{Resh} we can assume that this set contains dual vector spaces  $V^{*}_{i}$ with $R$-matrices defined by the following rules:
$$
\begin{array}{l}
\Rtot^{\slope}_{V^{*}_{1},V_{2}}= ((\Rtot^{\slope}_{V_{1},V_{2}})^{-1})^{*_1}\\
\\
\Rtot^{\slope}_{V_{1},V_{2}^{*}}= ((\Rtot^{\slope}_{V_{1},V_{2}})^{-1})^{*_2}\\
\\
\Rtot^{\slope}_{V^{*}_{1},V_{2}^{*}}= (\Rtot^{\slope}_{V_{1},V_{2}})^{*_{12}}
\end{array}
$$
where $*_k$ means transpose with respect to the $k$-th factor. One checks that the $R$-matrices defined this way
satisfy the quantum Yang-Baxter equation in the tensor product of any three spaces from the set $\frak{V}$.

\subsubsection{}
In the FRT formalism the quantum algebra $\Uq^{\slope}(\fgh_{Q})$ is defined as the subalgebra
$$
\Uq^{\slope}(\fgh_{Q}) \subset \prod\limits_{V\in \frak{V}}\, \textrm{End}( V )
$$
generated by matrix elements of
\be
\label{Ugdef}
\Rtot^{\slope}_{V,V_{0}}(u) \in \textrm{End}(V) \otimes  \textrm{End}(V_0)
\ee
in the ``auxiliary space'' $V_0$ for all choices of $V_0 \in \frak{V}$.


{\bb An element of $\Uq^{\slope}(\fgh_{Q})$ is fixed by a choice of the following data: an auxiliary space $V_{0}$, a finite rank operator
$$
m(a_0) \in \textrm{End}(V_{0})(a_0)
$$
an integer $l\in \mathbb{Z}$ and $i\in\{+,-\}$. The element of $\Uq^{\slope}(\fgh_{Q})$ corresponding to this data acts in a representation $V(a)$ as the following operator:
\be \label{uqgqelems}
\rho^{i}_{V_0,m,l}=\textrm{Coeff}^{\,i}_{a_0^{l}}\Big( \textrm{tr}_{V_0}( 1 \otimes m(a_0)\, \Rtot^{\slope}_{V,V_0}(u) )\Big) \in \textrm{End}(V(a))
\ee
where $\Rtot^{\slope}_{V,V_0}(u)$ is the $R$-matrix acting in $V(a)\otimes V_0(a_0)$ with $u=a/a_0$, and $\textrm{Coeff}^{+}_{a_0^{l}}$, $\textrm{Coeff}^{-}_{a_0^{l}}$ denote the coefficient of $a_0^l$ in the Laurent series expansions near $a_0=0$ or $a_0=\infty$ respectively. Since $m(a_0)$ is of finite rank the trace over the auxiliary space $V_0$  is defined even if it is infinite-dimensional.

The algebra  $\Uq^{\slope}(\fgh_{Q})$ is generated by all $\rho^{i}_{V_0,m,l}$.

 }

{\rr
	\begin{Proposition}
		The algebras $\Uq^{\slope}(\fgh_{Q})$ are isomorphic for all $\slope$.
	\end{Proposition}
	
	\begin{proof}
		Let  $s$ and $s'$ be two slopes separated by a single wall $w$. Enough to show that $\Uq^{\slope}(\fgh_{Q})$ and $\Uq^{\slope'}(\fgh_{Q})$ are isomorphic. From Khoroshkin-Tolstoy factorization we find that
		$$
		\Rtot^{\slope}(u)=(R^{-}_{w}(u))^{-1} \Rtot^{\slope'}(u) R^{+}_{w}(u)=(R^{+}_{w}(u^{-1})_{21})^{-1}\, \Rtot^{\slope'}(u) \,R^{+}_{w}(u)
		$$
		where the last equality is by (\ref{rtransp}) and definition (\ref{defwall}).
		
		It is known that the wall $R$-matrices $R^{+}_{w}$ satisfy the cocycle condition, see Corollary \ref{corrcocy} in Section \ref{cocysec}. Thus, the $R$-matrices $\Rtot^{\slope}(u)$ and $\Rtot^{\slope'}(u)$ provide isomorphic algebras by Theorem 2.3.4 in \cite{Majid}. 
	\end{proof}
}

\begin{Proposition}
	\be
	\label{rtransp}
	\Rwal^{\mp}_{w}=\left.(\Rwal^{\pm}_{w})_{21}\right|_{u=u^{-1}}, \ \ \
	R^{\mp}_{w}=\left.(R^{\pm}_{w})_{21}\right|_{u=u^{-1}} 
	\ee
\end{Proposition}
\begin{proof}
	{\bb The first equality follows from (\ref{unitarity}) together with limits (\ref{limRU}) and (\ref{limRU2}).
		
		In notations of Section \ref{unitsectio}, the codimensions of the torus fixed component $F=\cM(\bv',\bw')\times \cM(\bv'',\bw'')$ and  $F_{21}= \cM(\bv'',\bw'')\times \cM(\bv',\bw')$ in $\cM(\bv,\bw)$ are equal.
		Therefore $\Omega=\Omega_{21}$ which gives the second equality.  
	}
\end{proof}

As the algebras $\Uq^{\slope}(\fgh_{Q})$ are isomorphic for all $s$ we will denote them by~$\Uq(\fgh_{Q})$.


\subsection{Wall subalgebra $\Uq(\fg_{w}) \subset \Uq(\fgh_{Q})$ \label{subalgebra}}


{
{
\subsubsection{}
Let us define the wall algebra:
\be \label{algebrawall}
\Uq(\fg_{w})\subset \prod\limits_{V\in \frak{V}}\, \textrm{End}( V )
\ee
as an algebra generated by the matrix elements of $(\Rwal^{+}_{w})_{V,V_0}$  and of $(\Rwal^{-}_{w})_{V,V_0}^{-1}$
in the auxiliary space $V_0$ for all $V_{0} \in \frak{V}$.

{\bb 

For a choice of an auxiliary space $V_{0}$ and a finite rank operator
$
m \in \textrm{End}(V_{0}) 
$
we have an element of $\Uq(\fg_{w})$ acting in a representation $V(a)$ as the following operator:
\be \label{elemwallal}
\rho^{+}_{V_0,m}= \textrm{tr}_{V_0}( 1 \otimes m\, (\left.\Rwal^{+}_{w})_{V,V_0}\right|_{a_0=1} ) \in \textrm{End}(V).
\ee
Note that by Theorem \ref{udepth} the matrix elements of $(\Rwal^{\pm}_{w})_{V,V_0}$ are monomials in $u=a/a_0$. Thus we do not need to consider all coefficients in the  Laurent series expansion as in (\ref{uqgqelems}). 

Algebra (\ref{algebrawall}) is generated by all such $\rho^{+}_{V_0,m}$ and also $\rho^{-}_{V_0,m}$ which are given by (\ref{elemwallal}) with $\Rwal^{+}_{w}$ substituted by $(\Rwal^{-}_{w})^{-1}$.

}

\subsubsection{\label{generatorssection}}

{\bb

Next we show that $\Uq(\fg_{w})$ is a subalgebra of $\Uq(\fgh_{Q})$. For this we show that all matrix elements (\ref{elemwallal}) appear as matrix elements (\ref{uqgqelems}) for some choices of $l$ and $m$.

Let $w$ be a wall and $s$ be a generic slope obtained by a shift $s=w-\epsilon$ for an infinitesimal  ample $\epsilon$. Let $\Rtot^{s}_{V,V_0}(u)$ with $u=a/a_0$ be the $R$-matrix with slope $s$ acting in $V(a)\otimes V_0(a_0)$.  Let $U$ be the diagonal matrix acting in $V_0(a_0)$ by
	$
	\left.U\right|_{ \cM(\bv_0,\bw_0)}=a_0^{\langle\bv_0, \cL_w \rangle}.
	$
	The action of $U$ by conjugation gives the  decomposition:
	$$
	\textrm{End}(V_0) = \bigoplus_{l} \textrm{End}_l(V_0)
	$$
	with 
	$
	\textrm{End}_l(V_0)=\{  m \in \textrm{End}(V_0): U m U^{-1}   = a_0^l m\}.
	$
	Since $\cL_w$ is a fractional line bundle, the weights $l$ appearing in this decomposition are rational. We denote by $\textrm{End}^{(w)}(V_0)$  the subspace spanned by integral  weights:
	\be \label{integrmat}
	\textrm{End}^{(w)}(V_0) = \bigoplus_{l \in \mathbb{Z}} \textrm{End}_l(V_0).
	\ee
	Let $m\in \textrm{End}_l(V_0)$ for some $l\in \mathbb{Z}$ which is constant in $a_0$.  Let us consider an element (\ref{uqgqelems}) corresponding to this data:
	\be \label{ueleme}
	\rho^{+}_{V_0,m,-l}=\textrm{Coeff}^{+}_{a_0^{-l}}(\textrm{tr}_{V_0}( 1\otimes m\, \Rtot^{s}_{V,V_0}(u)  ))
	\ee
     Since  $U m U^{-1}   = a_0^l m$ we have
	$$
	\rho^{+}_{V_0,m,-l}=\textrm{Const}(\textrm{tr}_{V_0}( 1\otimes m\, (1\otimes U)^{-1}\, \Rtot^{s}_{V,V_0}(u) \, (1\otimes U) ))
	$$
	where $\textrm{Const}$ denotes the constant term in the series expansion at $a_0=0$.
	By (\ref{limRU}), at $a_0=0$ we have the following expansion:
	$$
	(1\otimes U)^{-1}  \Rtot^{s}_{V,V_0}(u) (1\otimes U)  =  (1\otimes U)^{-1} (\textsf{\Rwal}^{+}_{w})_{V,V_0} (1\otimes U) + \dots
	$$
	where $\dots$ denote the higher order terms vanishing at $a_0=0$. Note that by~(\ref{degRbal})
	the first term $ (1\otimes U)^{-1} (\textsf{\Rwal}^{+}_{w})_{V,V_0} (1\otimes U)$ does not depend on $a_0$. 
	Thus, since $m$ does not depend on $a_0$ we have:
	$$
	\rho^{+}_{V_0,m,-l}=\textrm{tr}_{V_0}(1\otimes m \, (1\otimes U)^{-1}  (\textsf{\Rwal}^{+}_{w})_{V,V_0} (1\otimes U) ) =\textrm{tr}_{V_0}(1\otimes m \,   \left.(\textsf{\Rwal}^{+}_{w})_{V,V_0}\right|_{a_0=1} )
	$$
We find that $\rho^{+}_{V_0,m,-l}$ is of the form (\ref{elemwallal}) and therefore represents an element from $\Uq(\fg_{w})$. By Theorem \ref{udepth} the matrix elements of $\textsf{\Rwal}^{+}_{w}$ are non-trivial only for $m\in \textrm{End}_{l}(V_0)$ with integral $l$. Thus, all generators $\rho^{+}_{V_0,m}$ (\ref{elemwallal}) appear as (\ref{uqgqelems})  with $m$ from (\ref{integrmat}).

Applying the same logic to the power series expansion near $a_0=\infty$  we also find that all generators of $\Uq(\fg_{w})$ corresponding to $\rho^{-}_{V_0,m}$ appear in the same way.

By definition, these elements generate $\Uq(\fg_{w})$ and therefore $\Uq(\fg_{w})$ is a subalgebra of $\Uq^{s}(\fgh_{Q})$. Finally, since $\Uq^{s}(\fgh_{Q})$ are isomorphic for all $s$ we obtain the following result:
	
}

\begin{Proposition}
$\Uq(\fg_{w})$ is a subalgebra of $\Uq(\fgh_{Q})$ for every wall~$w$.
\end{Proposition}

{\bb 
\begin{Remark}	
	The quantum group elements (\ref{uqgqelems}) are defined as coefficient $a_0^{l}$ of $R$-matrices depending on
$u=a/a_0$. This means that $\rho^{+}_{V_0,m,l}$ acts as a monomials $a^{-l}$ in any representation $V(a)$.
\end{Remark}}


\begin{Remark}
{\bb In the following text we often understand the $R$-matrices and other operators as universal elements of the corresponding quantum groups or their completions. 
In particular, such universal elements do not depend on the evaluation parameters $u$, which are the parameters of the representations, not of the quantum groups. For instance, the unitarity relation (\ref{unitarity}) for the universal $R$-matrix takes the form:
$$
\Rtot^{s}=(\Rtot^{s}_{21})^{-1}
$$
where $21$ denotes the permutation of factors $\Rtot^{s}\in \Uq(\fgh_{Q}) \otimes \Uq(\fgh_{Q})$.
Similarly, the wall $R$-matrices give the universal elements $\textsf{\Rwal}^{\pm}_{w} \in \Uq(\fg_{w})\otimes \Uq(\fg_{w})$. 
Relations (\ref{rtransp}) are understood as 
\be \label{Rtranuniv}
(\textsf{\Rwal}^{\pm}_{w})_{21}=\textsf{\Rwal}^{\mp}_{w}, \ \ \ (R^{\pm}_{w})_{21}=R^{\mp}_{w}
\ee
for these elements.

As an example we refer to explicit formulas for universal R-matrices of ${{\mathscr{U}_{{\rr \sqrt{\hbar}}}(\widehat{\frak{gl}}_2)}}$
given by (\ref{sl2w}) and (\ref{sl2low}), which are related by permutation of tensor factors (\ref{Rtranuniv}).  The generators of this algebra act in a representation $\mathbb{C}^2(a)$ 
as monomials $E_w \sim a^{w}$ and $F_w \sim a^{-w}$. Evaluation of these universal elements in a representation $\mathbb{C}^2(u_1)\otimes \mathbb{C}^2(u_2)$ gives  (\ref{walRTP})
which depends on the spectral parameter $u=u_1/u_2$. The unitarity relation should be understood as (\ref{rtransp}) for these matrices. 
}
\end{Remark}

\subsection{Hopf structures}

\subsubsection{}
The algebra $\Uq(\fgh_{Q})$  carries Hopf structures labeled by the slope $\slope$.
The set $\frak{V}$ is closed with respect to tensor product. It induces the natural projection:
$$
\prod\limits_{V\in \frak{V}}\,\textrm{End}(V) \,\rightarrow \, \prod\limits_{V_1, V_2\in \frak{V}}\,\textrm{End}(V_1\otimes V_2)
$$
which restricts to a coproduct map on matrix elements of $\Rtot^{\slope}(u)$:
$$
\Delta_s :\, \Uq(\fgh_{Q}) \rightarrow \Uq(\fgh_{Q})\hat{\otimes} \Uq(\fgh_{Q})
$$
Note that this map depends on KT factorization of $R$-matrix and thus on the slope $\slope$.

The set $\frak{V}$ is closed with respect to taking dual $*$ and thus we have an antipode map:
$$
S_s: \,  \Uq(\fgh_{Q}) \rightarrow \Uq(\fgh_{Q})
$$
which is the restriction of:
$$
\textrm{End}(V)\stackrel{*}{\longrightarrow}\textrm{End}(V^*)
$$
The set $\frak{V}$ contains the trivial representation $\C$ which, similarly, induces a counit map:
$$
\epsilon_s: \Uq(\fgh_{Q}) \rightarrow \C
$$
The main result of FRT procedure is that $(\Delta_s, S_s, \epsilon_s)$ provides
$\Uq(\fgh_{Q})$ with a Hopf algebra structure for arbitrary
slope $\slope$. The algebra $\Uq(\fgh_{Q})$ becomes a triangular Hopf algebra with the triangular structure $\Rtot^{\slope}(u)$.

\subsubsection{}
The same procedure applied to $\Rwal^{+}_{w}$ in place of $\Rtot^{\slope}(u)$ defines a structure of triangular Hopf algebra ($\Delta_w$, $S_w$, $\epsilon_w$)  on  $\Uq(\fg_w)$. It should be clear from definitions that ($\Delta_w$, $S_w$, $\epsilon_w$) does not necessarily coincide with  restriction of ($\Delta_\slope$, $S_{\slope}$, $\epsilon_\slope$) from the ambient algebra  $\Uq(\fgh_Q)$. The next proposition explains the relation between these Hopf structures.

\begin{Proposition} \label{trstruprop}
Assume that the Khoroshkin-Tolstoy factorization for a total $R$-matrix with slope
$\slope$ starts with some wall $w$, i.e.  has the form:
$$
\Rtot^{\slope}(u)=\cdots \, R^{+}_{w_{1}} \,R^{+}_{w}
$$	
\label{Hopfcoin}
then the Hopf structure ($\Delta_w$, $S_w$, $\epsilon_w$) on $\Uq(\fg_w)$ coincides with the restriction of
($\Delta_\slope$, $S_{\slope}$, $\epsilon_\slope$) from the ambient algebra $\Uq(\fgh_Q)$.
\end{Proposition}

\begin{proof}
	
{\bb Enough to check this statement for coproducts.
Let $V_1 $ and $V_2$ be two representations of $\Uq(\fgh_Q)$. We need to show that for any element $x\in \Uq(\fg_w)$, the identity $\Delta_\slope(x) =\Delta_w(x)$ holds in $\textrm{End}(V_1\otimes V_2)$.

Assume that $x$ is a generator of $\Uq(\fg_w)$ as in Section \ref{generatorssection} corresponding to an auxiliary space $V_0$ and a matrix element $m\in \textrm{End}_{l}(V_0)$. 

By definition $\Delta_\slope(x)$ and $\Delta_w(x)$ act in
$V_1\otimes V_2$ by $\rho^{i}_{V_0,m,-l}$ and $\rho^{i}_{V_0,m}$ defined by formulas (\ref{uqgqelems}) and (\ref{elemwallal}) for $V=V_1\otimes V_2$. But in Section \ref{generatorssection} we proved the equality $\rho^{i}_{V_0,m,-l}=\rho^{i}_{V_0,m}$ for an arbitrary $V$. The Proposition follows since such elements generate  $\Uq(\fg_w)$.

 }

\end{proof}

\begin{Corollary}
 \label{coprdotrel}
If $\slope$ and $w$ are as in the previous proposition then  for $x\in \Uq(\fg_w)$ we have:
\be
\label{crel}
\Rtot^{\slope}(u)\, \Delta_{\slope}(x) \Rtot^{\slope}(u)^{-1} = \textsf{R}^{+}_{w} \Delta_{w} (x) (\textsf{R}^{+}_{w})^{-1} = (\textsf{R}^{-}_{w})^{-1}  \Delta_{w} (x)  \textsf{R}^{-}_{w}
\ee
with $\Rwal^{\pm}_{w}$ as in Theorem \ref{thmR}.
\end{Corollary}
\begin{proof}
In any triangular Hopf algebra we have  $\Rtot^{\slope}(u) \Delta_{s}(x) \Rtot^{\slope}(u)^{-1}= \Delta^{op}_{s}(x)$.
But, for $x\in \Uq(\fg_w)$ we have $\Delta^{op}_{s}(x)=\Delta^{op}_{w}(x)= \Rwal^{+}_w \Delta_{w}(x) (\Rwal^{+}_w)^{-1}$.
This proves the first equality. Applying (\ref{Rtranuniv}) we arrive at the second equality.

\end{proof}

\subsubsection{}
Let $\slope$ and $\slope'$ be two slopes and let $\Gamma$ be a path in $H^{2}(X,{\mathbb{R}})$ connecting them.
This path intersects finitely many walls in some order $I_{\Gamma}=\{ w_1,w_2,...,w_n \}$. We define operators:
$$
T^{+} = \prod\limits_{w \in I_{\Gamma}}^{\leftarrow}\, R^{+}_{w}, \ \ \  T^{-} = \prod\limits_{w \in I_{\Gamma}}^{\leftarrow}\, R^{-}_{w}
$$
Then, from Khoroshkin-Tolstoy factorization we obtain:
$$
 \Rtot^{\slope'}(u) T^{+} =T^{-} \Rtot^{\slope}(u)
$$
which implies that coproducts at different slopes are related by:
\be
\label{conjcop}
T^{+} \Delta_{\slope}= \Delta_{\slope'} T^{+}, \ \ \ T^{-} \Delta^{op}_{\slope}= \Delta^{op}_{\slope'}\, T^{-}.
\ee

\subsubsection{\label{infcopsec}}
As a slope $\slope$ approaches infinity (in the ample cone) we obtain a special Hopf structure
with the coproduct which we denote by $\Delta_{\infty}$. The corresponding wall subalgebra
$\Uq(\fg_{\infty})$ is generated by the matrix elements of (\ref{infr}). 
This infinite slope R-matrix is diagonal in the basis of fixed components 
 with matrix elements given by operators of multiplication by tautological
bundles in the equivariant $K$-theory (\ref{infr}). In particular, these operators are elements of $\Uq(\fg_{\infty})$. Moreover, the line bundles $\cL\in \textrm{Pic}(X)$
are group-like: 
\be
\label{infcop}
\Delta_{\infty}(\cL)=\cL \otimes \cL.
\ee

\subsubsection{}
Let $\kappa=(\kappa_1,\kappa_2)$ where each
$\kappa_{i}\in (\frac12\Z)^I$ is
a function on the vertices of the quiver with values in $\frac12\Z$.
Define an operator $\hbar^{\kappa}$
acting in $K_{\KG}(\cM(\bw))$ by multiplication by
$
\hbar^{\langle\kappa_1,\bv\rangle+\langle\kappa_2,\bw\rangle}
$
on the component $\cM(\bv,\bw)$ (recall that the square root $\hbar^{1/2}$ 
exists in the equivariant $K$-theory, see Section \ref{defsect}.). 
As we discussed above, the operators of multiplications by tautological
bundles, and in particular the operators of multiplication by their dimensions 
are elements of $\Uq(\fgh_Q)$. Thus $\hbar^{\kappa}\in \Uq(\fgh_Q)$. 
These elements enjoy the following properties:
\be
\label{antipcart}
\Delta_{w}(\hbar^{\kappa})=\hbar^{\kappa} \otimes \hbar^{\kappa}, \ \ S_{w}(\hbar^{\kappa})=\hbar^{-\kappa}
\ee
Recall that the codimension function $\Omega$ is quadratic in $\bw$, $\bv$ which gives:
\be
\label{antiom}
S_w\otimes S_w(\Omega) = \Omega
\ee
Finally, in any triangular Hopf algebra we have
\be
\label{ssr}
S_w\otimes S_w (\Rwal^{+}_{w}) = \Rwal^{+}_{w}
\ee
and thus from (\ref{antiom}) we conclude:
\be
\label{Rwalanti}
S_w\otimes S_w(R^{+}_{w})= \hbar^{-\Omega}\, R^{+}_{w}\,\hbar^{\Omega}
\ee

\section{Quantum K-theory of Nakajima varieties}\label{s_quant_K}
In this section we recall the main facts about the commuting difference equations
which govern the quasimap count for Nakajima varieties. We refer the reader to
\cite{pcmilect} for a detailed exposition. 
\subsection{Stable quasimaps to Nakajima varieties}
\subsubsection{}
Let us consider a quiver with set of vertices $I$ and $m_{i j}$ arrows from a vertex
$i \in I$ to a vertex $j \in I$. Let $n=|I|$ be the number of vertices.

Recall that a Nakajima variety $\cM(\bv,\bw)$ with dimension vectors
$\bv,\bw \in \mathbb{N}^{n}$ is defined as the following symplectic reduction:
\be \label{git}
\cM(\bv,\bw)=T^{*} M /\!\!/\!\!/\!\!/_{\!\theta} \textsf{G}=\mu^{-1}(0)/\!\!/_{\!\theta}\textsf{G}
\ee
where $M$ is the representation of the quiver
$$
M=\bigoplus\limits_{i,j \in I} \, \textrm{Hom}(V_i,V_j) \otimes Q_{i j} \oplus \bigoplus\limits_{i \in I}\, \textrm{Hom}(W_i,V_i)
$$
by vector spaces $V_i$ of dimensions $\bv_i$ and framing spaces $W_i$ of dimensions $\bw_i$. We denote by $Q_{i j}$ the linear vector space of dimension $m_{i j}$ (the multiplicity space). The representation $M$ is equipped with an obvious action of $\textsf{G}=\prod\limits_{i \in I}\, GL(V_i)$
and 
$$
\mu : T^{*}M \rightarrow \textrm{Lie}(\textsf{G})^{*}
$$
stands for the corresponding moment map. Finally, $\theta \in \Z^n$ denotes the character of $G$ 
$$
\theta : \, (g_i)^{n}_{i=1} \to \prod\limits_{i=1}^{n} \det(g_i)^{\theta_{i}}
$$
which defines a stability parameter for GIT quotient (\ref{git}).

The Nakajima varieties come together with a natural action of a group $\textsf{Aut}$
whose action preserves the symplectic form. Let $G=\bA \times \C^{\times}_{\hbar}$ where
$\bA$ is a maximal torus of $\textsf{Aut}$ and $\C^{\times}_{\hbar}$ is one-dimensional torus scaling the cotangent direction in  (\ref{git}) with a character $\hbar^{-1}$.

\subsubsection{}
The general theory of quasimaps to GIT quotients was developed in \cite{quasimaps}. 
Here we briefly recall this construction specialized to the case of Nakajima quiver varieties, see also Section 4.3 in \cite{pcmilect}.  


A quasimap  
$$
f: C \dashrightarrow X
$$
with a domain $C \simeq \mathbb{P}^{1}$ to a Nakajima variety $X=\cM(\bv,\bw)$
is defined by the following data: 

\begin{itemize}

\item A collection of vector bundles $\qmV_i$, $i \in I$ on $C$ with ranks $\bv_i$.

\item A collection of trivial vector bundles $\qmQ_{i j}$ and ${\qmW}_i$, $i,j \in I$ on  $C$ with ranks $m_{i j}$ and $\bw_i$ respectively

\item  A section 
$$
f\in H^{0}\Big(C, \qmM \oplus \qmM^{*} \otimes \hbar^{-1} \Big)
$$
satisfying the the moment map condition $\mu=0$, where
$$
\qmM=\bigoplus\limits_{i,j \in I} \, \Hom(\qmV_i,\qmV_j) \otimes {\qmQ}_{i j} \oplus \bigoplus\limits_{i \in I}\, \Hom(\qmW_i,\qmV_i).
$$
and $\hbar^{-1}$ stands for a trivial line bundle on $C$ with $G$-equivariant weight $\hbar^{-1}$.
\end{itemize}
The degree of a quasimap is defined as $d=(\textrm{deg}(\qmV_i))_{i=1}^{n} \in \Z^{n}$.

\subsubsection{}
Let $p\in C$ be a point in the domain of a quasimap $f$ and fix a local trivialization of
${\qmQ}_{i,j}$ and $\qmW_i$ at $p$. The value $f(p)$ defines an $\textsf{G}$-orbit in 
$\mu^{-1}(0)$.  This orbit does not necessarily consist of semistable points in $\mu^{-1}(0)$
and thus it only defines an \emph{evaluation map} into a quotient stack:
$$
\textrm{ev}_p:\, f \mapsto f(p) \in \mu^{-1}(0)/\textsf{G}.
$$  
The quotient stack contains the Nakajima variety as an open subset
$$
X= \mu^{-1}(0)_{stable}/\textsf{G}   \subset \mu^{-1}(0)/\textsf{G}.
$$  
A quasimap $f$ is called stable if $f(p) \in X$ for all but finitely many points $p\in C$.
The finite set of points for which $f(p) \notin X$ is called \emph{singularities of the quasimap}.

The moduli space $\qm^{d}(X)$ parameterizes the degree $d$ stable quasimaps up to isomorphism
which is required to be identity on the curve $C$, the multiplicity ${\qmQ}_{i,j}$ and the framing bundles $\qmW_i$  \cite{quasimaps}:
$$
\qm^{d}(X)=\{ \textrm{degree $d$ stable quasimaps to} \ \ X  \}/\cong
$$
This means that moving a point on this moduli space results in varying the bundles $\qmV_i$ and the section
$f$, while the curve $C$, bundles $\qmW_i$ and ${\qmQ}_{ij}$ remain fixed.

Let $\qm^{d}(X)_{\textrm{nonsing} {\kern 3pt  p}}\subset \qm^{d}(X)$ be the open subset of the moduli space corresponding to the stable quasimaps nonsingular at a point $p$.   
By definition this open subset is equipped with the evaluation morphism:
\be
\qm^{d}(X)_{\textrm{nonsing} {\kern 3pt  p}} \stackrel{\ev_p}{\longrightarrow} X.
\ee

The moduli space of \emph{relative quasimaps} $ \qm^{d}(X)_{\textrm{relative} {\kern 3pt  p}}$ is a compactification of the map $\ev_p$ meaning that it fits into the following  commutative diagram:
\[
\xymatrix{
 & \qm^{d}(X)_{\textrm{relative} {\kern 3pt  p}} \ar[dr]^{\tilde{\ev}_p} &  \\
\qm^{d}(X)_{\textrm{nonsing} {\kern 3pt  p}} \ar@{^{(}->}[ur] \ar[rr]^{\ev_p} &  & X
}
\]
with \emph{proper} evaluation map $\tilde{\ev}_p$.  The construction of the moduli space of relative
quasimaps $\qm^{d}(X)_{\textrm{relative} {\kern 3pt  p}}$ is explained in Section 6 of \cite{pcmilect}. It follows similar constructions of relative moduli spaces
in Gromow-Witten theory \cite{GWquas1,GWquas2} and Donaldson-Thomas theory \cite{DTquas}.

\subsection{Difference equations \label{geomqde}}

\subsubsection{}
As explained in \cite{pcmilect} the moduli spaces defined in the previous sections carry natural virtual structure sheaves $\vss$.
Using these virtual sheaves one constructs different enumerative invariants of $X$.
For example, one of the main objects in quantum K-theory is the \textit{capping operator}
which is defined as follows: let us consider the moduli space $\qm^{d}_{\textrm{relative} p_1 \atop \textrm{nonsing} p_2}(X)$ of quasimaps with relative conditions at $p_1 \in C$ and nonsingular at $p_2 \in C$ (we will assume that
$p_1=0$ and $p_2=\infty$ in $C=\mathbb{P}^1$). These two marked points define the evaluation map:
\be \label{sqev}
\ev=\tilde{\ev}_{p_1}\times \ev_{p_2}: \qm^{d}_{\textrm{relative} p_1 \atop \textrm{nonsing} p_2}(X)
 \longrightarrow X \times X  \ee

This moduli space is equipped with an action of $\KG\times \C^{\times}_q$ where the action of $\KG$ comes from its action on $X$ and $\C^{\times}_q$ scales the local coordinate of $C$
at the point $p_1$ with character $q$. Note that this action preserves $p_1$ and $p_2$. The capping operator is defined as the $\KG\times \C^{\times}_q$ equivariant
push-forward:
\be
\label{capp}
\capping=\sum\limits_{d\in \Z^{n}} \, z^d  \ev_{*}\Big( \qm^{d}_{\textrm{relative} p_1 \atop \textrm{nonsing} p_2}(X), \vss \Big) \in K_{\KG\times \C^{\times}_q}(X)^{\otimes 2}_{\textrm{localized}} \otimes \Q[[z]]
\ee
The map (\ref{sqev}) is not proper, as we already mentioned in the previous section. However, it becomes proper on the subset of fixed points $ \qm^{d}_{\textrm{relative} p_1 \atop \textrm{nonsing} p_2}(X)^{\KG\times \C^{\times}_q}$, see \cite{pcmilect}. Thus 
the pushforward (\ref{capp}) is well defined in the localized $K$-theory.

The degrees of the quasimaps are counted with weight $z^d= z_1^{d_1}\cdots z_n^{d_n}$.
The parameters $z_i$ are referred to as K{\"a}hler parameters. 
\subsubsection{\label{tcl1}}
Assume that we fixed some basis in $K_{\KG}(X)$, then the capping operator is represented by a matrix
whose entries are certain power series in K{\"a}hler parameters with coefficients given
by rational functions of equivariant parameters for  $\KG\times \C^{\times}_q$.
{\rr By theorems  8.1.16 and 8.2.20 from \cite{pcmilect}} this matrix is the matrix of fundamental solution of a system of $q$-difference equations\footnote{\rr $\shift(u,z)$ is denoted by $\shift_{\sigma}(u,z)$ in \cite{pcmilect} for a shift $u\to u q^{\sigma}$ by a specific $\bA$-character~$\sigma$.}:
\be
\label{qdeS}
\begin{array}{l}
\capping(u, z q^{\cL})  \cL = \textbf{M}_{\cL}(u,z)  \capping(u , z) \\ 
\\
\capping(u q, z) \textsf{E} (u,z) =\shift(u,z) \capping(u,z) 
\end{array}
\ee
{\rr Here $\cL$ denotes the operator of multiplication by a line bundle $\cL \in \Pic(X)$, $\textsf{E} (u,z)$ is the operator of multiplication by K-theory class given by (8.2.13) in~\cite{pcmilect}. In particular $\textsf{E} (u,z)$ and $\cL$ commute. }

Recall that the $\Pic(X)$ is generated by the tautological line bundles $\cL_i=\det(\tb_i) $, $i=1,\dots,n$.
For a bundle $\cL = \cL_1^{\otimes m_1}\otimes \cdots \otimes \cL^{\otimes m_n}_n$  the following notation is used in (\ref{qdeS}):
$$
z q^{\cL} = (z_1 q^{m_1},\dots,z_n q^{m_n}).
$$

The operators $\shift(u,z)$ shifting the equivariant parameters are called
\textit{shift operators}. The operators $\geomM_{\cL}(u,z)$ corresponding to line bundles $\cL \in \Pic(X)$ are called
the \textit{quantum difference operators}. They are the main object of study in our paper.

\subsubsection{\label{tcl2}}
We can write the system (\ref{qdeS}) in the following equivalent form:
\be \label{shiftequations}
\begin{array}{l}
 \calK\, \capping(u,z)=\capping(u, z)\, \calK^{\infty}\\
\\
 \calA_{\cL}\, \capping(u,z)=\capping(u,z) \, \calA^{\infty}_{\cL}
\end{array}
\ee
with the following $q$-difference operators:
\be \label{qdiffops}
\begin{array}{ll}
 \calK = T^{-1}_{u} \shift(u,z), &  \calK^{\infty} = T^{-1}_{u} \textsf{E}(u,z) \\
\\
\calA_{\cL}= T^{-1}_{\cL} \textbf{M}_{\cL}(u,z) &  \calA^{\infty}_{\cL}= T^{-1}_{\cL} \cL
\end{array}
\ee
where $T_{\cL} f(u,z)=f(u,z q^{\cL})$ and $T_{u} f(u,z)=f(u q ,z)$. As $\cL$ and $\textsf{E}(u,z)$ commute, the consistency of this system of difference equations can be represented in the form of ``zero curvature'' condition:
\be
\label{compat1}
[ \, \calA_{\cL}, \calA_{\cL'} ]=0, \ \ [ \, \calA_{\cL}, \calK ] =0
\ee
where by $[A,B]=A B-B A$ we denote the commutators for $q$-difference operators.

\subsubsection{\label{torusdefsec}}
Let $\bA=\C^{\times}$ be a torus splitting the framing as $\bw=u \bw'+\bw''$. This torus acts on the Nakajima variety
$X=\cM(\bv,\bw)$ with the set of fixed points:
$$
X^{\bA}=\coprod\limits_{\bv'+\bv''=\bv}\,\cM(\bv',\bw')\times \cM(\bv'',\bw'')
$$
The stable map defined in the previous section can be used to identify $K_\KG(X)$ with $K_{\KG}(X^{\bA})$.
After such identification, {\rr the first equation in (\ref{shiftequations})} gets identified with the \textit{quantum Knizhnik-Zamolodchikov equation} ( qKZ ):\footnote{See Theorem 9.3.1 in \cite{MO} for similar statement in the case of equivariant cohomology.}
\begin{Theorem}(\cite{pcmilect}, Section 10)
\label{okth}
Let $\nabla\subset H^{2}(X,\R)$ be the alcove uniquely defined by the conditions:

1) $0 \in H^{2}(X,\R)$ is one of the vertices of $\nabla$

2) $\nabla \subset -C_{\textrm{ample}}$ ( opposite of the ample cone)

\noindent then for all $\slope \in \nabla$ we have\footnote{Note, that we use modified quantum parameter $z$ which differs by a sign:
$$
z^{\bv} \mapsto (-1)^{\textrm{codim}/2}  z^{\bv},
$$
see Theorem 10.2.8 in \cite{pcmilect}. Explicitly, this change of variables amounts to the following substitution of K{\"a}hler parameters:
$$
z_i \mapsto (-1)^{2 \kappa_i } z_i
$$
for canonical vector (\ref{shiftvect}).  To get rid of the minus sign, we will use modified notations in this paper.
:}
$$
Stab^{-1}_{+,T^{1/2},\slope}\,  {\rr \calK} \, Stab_{+,T^{1/2},\slope} = {\rr \qKZ^{s}}
$$
{\rr where $\calK$ is the q-difference operator defined by (\ref{qdiffops}) and ${\rr \qKZ^{s}}$ is the quantum Knizhnik-Zamolodchikov difference operator
\be \label{qkzop}
{\rr \qKZ^{s}}=\hbar_{(1)}^{\lambda} T^{-1}_{u} \Rtot^{\slope}(u)
\ee
for $R$-matrix $\Rtot^{\slope}(u)$ with slope $s$ (\ref{RL}) and $\hbar_{(1)}^{\lambda}$ defined by (\ref{hlam}).}
\end{Theorem}
Therefore, in the stable basis the first equation in (\ref{qdeS}) turns to the standard quantum Knizhnik-Zamolodchikov equation~\cite{FR}

\subsubsection{}
In Section \ref{comdif} we construct a system of difference operators
$$\qA^{\slope}_{\cL}=\T^{-1}_{\cL} \repM^{\slope}_{\cL}(u,z) , \ \ \cL \in \Pic(X)$$
with $\repM^{\slope}_{\cL}(u,z)$ given explicitly in terms of the algebra $\Uq(\widehat{\fg}_Q)$. These operators commute among themselves and with the qKZ operator \eqref{qkzop} for all slopes $\slope\in H^{2}(X,\mathbb{R})$:
\be
\label{compat2}
[ \, \qA^{\slope}_{\cL}, \qA^{\slope}_{\cL'} ]=0, \ \ [ \, \qA^{\slope}_{\cL}, \qKZ^{\slope} ] =0
\ee
We then prove our main result Theorem \ref{mainth}: 
the quantum difference operator $\geomM_{\cL}(u,z)$ is identified with $\repM_{\cL}^{\slope}(u,z)$ for $\slope$ as in Theorem \ref{okth}.
In particular the compatibility condition (\ref{compat1}) is identified with~(\ref{compat2}) for this slope.


\section{Commuting difference operators}

\subsection{Wall Knizhnik-Zamolodchikov equations}

\subsubsection{ \label{csec}}
It will be convenient to introduce a vector $\lambda=(\textbf{t}_1,...,\textbf{t}_n)$ such that $\hbar^{\textbf{t}_{i}}=z_i$, which means that $\lambda$ is a coordinate on a
universal cover $H^{2}(X,\C) \simeq {\mathbb{C}}^{|I|}$ of the K{\"a}hler moduli space.

Let us consider a Nakajima variety $X=\cM(\bv,\bw)$ and denote by $\bA$ a subtorus of the framing torus corresponding to a decomposition:
\be
\label{fset}
X^{\bA}=\coprod\limits_{\bv_1+...+\bv_n=\bv} \, \cM(\bv_1,\bw_1)\times\cdot\cdot\cdot\times \cM(\bv_n,\bw_n)
\ee
In this section we consider rational functions of parameters $z_i$ which take values in $End(K_{G}(X^{\bA}))$.
Using the above notations we will denote such functions as $f(z_i)$ or $f(\lambda)$.


 The first function we need $\hbar_{(k)}^{\lambda}\in End(K_{\bT}(X^{\bA}))$ is defined to be diagonal
in the basis supported on the set fixed points:
\be \label{hlam}
\hbar_{(k)}^{\lambda}(\gamma) = \hbar^{(\lambda, \bv_k)} \gamma =z^{\bv_{k,1}}_{1}\cdot\cdot\cdot \dz^{\bv_{k,n}}_{n}  \gamma
\ee
for a class $\gamma$ supported on a component $F=\cM(\bv_1,\bw_1)\times\cdot\cdot\cdot\times \cM(\bv_n,\bw_n)$.

We will need the so called \textit{dynamical notations} below. Let $\kappa$  be a linear combination of
dimension vectors; the particular combination of importance to us is
$$
\kappa=\frac12( C\bv-\bw)\,,
$$
where  $C$ is the Cartan matrix of the quiver. We define
$f(\lambda+\hat{\kappa}_{(i)})$ by:
$$
f(\lambda+\hat{\kappa}_{(i)})(\gamma) = f(\lambda+\kappa(\bv_i,\bw_i))(\gamma)
$$
for a class $\gamma$ supported on a component $F=\cM(\bv_1,\bw_1)\times\cdot\cdot\cdot\times \cM(\bv_n,\bw_n)$.
We will refer to such a transformation $f(\lambda) \rightarrow f(\lambda+\hat{\kappa}_{(i)})$ as the \textit{dynamical shift} of $f$ by a weight $\kappa$
in the $i$-component. In the case of one component we will omit subscript $(1)$ and write $f(\lambda+\hat{\kappa})$.

Define $q$-difference operators by
$T_{\dz_i} f(\dz_1,...,\dz_i,...,\dz_n)=f(\dz_1,...,\dz_i q,...,\dz_n)$.
We extend it to the action of $\textrm{Pic}(X) \simeq \Z^{n}$ by $q$-difference operators
$T_{\cL}$ as in Sections \ref{tcl1}-\ref{tcl2}.


\subsubsection{}
Below, we use definitions of triangular operators from Section \ref{triangularity}. {\rr The torus $\bA$ is as defined in Section \ref{torusdefsec}.}
\begin{Proposition}
 \label{Jtheorem}
There exist unique strictly upper triangular $J^{+}_{w}(\lambda)$ and strictly lower triangular $J^{-}_{w}(\lambda)$ solutions of the following ABRR equations:
\be
\label{wallKZ}
J^{+}_w(\lambda) \hbar_{(1)}^{-\lambda} \, \textsf{R}^{+}_{w}  =  \hbar_{(1)}^{-\lambda}\, \hbar^{\Omega} J^{+}_w(\lambda), \ \
  \, \textsf{R}^{-}_{w} \hbar_{(1)}^{-\lambda} J^{-}_w(\lambda) = J^{-}_w(\lambda)  \hbar^{\Omega} \hbar_{(1)}^{-\lambda}\,
\ee
Moreover, $J^{\pm}_{w}(\lambda)$ are elements in a completion $\Uq(\fg_w)\widehat{\otimes} \Uq(\fg_w)$  satisfying:
\be
\label{JpJm}
S_w\otimes S_w \Big((J^{+}_w(\lambda))_{21}\Big) = J^{-}_w(\lambda)
\ee
where the subscript $(21)$ stands for the transposition $(a\otimes b)_{(21)}=b\otimes a$ and $S_w$ is the antipode in $\Uq(\fg_w)$.
\end{Proposition}
\begin{proof}
We write the first ABRR equation in the form:
$$
Ad_{\hbar_{(1)}^{\lambda} \hbar^{-\Omega}} \Big(J^{+}_w(\lambda) \Big)= J^{+}_w(\lambda)  (R^{+}_{w})^{-1}
$$
(recall that  $R^{+}_w$ and $\textsf{R}^{+}_{w}$ are related by Theorem \ref{thmR}).
By assumption $J^{+}_w(\lambda)=\bigoplus\limits_{\langle \alpha, \theta \rangle >0 } \,J^{+}_w(\lambda)_{\alpha}$ where
$\theta$ is the stability parameter of the Nakajima variety. The wall $R$-matrix $R^{+}_w$ is upper triangular, thus, it has the same decomposition.  In the components the last equation is equivalent to the following system:
$$
Ad_{\hbar_{(1)}^{\lambda} \hbar^{-\Omega}} \Big(J^{+}_w(\lambda)_{\alpha} \Big) = J^{+}_w(\lambda)_{\alpha}+\cdots
$$
where $\cdots$ stands for the lower terms $J^{+}_w(\lambda)_{\alpha'}$,  i.e., the terms with
$\langle \alpha, \theta \rangle>\langle \alpha', \theta \rangle$. The operator $Ad_{\hbar_{(1)}^{\lambda} \hbar^{-\Omega}} -1$ is invertible for general $\lambda$, thus we can solve the last system recursively starting from the component of the minimal weight
 $J^{+}_w(\lambda)_{0}=1$. Thus the solution is unique. By construction
of the wall quantum algebra the $R$-matrix $\textsf{R}^{+}_{w}$ is an element of $\Uq(\fg_w)^{\otimes 2}$. Thus, the same is true for $J^{+}_w(\lambda)$.

Next, we apply the antipode $S_w\otimes S_w$ and the transposition to the first ABRR equation and use (\ref{ssr})-(\ref{antiom}) to obtain:
$$
\textsf{R}^{-}_{w} \hbar_{(2)}^{\lambda} S_w\otimes S_w \Big((J^{+}_w(\lambda))_{21} \Big)      =
 S_w\otimes S_w \Big((J^{+}_w(\lambda))_{21} \Big) \hbar_{(2)}^{\lambda}\, \hbar^{\Omega} \,
$$
It is clear that for any upper or lower triangular operator $X$ we have $\hbar_{(2)}^{\lambda} X \hbar_{(2)}^{-\lambda}=\hbar_{(1)}^{-\lambda} X \hbar_{(1)}^{\lambda}$, therefore, the last equation takes the form:
$$
\textsf{R}^{-}_{w} \hbar_{(1)}^{-\lambda} S_w\otimes S_w \Big((J^{+}_w(\lambda))_{21} \Big)      =
 S_w\otimes S_w \Big((J^{+}_w(\lambda))_{21} \Big) \hbar_{(1)}^{-\lambda}\, \hbar^{\Omega}
$$
By uniqueness of the solution we conclude $S_w\otimes S_w \Big((J^{+}_w(\lambda))_{21} \Big)=J^{-}_w(\lambda)$.
\end{proof}

{\bb 
	
Let $f(z)=f(z_1,\dots,z_n)$ be a function of the K\"ahler variables and $\theta=(\theta_1,\dots,\theta_n)$ be the  stability parameter of the Nakajima variety. We denote
$$
f(0_{\theta})=\lim\limits_{z\to 0} f(z^{\theta_1},\dots,z^{\theta_n}), \ \ \
f(\infty_{\theta})=\lim\limits_{z\to \infty} f(z^{\theta_1},\dots,z^{\theta_n}) 
$$
if these limits exist. 

\begin{Proposition} \label{limitJpropos}
$$
J^{+}_{w}(\infty_{\theta})=1, \ \ \ J^{+}_{w}(0_{\theta})=R^{+}_{w}
$$
\end{Proposition}		
\begin{proof}
We write the first ABRR equation in the form:
\be \label{abrr_euasss}
Ad_{\hbar_{(1)}^{\lambda} \hbar^{-\Omega}} \Big(J^{+}_w(\lambda) \Big)= J^{+}_w(\lambda)  (R^{+}_{w})^{-1}
\ee
Let us consider the corresponding components:
$$
J^{+}_w(\lambda)=1+\bigoplus\limits_{\langle \alpha, \theta \rangle >0 } \,J_{\alpha}(z), \ \ (R^{+}_w)^{-1}=1+\bigoplus\limits_{\langle \alpha, \theta \rangle >0 } \,R_{\alpha}
$$
The  $\alpha$-component of (\ref{abrr_euasss}) is
$$
J_{\alpha}(z) z^{\alpha} \hbar^{m} = J_{\alpha}(z)+\sum\limits_{{\gamma+\delta=\alpha} \atop {\langle \gamma,\theta \rangle <\langle\alpha,\theta\rangle}}\, J_{\gamma}(z) R_{\delta} 
$$
for some $m$ and where $z^{\alpha}=z_1^{\alpha_1}\cdots z_n^{\alpha_n}$. Thus
$$
J_{\alpha}(z)=\dfrac{1}{{z^{\alpha}} \hbar^{m} -1 } \sum\limits_{{\gamma+\delta=\alpha} \atop {{\langle \gamma,\theta \rangle <\langle\alpha,\theta\rangle}}}\, J_{\gamma}(z) R_{\delta} 
$$
By induction, assume that $J_{\gamma}(\infty_{\theta})=0$ for all $\gamma\neq 0$ with ${\langle \gamma,\theta \rangle <\langle\alpha,\theta\rangle}$. By triangularity $\langle \alpha, \theta \rangle>0$, and thus
$$
J_{\alpha}(\infty_{\theta})=\lim_{z\to \infty}\, \Big( \dfrac{1}{z^{\langle \alpha, \theta \rangle } \hbar^{m} -1 } \sum\limits_{{\gamma+\delta=\alpha} \atop {\langle \gamma,\theta \rangle <\langle\alpha,\theta\rangle}}\, J_{\gamma}(z^{\theta_1},\dots,z^{\theta_n}) R_{\delta} \Big)=0
$$
Therefore $J^{+}_{w}(\infty_{\theta})=1$.

Similarly, by induction, assume that $J_{\gamma}(0_{\theta})$ exists for all
$\gamma$ with ${\langle \gamma,\theta \rangle <\langle\alpha,\theta\rangle}$.  Then 
$$
J_{\alpha}(0_{\theta})=\lim_{z\to 0}\, \Big( \dfrac{1}{z^{\langle \alpha, \theta \rangle } \hbar^{m} -1 } \sum\limits_{{\gamma+\delta=\alpha} \atop {{\langle \gamma,\theta \rangle <\langle\alpha,\theta\rangle}}}\, J_{\gamma}(z^{\theta_1},\dots,z^{\theta_n}) R_{\delta} \Big)
$$
also exists. We conclude that $J^{+}_{w}(0_{\theta})$ exists.

Let us denote $\tilde{J}^{+}_{w}(\lambda)=\hbar^{\lambda}_{(1)} J^{+}_{w}(\lambda) \hbar^{-\lambda}_{(1)}$. Then
$$
\tilde{J}^{+}_{w}(\lambda)=1+\bigoplus\limits_{\langle \alpha, \theta \rangle >0 } \,\tilde{J}_{\alpha}(z)
$$
with $\tilde{J}_{\alpha}(z) = z^{\alpha} {J}_{\alpha}(z)$. Since $J_{\alpha}(0_{\theta})=\lim\limits_{z\to 0}\,J_{\alpha}(z^{\theta_1},\dots,z^{\theta_n})$ exists and $\langle \alpha, \theta \rangle>0$ we have
$$
\tilde{J}_{\alpha}(0_{\theta})=\lim\limits_{z\to 0}\, J_{\alpha}(z^{\theta_1},\dots,z^{\theta_n}) z^{\langle \alpha, \theta \rangle} =0.
$$
Therefore $\tilde{J}^{+}_{w}(0_{\theta})=1$. 

Finally, we rewrite the ABRR equation in the form:
$$
\tilde{J}^{+}_{w}(\lambda) \textsf{R}^{+}_{w}=\hbar^{\Omega} J^{+}_{w}(\lambda) 
$$
Using above limits at $0_{\theta}$ we obtain:
$$
\textsf{R}^{+}_{w}=\hbar^{\Omega} J^{+}_{w}(0_{\theta})  
$$
and therefore $J^{+}_{w}(0_{\theta})=R^{+}_{w}$. 
	
\end{proof}}

\subsubsection{}
Let $F=\cM(\bv_1,\bw_1)\times \cM(\bv_2,\bw_2)$ and $F^{\prime}=\cM(\bv'_1,\bw'_1)\times \cM(\bv'_2,\bw'_2)$ be two fixed components.
As we discussed in Section \ref{trian} the dependence of matrix elements of a wall $R$-matrix on the equivariant parameter $u$ is given by:
$$
\left.R^{+}_{w}(u)\right|_{F\times F^{\prime}} \sim  u^{\langle \bv_1-\bv_1', \cL_w \rangle}
$$
Thus, for $\textbf{s}$ with $\hbar^{\textbf{s}} =q$ and $\tau_{w}=\textbf{s} \cL_w $ we have
\be
\label{Rshift}
\hbar^{\tau_w}_{(1)}\, R^{+}_{w}(u)\, \hbar^{-\tau_w}_{(1)}=R^{+}_{w}(u q)
\ee
From the previous proposition we obtain:
\be
\label{Jshift}
\hbar^{\tau_w}_{(1)}\, J^{+}_{w}(u )\, \hbar^{-\tau_w}_{(1)}=J^{+}_w(u q)
\ee
Shifting  $\lambda \rightarrow \lambda - \tau_w$ in the ABRR equation (\ref{wallKZ}) and using the previous two identities we find:
$$
J^{+}_w(u,\lambda-\tau_w) \hbar_{(1)}^{-\lambda} \, \textsf{R}^{+}_{w}(u q)  =  \hbar_{(1)}^{-\lambda}\, \hbar^{\Omega} J^{+}_w(u q, \lambda-\tau_w)
$$
and same for $J^{-}_{w}$. Finally, denoting
\be\label{JJ}
\wJ^{\pm}_{w}(\lambda)=J^{\pm}_w(\lambda-\tau_w)
\ee
we rewrite the last relation in the form:
\begin{Proposition} \label{wallprop}
There exist unique strictly upper triangular $\wJ^{+}_{w}(\lambda) \in \Uq(\fg_{w})^{\otimes 2}$ and strictly lower triangular $\wJ^{-}_{w}(\lambda)\in \Uq(\fg_{w})^{\otimes 2}$ solutions of \textbf{wall Knizhnik-Zamolodchikov equations}:
\be
\label{wallKZ2}
\begin{array}{l}
\wJ^{+}_w(\lambda) \hbar_{(1)}^{-\lambda}  T_{u}   \textsf{R}^{+}_{w}  =   \hbar_{(1)}^{-\lambda}\, T_{u} \hbar^{\Omega} \wJ^{+}_w(\lambda) , \\
\\
\textsf{R}^{-}_{w} \hbar_{(1)}^{-\lambda} \,T_{u} \wJ^{-}_w(\lambda) = \wJ^{-}_w(\lambda)  \hbar_{(1)}^{-\lambda}\,T_{u} \hbar^{\Omega}
\end{array}
\ee
where  $T_u f(u)=f(u q)$.
\end{Proposition}

\subsection{Dynamical operators $\repM^{\slope}_{\cL}(\lambda)$ \label{comdif}}

\subsubsection{ \label{Boper}}
The following operator is playing a fundamental role in our paper. For a wall $w$ 
in the hyperplane arrangement (\ref{wallL}) we define:
\be
\label{bdef}
\begin{array}{|c|}\hline
	\\
	\ \ \repM_w(\lambda)=
	\left.\textbf{m} \Big(1 \otimes S_{w}( \, \wJ^{-}_{w}(\lambda)^{-1}\, )
	\Big)\right|_{\lambda\rightarrow\lambda+\kappa}\\
	\\ 
	\hline
\end{array}
\ee

Here $S_{w}$ is the antipode of the Hopf algebra $\Uq(\fg_{w})$ and $\textbf{m}(a\otimes b) \stackrel{def}{=}  a b $.
We denote by $\lambda\rightarrow\lambda+\kappa$ the  dynamical shift by the following vector:
\be
\label{shiftvect}
\kappa=(C \bv-\bw)/2\ee
where $C$ is the Cartan matrix of the corresponding quiver. Note that this operator is well defined
in the evaluation modules (even infinite dimensional) because the operator $\wJ^{-}_{w}(\lambda)$ is lower triangular and thus $\repM_w(\lambda)$ is normally ordered.
Note that by definition $\repM_w(\lambda)$ is an element in a completion of $\Uq(\fg_{w})(z_1,\dots, z_n)$.

{\rr
\begin{Remark}
In Section \ref{bslsec} we compute a universal formula for 	$\repM_w(\lambda)$ in the case of $\sldh$. Up to a difference in notations, this operator coincides with the element of the dynamical quantum group associated to a real root reflection. See Proposition 14 in \cite{EV} for an explicit formula in this case.
Thus, in the case of real roots the operator \eqref{bdef} coincides with the one constructed by Etingof-Varchenko. In contrast with the approach of \cite{EV}, the element (\ref{bdef}) is defined in a more general situation, see examples in Section~\ref{apa} for imaginary roots. 
\end{Remark}}

\subsubsection{\label{conjec1}}
Let $\cL \in \Pic(X)$ be a line bundle. Let us
fix a slope $\slope \in H^{2}(X,\R)$ and choose a path in $H^{2}(X,\R)$ from $\slope$ to $\slope-\cL$.
This path crosses finitely many walls in some order $\{w_{1},w_{2},...,w_m\}$.
For this choice of a slope, line bundle and a path we associate the following operator:
\be
\begin{array}{|c|}
	\hline\\
	\label{dMdef}
	\repM^{\slope}_{\cL}(\lambda)=  \cL \repM_{w_m}(\lambda) \cdot\cdot\cdot \repM_{w_1}(\lambda)\,  \ \   \\
	\\
	\hline
\end{array}
\ee
The symbol $\cL$ on the right side denotes the operator of multiplication by a line bundle in $K_{\KG}(X)$.
By construction, $\repM^{\slope}_{\cL}(\lambda)$ is an element in a completion of $\Uq(\fgh_Q)(z_1,\dots,z_n)$.

We define the $q$-difference operators:
\be
\label{qA}
\qA^{s}_{\cL}=\T^{-1}_{\cL}\, \repM^{\slope}_{\cL}(\lambda).
\ee
In Section \ref{corsec}  we  will show that (\ref{dMdef}) does not depend on a choice of a path and  for every slope $\slope$ the operators (\ref{qA}) commute. Thus,
they provide a representation of $\textrm{Pic}(X)$ by $q$-difference operators.

\subsection{Some properties of $\repM_{w}(\lambda)$ }
In this section we discuss various properties of
the operators (\ref{dMdef}) and associated $q$-difference connection (\ref{qA}). Our approach is close to one used in~\cite{ESS}.
\subsubsection{ \label{cocysec}}
Let $J^{\pm}_w(\lambda)$ be the operators introduced in Proposition~\ref{Jtheorem}. Let us denote $J^{\pm}(\lambda)^{12}=J^{\pm}_{w}(\lambda) \otimes 1$,
$J^{\pm}(\lambda)^{23}=1\otimes J^{\pm}_{w}(\lambda)$, $J^{\pm}(\lambda)^{12,3}=(\Delta_{w} \otimes 1)J^{\pm}_{w}(\lambda) $, $J^{\pm}(\lambda)^{1,23}=(1\otimes \Delta_{w} )J^{\pm}_w(\lambda)$ the operators in the corresponding completion of $\Uq(\fg_{w})^{\otimes 3}$. 
\begin{Theorem} \label{cocycle}
	The operators $J^{\pm}(\lambda)$ satisfy the dynamical cocycle conditions:
	\be
	\begin{array}{c}
		J^{-}(\lambda)^{12,3} J^{-}(\lambda+\hat{\kappa}_{(3)})^{12}=J^{-}(\lambda)^{1,23} J^{-}(\lambda-\hat{\kappa}_{(1)})^{23}\\
		\\
		J^{+}(\lambda+\hat{\kappa}_{(3)})^{12} J^{+}(\lambda)^{12,3}= J^{+}(\lambda-\hat{\kappa}_{(1)})^{23}J^{+}(\lambda)^{1,23}
	\end{array}
	\ee
	with dynamical shift $\kappa= (C \bv-\bw)/2 $ where $C$ is the Cartan matrix of the quiver.
\end{Theorem}
We will need the three-component analog of Proposition \ref{Jtheorem}. We start with the definition of upper/lower triangular
operators acting in a tensor product of three $\Uq(\fg_{w})$ modules.
Let $X=\cM(\bv,\bw)$ - be a Nakajima variety, and let $\bA$ be a torus splitting the framing such that:
\be
\label{cmp}
X^\bA=\coprod\limits_{\bv_1+\bv_2+\bv_3=\bv}\, \cM(\bv_1,\bw_1)\times\cM(\bv_2,\bw_2)\times\cM(\bv_3,\bw_3).
\ee
We say that an operator $A \in  \textrm{End}(K_{\bT}(X^\bA))$ is \textit{upper triangular} if $A=\bigoplus\limits_{{\langle\alpha, \theta \rangle>0} \atop {\langle \beta, \theta \rangle<0} }\,A_{\alpha, \beta}$ where $\theta$ is the stability parameter of the Nakajima variety and:
\begin{align}
\nonumber
A_{\alpha, \beta}: K_{\bT}(\cM(\bv_1,\bw_1)\times\cM(\bv_2,\bw_2)\times\cM(\bv_3,\bw_3)) \rightarrow \\ \nonumber
\\ \nonumber
K_{\bT}(\cM(\bv_1+\alpha,\bw_1)\times\cM(\bv_2+\gamma,\bw_2)\times\cM(\bv_3+\beta,\bw_3))
\end{align}
where $\gamma$ is fixed by the condition $\alpha+\beta+\gamma=0$. Similarly, the operator is lower triangular
if  $A=\bigoplus\limits_{{\langle\alpha, \theta \rangle<0} \atop {\langle \beta, \theta \rangle>0} }\,A_{\alpha, \beta}$ with the same $A_{\alpha,\beta}$ as above.
Finally, we say that an operator is strictly upper or lower triangular if, in addition,
$A_{0,0}=1$. For example, the product of wall $R$-matrices $R^{+,13}_{w} R^{+,12}_{w}$ or $R^{+,13}_{w} R^{+,23}_{w} $ (where the indices indicate in which components of (\ref{cmp}) the $R$-matrices act), are strictly upper triangular.

In the three-component case we have two types of qKZ operators
$\hbar^{\lambda}_{(3)} \textsf{R}_{w}^{+,13} \textsf{R}_{w}^{+,23}$  and $\hbar^{-\lambda}_{(1)} \textsf{R}_{w}^{+,13} \textsf{R}_{w}^{+,12}$ which correspond to the coproducts of the wall qKZ operators in the first or the second component. 
\begin{Proposition}
	If there exists a strictly upper triangular operator $J(\lambda) \in  \textrm{End}(K_{\bT}(X^\bA))$ satisfying:
	\be
	\begin{array}{l}
		J(\lambda)\, \hbar^{\lambda}_{(3)} \textsf{R}_{w}^{+,13} \textsf{R}_{w}^{+,23} =  \hbar^{\lambda}_{(3)} \hbar^{\Omega_{13}+\Omega_{23}}\,J(\lambda)\,  \\
		\\
		J(\lambda)\, \hbar^{-\lambda}_{(1)} \textsf{R}_{w}^{+,13} \textsf{R}_{w}^{+,12} = \hbar^{-\lambda}_{(1)} \hbar^{\Omega_{13}+\Omega_{12}}\, J(\lambda)
	\end{array}
	\ee
	or a strictly lower-triangular operator $J(\lambda) \in  \textrm{End}(K_{\bT}(X^\bA))$ satisfying
	$$
	\begin{array}{l}
	\textsf{R}_{w}^{-,23} \textsf{R}_{w}^{-,13} \hbar^{\lambda}_{(3)}  J(\lambda) =  \,J(\lambda)  \hbar^{\Omega_{23}+\Omega_{13}}\,\hbar^{\lambda}_{(3)}  \\
	\\
	\textsf{R}_{w}^{-,12} \textsf{R}_{w}^{-,13}\, \hbar^{-\lambda}_{(1)} \, J(\lambda) =  J(\lambda) \hbar^{\Omega_{12}+\Omega_{13}}\,\hbar^{-\lambda}_{(1)}
	\end{array}
	$$
	then it is unique.
\end{Proposition}
\begin{proof}
	We prove the upper-triangular case. The lower-triangular case is similar. Following \cite{EV} we introduce the operators:
	$$
	\begin{array}{l}
	A_{R}(X)= \hbar^{-\Omega_{13}-\Omega_{23}} \hbar^{-\lambda}_{(3)} X \hbar^{\lambda}_{(3)} \textsf{R}_{w}^{+,13}  \textsf{R}_{w}^{+,23}, \\
	\\
	A_{L}(X)= \hbar^{-\Omega_{13}-\Omega_{12}} \hbar^{\lambda}_{(1)}  X  \hbar^{-\lambda}_{(1)} \textsf{R}_{w}^{+,13} \textsf{R}_{w}^{+,12}
	\end{array}
	$$
	Assume that there exists an operator $J(\lambda)$ satisfying the conditions of the proposition. Then 
	$A_{R} A_{L} (J(\lambda))=J(\lambda)$. It is enough to check that the solution for this equation is unique.
	We are given that $J(\lambda) = \bigoplus\limits_{{\langle\alpha, \theta \rangle>0} \atop {\langle \beta, \theta \rangle<0} } J_{\alpha,\beta}(\lambda)$,
	and thus this equation has the following form in components:
	\be
	\label{trisystem}
	J_{\alpha,\beta}(\lambda)=Ad_{\hbar^{\lambda}_{(1)}\hbar^{-\lambda}_{(3)}\hbar^{-\bar{\Omega}}}\Big( J_{\alpha,\beta}(\lambda) \Big) +\cdots
	\ee
	where $\bar{\Omega}=2\,\Omega_{13}+\Omega_{23}+\Omega_{12}$ and $\cdots$ stands for the lower terms $J_{\alpha',\beta'}(\lambda)$  with
	$$
	\langle \alpha'-\beta', \theta \rangle  < \langle \alpha-\beta, \theta \rangle
	$$
	Note that the operator $1-Ad_{\hbar^{\lambda}_{(1)}\hbar^{-\lambda}_{(3)}\hbar^{-\bar{\Omega}}}$ is invertible for generic $\lambda$. This means that all $J_{\alpha,\beta}(\lambda)$ can be expressed through the lowest term $J_{0,0}(\lambda)=1$ and therefore they are uniquely determined by (\ref{trisystem}).
\end{proof}

Let $J(\lambda)$ be as in Proposition  \ref{Jtheorem}. It is obvious that
$J^{+}(\lambda+\hat{\kappa}_{(3)})^{12} J^{+}(\lambda)^{12,3} $ is a solution
of $A_{R}(X)=X$. Similarly
$ J^{+}(\lambda-\hat{\kappa}_{(1)})^{23} J^{+}(\lambda)^{1,23} $ is a solution of $A_{L}(X)=X$. Thus, by the previous proposition, to prove Theorem \ref{cocycle} it is enough to prove
the following lemma:
\begin{Lemma}
	$$
	\begin{array}{l}
	X=J^{+}(\lambda+\hat{\kappa}_{(3)})^{12} J^{+}(\lambda)^{12,3}  \ \  \textrm{is a solution of} \ \  A_{L}(X)=X\\
	\\
	X= J^{+}(\lambda-\hat{\kappa}_{(1)})^{23} J^{+}(\lambda)^{1,23}  \ \ \textrm{is a solution of} \ \  A_{R}(X)=X
	\end{array}
	$$
\end{Lemma}
\begin{proof}
	As noted above the element $X= J^{+}(\lambda+\hat{\kappa}_{(3)})^{12}J^{+}(\lambda)^{12,3}$ is a solution of $A_{R}(X)=X$. Note that $A_{R}$ and $A_{L}$ commute
	(due to the Yang-Baxter equation for $\Rwal^{+}_{w}$). Thus, $Y=A_{L}(X)$ is also a solution of this equation. The solution of $A_{R}(X)=X$
	is uniquely determined by the degree zero part in the third component. Let us denote this component of $X$ by $X_0$ and similarly for $Y$ by $Y_0$.
	Enough to prove that $X_0=Y_0$. For $X_0$ we obtain:
	$$
	X_0=J^{+}(\lambda+\hat{\kappa}_{(3)})^{12}
	$$
	For $Y_{0}$ we have:
	\be
	\label{inteq}
	Y_{0}=  \hbar^{-\Omega_{13}-\Omega_{12}} \hbar^{\lambda}_{(1)} J^{+}(\lambda+\hat{\kappa}_{(3)})^{12}\, \hbar^{-\lambda}_{(1)} \hbar^{\Omega_{13}} \textsf{R}_{w}^{+,12}
	\ee
	Set $Z=J^{+}(\lambda+\hat{\kappa}_{(3)})^{12}$. By triangularity of $R$-matrix and $J(\lambda)$ it factors
	$Z=\bigoplus\limits_{\alpha \in {\mathbb{N}}^{I}}\, Z_{\alpha}$ with:
	\begin{align}
	\nonumber
	Z_{\alpha}: K_{\KG}\Big(\cM(\bv_1,\bw_1)\times\cM(\bv_2,\bw_2)\times\cM(\bv_3,\bw_3)\Big) \\ \nonumber
	\\ \nonumber
	\rightarrow K_{\KG}\Big(\cM(\bv_1+\alpha,\bw_1)\times\cM(\bv_2-\alpha,\bw_2)\times\cM(\bv_3,\bw_3)\Big)
	\end{align}
	Thus, by the definition of codimension function (\ref{codfun}) we have $\hbar^{-\Omega_{13}} Z_{\alpha} \hbar^{\Omega_{13}}=\hbar^{m_\alpha} Z_{\alpha}$ where
	$$
	\begin{array}{l}
	m_\alpha = \\
	\\
	\frac{1}{2}\Big(\langle \bv_1,\bw_3\rangle+\langle \bv_3,\bw_1 \rangle- \langle \bv_1, C \bv_3\rangle\Big)
	-\frac{1}{2}\Big(\langle \bv_1+\alpha,\bw_3\rangle+\langle \bv_3,\bw_1\rangle- \langle \bv_1+\alpha, C \bv_3\rangle\Big) \\
	\\
	= \langle \alpha, {\kappa}_{(3)}\rangle
	\end{array}
	$$
	with $\kappa_{(3)}=(C \bv_3-\bw_3)/2$. Therefore, using the dynamical notations we can write the equation (\ref{inteq}) in the form:
	$$
	Y_{0}=  \hbar^{-\Omega_{12}} \hbar^{\lambda+\hat{\kappa}_{(3)}}_{(1)} J^{+}(\lambda+\hat{\kappa}_{(3)})^{12}\, \hbar^{-\lambda-\hat{\kappa}_{(3)}}_{(1)} \textsf{R}_{w}^{+,12}
	$$
	As $J^{+}(\lambda)$ satisfies the condition of Proposition \ref{Jtheorem} we obtain $Y_{0}=J^{+}(\lambda+\hat{\kappa}_{(3)})^{12}$. Therefore $Y=X$.
\end{proof}

{\rr
\begin{Corollary} \label{corrcocy}
The wall $R$-matrices $R^{+}_{w}$ satisfy the cocycle condition:
$$
(R^{+}_{w})^{12} (R^{+}_{w})^{12,3}=(R^{+}_{w})^{23} (R^{+}_{w})^{1,23}   
$$  
\end{Corollary}
\begin{proof}
By Proposition \ref{limitJpropos}, $J_w^{+}(0_{\theta})=R^{+}_{w}$. The result follows by evaluating the limit of the second identity in Theorem \ref{cocycle}  at $0_{\theta}$. 
\end{proof}
}

\subsubsection{}
Let us consider the operators:
$$
\tilde{\Wc}_w^{\prime}(\lambda)=\textbf{m}\Big( 1\otimes S_w ( J^{-}_w(\lambda)^{-1}) \Big), \ \
\Wc_w^{\prime}(\lambda)=\textbf{m}_{21}\Big( S_w^{-1} \otimes 1 ( J_w^{-}(\lambda)^{-1}) \Big), $$
where $S_w$ is the antipode of $\Uq(\fg_w)$ and $\textbf{m}(a\otimes b) \stackrel{def}{=}a b $,  $\textbf{m}_{21}(a\otimes b)\stackrel{def}{=} ba $.
We define:
\be
\label{rB}
\tilde{\Wc}_w(\lambda)= \tilde{\Wc}_w^{\prime}( \lambda +\hat{\kappa}), \ \ \Wc_w(\lambda)= \Wc_w^{\prime}( \lambda-\hat{\kappa})
\ee
with $\kappa$ as in Theorem \ref{cocycle}. 
\begin{Theorem}
	\label{coprB}
	$$
	\begin{array}{l}
	1) \ \  \Delta_w \tilde{\Wc}_w(\lambda) = J^{-}_w(\lambda) \,  \Big( \tilde{\Wc}_w (\lambda+\hat{\kappa}_{(2)} ) \otimes \tilde{\Wc}_w (\lambda -\hat{\kappa}_{(1)} ) \Big) \,  J^{+}_w(\lambda) \\
	\\
	2) \ \ \Delta_w \Wc_w(\lambda) = J^{-}_{w}(\lambda) \,  \Big( \Wc_w(\lambda +\hat{\kappa}_{(2)}  ) \otimes \Wc_w(\lambda - \hat{\kappa}_{(1)}) \Big) \, J^{+}_w(\lambda)
	\end{array}
	$$
\end{Theorem}

\begin{proof}
	Let $X(\lambda)=J^{-}_{w}(\lambda)^{-1}$. By Theorem \ref{cocycle}:
	\be
	\label{dualcocycle}
	X^{12}(\lambda+\hat{\kappa}_{(3)}) X^{12,3}(\lambda)= X^{23}(\lambda-\hat{\kappa}_{(1)}) X^{1,23}(\lambda)
	\ee
Set
	$\hbar^{(\lambda, m )}= z^m = z_1^{m_1}\cdot\cdot\cdot {\rr z_n}^{m_n}$ where $m=(m_1,\cdots,m_n)$ is a multi-index.
	We write our operators as power series:
	$$
	X(\lambda)=\sum\limits_{i,m} a_{i,m}\otimes b_{i, m}  z^m, \ \ \ J^{-}_{w}(\lambda)=X^{-1}(\lambda)=\sum\limits_{i,m} \bar{a}_{i,m}\otimes \bar{b}_{i, m}  z^m
	$$
	then
	$$
	\tilde{\Wc}^{\prime}_w(\lambda)= \textbf{m}\Big( 1\otimes S_w ( X(\lambda) ) \Big) = \sum\limits_{i,m} a_{i,m} S_w (b_{i, m})   z^m
	$$
	and in the sumless Sweedler notations we have:
	$$
	\Delta_w \tilde{\Wc}^{\prime}_{w}(\lambda) =  \sum\limits_{i,m} a^{(1)}_{i,m} S_w (b^{(2)}_{i, m}) \otimes a^{(2)}_{i,m} S_w (b^{(1)}_{i, m})   z^m
	$$
	We denote by $\hat{A}$ the following contraction:
	$$\hat{A}( a_1 \otimes a_2 \otimes a_3 \otimes a_4 ) = a_1 S_w(a_4) \otimes a_2 S_w(a_3),$$
	then, obviously $ \Delta_w \tilde{\Wc}^{\prime}_{w}(\lambda) = \hat{A} ( \Delta_w \otimes \Delta_w (X) )$.  From (\ref{dualcocycle}) we have:
	$$
	\Delta_w \otimes 1( X (\lambda) ) = X^{12}(\lambda + \hat{\kappa}_{(3)})^{-1} X^{23}(\lambda-\hat{\kappa}_{(1)}) X^{1,23}(\lambda)
	$$
	or in the components:
	$$
	\Delta_w \otimes 1( X (\lambda))= \sum\limits \,  (\bar{a}_{i,m} \otimes \bar{b}_{i,m} \otimes K^{m}) (K^{-s} \otimes a_{j,s} \otimes b_{j,s})
	(a_{k,l} \otimes b^{(1)}_{k,l} \otimes b^{(2)}_{k,l})\, z^{m+s+l}=
	$$
	$$
	=\sum\limits \,(\bar{a}_{i,m} K^{-s} a_{k,l}  \otimes \bar{b}_{i,m} a_{j,s}  b^{(1)}_{k,l} \otimes K^{m} b_{j,s} b^{(2)}_{k,l}) \, z^{m+s+l}
	$$
	where we denoted by $K=\hbar^{\kappa}$. Now,  $\Delta_w\otimes \Delta_w= (1 \otimes 1 \otimes \Delta_w)(\Delta_w \otimes 1) $ and therefore:
	$$
	\Delta_w \otimes \Delta_w X ({\rr \lambda}) = \sum\limits \,(\bar{a}_{i,m} K^{-s} a_{k,l}  \otimes \bar{b}_{i,m} a_{j,s}  b^{(1)}_{k,l} \otimes K^{m} b^{(1)}_{j,s} b^{(2),(1)}_{k,l}\otimes K^{m} b^{(2)}_{j,s} b^{(2),(2)}_{k,l}) \, z^{m+s+l}
	$$
	Applying contraction $\hat{A}$, taking into account that the antipode $S_w$ is an antihomomorphism and $S_w(K)=K^{-1}$ by (\ref{antipcart}) we obtain:
	$$
	\begin{array}{l}
	\hat{A} ( \Delta_w \otimes \Delta_w X )=\\
	\\
	\sum\limits \, \bar{a}_{i,m} K^{-s} a_{k,l} S_w(b^{(2),(2)}_{k,l}) S_w(b^{(2)}_{j,s}) K^{-m} \otimes  \bar{b}_{i,m} a_{j,s}  b^{(1)}_{k,l}  S_w(b^{(2),(1)}_{k,l}) S_w(b^{(1)}_{j,s}) K^{-m} \, z^{m+s+l}\\
	\\
	=J^{-}_{w}(\lambda-\hat{\kappa}_{(1)}-\hat{\kappa}_{(2)}) \sum\limits \,K^{-s} a_{k,l} S_w(b^{(2),(2)}_{k,l}) S_w(b^{(2)}_{j,s}) \otimes   a_{j,s}  b^{(1)}_{k,l}  S_w(b^{(2),(1)}_{k,l}) S_w(b^{(1)}_{j,s}) \, z^{s+l}
	\end{array}
	$$
	where $J^{-}_{w}(\lambda-\hat{\kappa}_{(1)}-\hat{\kappa}_{(2)})=\sum\, \bar{a}_{i,m} K^{-m}\otimes  \bar{b}_{i,m} K^{-m} z^m  $ and in the last step we used that the whole operator is weight zero and therefore commutes with $K\otimes K$. From the simple Lemma \ref{ellem} below, we obtain:
	$$
	\begin{array}{l}
	\hat{A} ( \Delta_w \otimes \Delta_w X(\lambda) )=\\
	\\
	J^{-}_{w}(\lambda-\hat{\kappa}_{(1)}-\hat{\kappa}_{(2)}) \sum\limits \,K^{-s} a_{k,l} S_w(b_{k,l}) S_w(b^{(2)}_{j,s}) \otimes   a_{j,s} S_w(b^{(1)}_{j,s}) \, z^{s+l}=
	\\
	\\
	J^{-}_{w}(\lambda-\hat{\kappa}_{(1)}-\hat{\kappa}_{(2)}) \tilde{\Wc}_w^{\prime}( \lambda ) \otimes 1 \cdot \Big(  \sum\limits \,K^{-s} S_w(b^{(2)}_{j,s}) \otimes   a_{j,s} S_w(b^{(1)}_{j,s}) \, z^{s} \Big).
	\end{array}
	$$
	Let us consider the contraction defined by $\hat{P}( a_1 \otimes a_2 \otimes a_3 ) = S_w(a_3) \otimes a_1 S_w(a_2)$. For the expression
	in the brackets in the last formula we have:
	$$
	\sum\limits \,K^{-s} S_w(b^{(2)}_{j,s}) \otimes   a_{j,s} S_w(b^{(1)}_{j,s}) \, z^{s} = \hat{P}\Big( X^{\rr 1,23}(\lambda+\hat{\kappa}_{(3)} ) \Big)
	$$
	Again, by (\ref{dualcocycle}) we have:
	$$
	\begin{array}{l}
	X^{1,23}( \lambda+\kappa_{(3)} )=X^{23}(\lambda-\hat{\kappa}_{(1)}+\hat{\kappa}_{(3)} )^{-1} X^{12}(\lambda+2 \hat{\kappa}_{(3)}) X^{12,3}(\lambda +\hat{\kappa}_{(3)})\\
	\\
	=\sum \, K^{-m} a_{j,s} a_{k,l}^{(1)} \otimes  \bar{a}_{i,m} b_{j,s} a_{k,l}^{(2)} \otimes \bar{b}_{i,m} K^{m+2 s} b_{k,l} K^{l} z^{s+m+l}.
	\end{array}
	$$
	Thus
	$$
	\begin{array}{l}
	\hat{P}\Big( X^{1,23}( \lambda+\hat{\kappa}_{(3)}  ) \Big)=\\
	\\
	\sum K^{-l} S_w(b_{k,l}) K^{-m-2 s} S_w(\bar{b}_{i,m}) \otimes  K^{-m} a_{j,s} a_{k,l}^{(1)} S(a_{k,l}^{(2)}) S_w(b_{j,s}) S_w(\bar{a}_{i,m}) z^{s+m+l}
	\end{array}
	$$
	Noting that $a_{k,l}^{(1)} S_w(a_{k,l}^{(2)}) = \epsilon_w( a_{k,l} )$ we find:
	$$
	\begin{array}{l}
	\hat{P}\Big( X^{\rr 1,23}( \lambda+\hat{\kappa}_{(3)} ) \Big)= \sum K^{-m-2 s} S_w(\bar{b}_{i,m}) \otimes  K^{-m} a_{j,s}  S_w(b_{j,s}) S_w(\bar{a}_{i,m}) z^{s+m}\\
	\\
	=\Big(\sum K^{-2 s} \otimes  a_{j,s}  S_w(b_{j,s})  z^{s}\Big) \Big( \sum K^{-m}  S_w(\bar{b}_{i,m}) \otimes  K^{-m}  S_w(\bar{a}_{i,m}) z^{m}\Big)\\
	\\
	= \Big(1\otimes \tilde{\Wc}^{\prime}_{w}(\lambda-2 \hat{\kappa}_{(1)})\Big) S_w\otimes S_w((J^{-}_w)_{21}) (\lambda-\hat{\kappa}_{(1)}-\hat{\kappa}_{(2)})
	\end{array}
	$$
	Overall we obtain the identity:
	$$
	\Delta_w \tilde{\Wc}_{w}^{\prime}(\lambda) =J^{-}_{w}(\lambda-\hat{\kappa}_{(1)}-\hat{\kappa}_{(2)}) \, \Big(\tilde{\Wc}^{\prime}_{w}( \lambda) \otimes \tilde{\Wc}^{\prime}_{w}(\lambda-2 \hat{\kappa}_{(1)}) \Big)\, S_w\otimes S_w((J^{-}_{w})_{21}) (\lambda-\hat{\kappa}_{(1)} -\hat{\kappa}_{(2)})
	$$
	Finally, after shifting $\lambda \rightarrow \lambda +\hat{\kappa}_{(1)}+\hat{\kappa}_{(2)} $ and using (\ref{JpJm}) we obtain  $1)$. The equation $2)$ is obtained similarly.
\end{proof}
\begin{Lemma}
	\label{ellem}
	\be
	\sum S(x^{(2),(2)}) \otimes  x^{(1)}  S(x^{(2),(1)}) = S(x) \otimes 1
	\ee
\end{Lemma}
\begin{proof}
	Consider the contraction $\hat{C}( a_1 \otimes a_2 \otimes a_3 ) = S(a_3) \otimes a_1 S(a_2)$ then, obviously
	$$
	\begin{array}{l}
	\sum S(x^{(2),(2)}) \otimes  x^{(1)}  S(x^{(2),(1)})= \hat{C}\Big( 1\otimes \Delta (\Delta x) \Big)
	
	=\hat{C}\Big( \Delta \otimes 1 (\Delta x) \Big) \\
	\\= S(x^{(2)}) \otimes x^{(1)(1)} S(x^{(1) (2) })
	=S(x^{(2)}) \otimes \epsilon( x^{(1)}) = S(x) \otimes 1
	\end{array}
	$$
\end{proof}

\begin{Corollary}
	The coproduct of the operator $\repM_w(\lambda)$ defined by (\ref{bdef}) has the following form:
	\be\label{Mcop}
	\Delta_w(\repM_w(\lambda)) = \wJ^{-}_{w}(\lambda) \,  \Big( \repM_w(\lambda +\hat{\kappa}_{(2)}  ) \otimes \repM_w(\lambda - \hat{\kappa}_{(1)}) \Big) \, \wJ^{+}_w(\lambda)
	\ee
\end{Corollary}
\begin{proof}
	Shift $\lambda \rightarrow \lambda -\tau_w$ and use definitions (\ref{JJ}) and (\ref{bdef}).
\end{proof}

\subsubsection{}
Let us consider the wall qKZ operators as in the Proposition \ref{wallprop}:
\be
\qKZ^{+}_{w} = T_{u} \hbar^{-\lambda}_{(1)} \Rwal^{+}_{w}, \ \ \  \qKZ^{-}_{w} = \Rwal^{-}_{w} T_{u} \hbar^{-\lambda}_{(1)}
\ee
acting in the tensor product of two evaluation modules of $\Uq(\fgh_{Q})$.
\begin{Proposition}
	\be
	\label{Mint}
	\qKZ^{-}_{w} \, \Delta_w(\repM_w(\lambda))=\Delta_w(\repM_w(\lambda)) \, \qKZ^{+}_{w}
	\ee
\end{Proposition}

\begin{proof}
	We have
	$$
	\begin{array}{l}
	\qKZ^{-}_{w}\Delta_w(\repM_w(\lambda))\stackrel{(\ref{Mcop})}{=}\\
	\\
	\Rwal^{-}_{w} T_{u} \hbar^{-\lambda}_{(1)}
	\wJ^{-}_{w}(\lambda) \,  \Big( \repM_w(\lambda +\hat{\kappa}_{(2)}  ) \otimes \repM_w(\lambda - \hat{\kappa}_{(1)}) \Big) \, \wJ^{+}_w(\lambda)\stackrel{(\ref{wallKZ2})}{=}\\
	\\   \wJ^{-}_{w}(\lambda) T_{u} \hbar^{-\lambda}_{(1)}  \hbar^{\Omega} \,  \Big(\repM_w(\lambda +\hat{\kappa}_{(2)}  ) \otimes \repM_w(\lambda - \hat{\kappa}_{(1)}) \Big) \, \wJ^{+}_w(\lambda)=\\
	\\
	\wJ^{-}_{w}(\lambda)  \,  \Big( \repM_w(\lambda +\hat{\kappa}_{(2)}  ) \otimes \repM_w(\lambda - \hat{\kappa}_{(1)}) \Big) \, T_{u}  \hbar^{-\lambda}_{(1)}  \hbar^{\Omega} \wJ^{+}_w(\lambda)\stackrel{(\ref{wallKZ2})}{=}\\
	\\
	\wJ^{-}_{w}(\lambda)  \,  \Big( \repM_w(\lambda +\hat{\kappa}_{(2)}  ) \otimes\repM_w(\lambda - \hat{\kappa}_{(1)}) \Big)    \wJ^{+}_w(\lambda) \hbar^{-\lambda}_{(1)}  T_{u} \textsf{R}^{+}_{w} =\\
	\\
	\Delta_w(\repM_w(\lambda))\, \qKZ^{+}_{w}
	\end{array}
	$$
\end{proof}

\begin{Proposition}
	\label{ProM}
	For $\cL \in \Pic(X)$ the operators $\repM_w(\lambda)$ satisfy:
	\be
	\label{dva}
	\cL T^{-1}_{\cL}\,  \repM_{w}(\lambda) =  \repM_{w+\cL}(\lambda) \, \cL T^{-1}_{\cL}.
	\ee

\end{Proposition}
\begin{proof}
	Let $\bA$ be a torus splitting the framing $\bw=u' \bw'+u'' \bw''$. For a Nakajima variety
	$X=\cM(\bv,\bw)$ the components of $X^{\bA}$ are of the form $F_{i}=\cM(\bv'_i,\bw')\times \cM(\bv''_i,\bw'')$.
	Let us consider the operators:
	$$
	S_{\fC,\slope}=i^{*}_{X^{\bA}} \, \circ \, \Stab_{\fC,T^{1/2}, \slope}  \,: K_{\KG}(X^{\bA}) \longrightarrow K_{\KG}(X^{\bA})
	$$
	where $i_{X^{\bA}}$ is the inclusion map. Let $\cL \in \Pic(X)$ be a line bundle. We denote by $U(\cL)$ a block diagonal operator acting in $K_{\KG}(X^{\bA})$ with the following matrix elements:
	$$
	\left.U(\cL)\right|_{F_{i}\times F_{i}} = \left. \cL\right|_{F_i}
	$$
	Let us consider an operator $\bar{S}_{\fC,\slope}=U(\cL) S_{\fC,\slope} U(\cL)^{-1}$. A conjugation by a diagonal matrix does not change the diagonal elements, thus:
	\be
	\label{hulleq}
	\left.\bar{S}_{\fC,\slope}\right|_{F_i\times F_i }=\left.S_{\fC,\slope}\right|_{F_i\times F_i }
	\ee
	For the non-diagonal elements we have:
	\be
	\begin{array}{l}
		\deg_{\bA} \Big(\left.\bar{S}_{\fC,\slope}\right|_{F_2\times F_1 }\Big) = \deg_{\bA}\Big(\left.S_{\fC,\slope}\right|_{F_2\times F_1 } \dfrac{\left.\cL\right|_{F_2}}{\left.\cL\right|_{F_1}}\Big)\\
		\\ \stackrel{(\ref{degK})}{\subset}  \deg_{\bA}\Big(\left.S_{\fC,\slope}\right|_{F_2\times F_2 } \dfrac{\left. s \otimes \cL\right|_{ F_2}}{\left.s \otimes \cL\right|_{F_1}}\Big)\stackrel{(\ref{hulleq})}{=}
		\deg_{\bA}\Big(\left.\bar{S}_{\fC,\slope}\right|_{F_2\times F_2 } \dfrac{\left. s \otimes \cL\right|_{ F_2}}{\left.s \otimes \cL\right|_{F_1}}\Big)
	\end{array}
	\ee
	Note that the stable map is defined uniquely by these restrictions and thus we conclude: $\bar{S}_{\fC,\slope}=S_{\fC,\slope+\cL}$.
	
	Recall that the  wall $R$-matrices are defined by $
	R^{\pm}_{w} = S_{\pm,\slope_2}^{-1} S_{\pm,\slope_1} $ for two slopes
	$\slope_1$ and $\slope_2$ separated by a single wall $w$. Therefore:
	$$
	U(\cL) R^{\pm}_{w} U(\cL)^{-1}=R^{\pm}_{w+\cL}.
	$$
	Conjugating both sides of ABRR equation (\ref{wallKZ}) by $U(\cL)$ we get:
	$$
	\textsf{R}^{-}_{w+\cL} \hbar_{(1)}^{-\lambda}  U(\cL) J^{-}_w(\lambda) U(\cL)^{-1} =
	U(\cL) J^{-}_w(\lambda) U(\cL)^{-1} \hbar_{(1)}^{-\lambda}\, \hbar^{\Omega}
	$$
	Thus, by uniqueness of the solution of this equation:
	\be
	\label{lbb}
	U(\cL)  J^{-}_w(\lambda) U(\cL)^{-1} =J^{-}_{w+\cL}(\lambda)
	\ee
	Without a loss of generality we can assume that $\cL = \det (\tb_{k})$ is the $k$-th tautological line bundle. Then, we have:
	$$
	U(\cL)=\tilde{\cL} \otimes \tilde{\cL}
	$$
	where $\tilde{\cL}$ is the same tautological bundle twisted by some powers of trivial line bundles $u'$ and $u''$:
	explicitly for the component $F=\cM(\bv',\bw')\times \cM(\bv'',\bw'')$ we have: $\left.\cL\right|_{F} =  (u')^{\bv'_{k}} \cL \otimes (u'')^{\bv''_{k}} \cL $.
	
	Let $(J^{-}_{w}(\lambda))^{-1}=\sum_i a_i \otimes b_i$, $(J^{-}_{w+\cL}(\lambda))^{-1}=\sum_i a^{\prime}_i \otimes b^{\prime}_i$
	so that $\tilde{B}'_w(\lambda)=\sum_i a_i S_w(b_i)$. Then we have:
	$$
	\begin{array}{l}
	\cL \tilde{B}'_w(\lambda) \cL^{-1}=\tilde{\cL} \tilde{B}'_w(\lambda) \tilde{\cL}^{-1}=\sum_i \tilde{\cL} a_i S_w(\tilde{\cL} b_i)=\\
	\\ \textbf{m}\Big(1\otimes S_w  (\sum_i \tilde{\cL} a_i \otimes \tilde{\cL} b_i) \Big)
	\stackrel{(\ref{lbb})}{=}\textbf{m}\Big( 1\otimes S_w ( \sum_i  a^{\prime}_i \tilde{\cL} \otimes  b^{\prime}_i \tilde{\cL}) \Big) \\
	\\ =\sum_i a^{\prime}_i S_w(b^{\prime}_i)=\tilde{B}'_{w+\cL}(\lambda).
	\end{array}
	$$
	In the first equality we substituted $\cL$ by $\tilde{\cL}$ because for the one component case the $u$-factors cancel.
	Thus we proved that:
	$$
	\cL \tilde{B}'_w(\lambda)=\tilde{B}'_{w+\cL}(\lambda) \cL
	$$
	Note that  $\tilde{B}'_w(\lambda)=\repM_w (\lambda-\kappa+\tau_{w})$  and thus:
	$$
	\cL \repM_w (\lambda-\kappa+\tau_{w})=\repM_{w+\cL}(\lambda-\kappa+\tau_{w+\cL}) \cL
	$$
	By definition $\tau_{w+\cL}-\tau_{w}=\textbf{s}\cL$ thus, after substitution $\lambda\rightarrow \lambda+\kappa -\tau_{w}- \textbf{s} \cL $
	we obtain:
	$$
	\cL \repM_w (\lambda-\textbf{s}\cL)=\repM_{w+\cL}(\lambda) \cL
	$$
	which gives (\ref{dva}).
\end{proof}
\subsubsection{}
The following proposition describes the action of the difference operators (\ref{qA})
in the tensor product of two $\Uq(\fgh_Q)$ modules.
\begin{Proposition}
	\label{coprodprop}
	$$
	\Delta_s(\qA^{\slope}_{\cL})=W_{w_0}( \lambda ) W_{w_1}( \lambda ) \cdots W_{w_{m-1}}( \lambda ) \Delta_{\infty}(\cL) T_{\cL}^{-1}
	$$
	where $w_0,\cdots, w_{m-1}$ is the ordered set of walls separating slopes $\slope$ and $\slope+\cL$, $W_{w}( \lambda ) = \Delta_{w}( \repM_{w}(\lambda)) (R^{+}_{w})^{-1}$ and $\Delta_{\infty}$ is the infinite slope coproduct from Section  \ref{infcopsec}.
\end{Proposition}

\begin{proof}
	First, by definition (\ref{dMdef}) we have:
	$$
	\qA^{\slope}_{\cL}= T_{\cL}^{-1} \cL \,\repM_{w_{-m}}(\lambda)  \cdots \repM_{w_{-2}}(\lambda) \repM_{w_{-1}}(\lambda)
	$$
	where, we denote by $w_{-1},\cdots, w_{-m}$ the ordered set of walls between the slope $\slope$ and $\slope -\cL$.
	By Proposition \ref{ProM}  we know that $T_{\cL}^{-1} \cL  \repM_{w_{k}}(\lambda) =\repM_{w_{k+m}}(\lambda) T_{\cL}^{-1} \cL$ and thus we obtain:
	$$
	\qA^{\slope}_{\cL}= \repM_{w_0}(\lambda) \repM_{w_1}(\lambda) \cdots \repM_{w_{m-1}}(\lambda) \,\cL\, T_{\cL}^{-1}
	$$
	where we denote $w_{k+m}=w_{k}+\cL$ (recall that the hyperplane arrangement is $\textrm{Pic}(X)$ periodic).
	
	Next, for the coproduct we have:
	$$
	\Delta_s(\qA^{\slope}_{\cL})= \Delta_{s}( \repM_{w_0}(\lambda) \repM_{w_1}(\lambda) \cdots \repM_{w_{m-1}}(\lambda) \cL) T_{\cL}^{-1}
	$$
	and by (\ref{conjcop}) the coproducts at different slopes are related as follows
	$$\Delta_{s} ( \repM_{w_k}(\lambda) ) = (R^{+}_{w_{0}})^{-1}  \cdots (R^{+}_{w_{k-1}})^{-1}\Delta_{w_k}(\repM_{w_k}) R^{+}_{w_{k-1}}\cdots  R^{+}_{w_{0}}.$$ Thus we obtain:
	$$
\Delta_{s}(\qA^{\slope}_{\cL}) = \Delta_{w_0}( \repM_{w_0}(\lambda)) (R^{+}_{w_{0}})^{-1}  \cdots \Delta_{w_{m-1}}(\repM_{w_{m-1}}(\lambda)) (R^{+}_{w_{m-1}})^{-1} R^{+}_{w_{m-1}}\cdots R^{+}_{w_{0}}  \Delta_{s}(\cL)  T_{\cL}^{-1}
	$$
	The proposition follows from next Lemma.
\end{proof}
\begin{Lemma}
	Let $w_0,\cdots, w_{m-1}$ be the ordered set of walls between $\slope$ and $\slope+\cL$. Then we have:
	\be
	\Delta_{\infty}( \cL ) =  R^{+}_{w_{m-1}}\cdots R^{+}_{w_{0}}  \Delta_{s}(\cL)
	\ee
\end{Lemma}
\begin{proof}
	By (\ref{conjcop}) the coproducts are related as follows:
	$$
	\Delta_{s} ( \cL ) = (R^{+}_{w_{0}})^{-1}  \cdots (R^{+}_{\infty})^{-1} \Delta_{\infty}(\cL) R^{+}_{\infty}\cdots  R^{+}_{w_{0}}
	$$
	By definition $\Delta_{\infty}(\cL) = \cL \otimes \cL$. In particular,
	$$
	\Delta_{\infty}(\cL) R^{+}_{w_k} \Delta_{\infty}(\cL)^{-1}=R^{+}_{w_k+\cL}=R^{+}_{w_{k+m}}
	$$
	We use this identity to cancel all but finitely many factors in the previous expression. 
\end{proof}

\subsubsection{}
{\rr Assume the torus $\bA$ splits the framing $\bw=\bw' u ' + \bw'' u''$. Let $\Rtot^{\slope}(u)$ be be the corresponding $R$-matrix with slope $s$ acting in the tensor product $K_{\bT}(\cM(\bw')) \otimes K_{\bT}(\cM(\bw''))$.} Let us define the qKZ operator with a slope $\slope$ by
\be
\label{qKZd}
\qKZ^{\slope}=\hbar^{\lambda}_{(1)} T_{u}^{-1} \Rtot^{\slope}(u).
\ee
where $u=u'/u''$.
\begin{Theorem}
	\label{flemma} Let $\slope, \slope'$ be two slopes separated by a single wall $w$, then we have:
	\be
	\label{wallchange}
	W^{-1} \qKZ^{\,\slope} W = \qKZ^{\, \slope'} , \ \ \  W^{-1} \qA_{\cL}^{\slope} W = \qA_{\cL}^{\slope'}
	\ee
	where $W= \Delta_{w}(\repM_{w}(\lambda)) (R^{+}_{w})^{-1}  $ and we assume that passing from $\slope$ to $\slope'$ we cross the wall $w$ in the positive direction.
\end{Theorem}
\begin{proof}
	We have
	$$
	\qKZ^{\slope}=\hbar^{\lambda}_{(1)} T_{u}^{-1} \Rtot^{\slope}(u), \ \ \qKZ^{\slope'}=\hbar^{\lambda}_{(1)} T_{u}^{-1} \Rtot^{\slope'}(u), \ \
	W=  \Delta_{w} (\repM_{w}(\lambda)) (R^{+}_{w})^{-1}.
	$$
	We need to check that $\qKZ^{\slope} W=W \qKZ^{\slope'}$. We have:
	$$
	\begin{array}{l}
	\qKZ^{\slope} W=\hbar^{\lambda}_{(1)} T_{u}^{-1} \Rtot^{\slope}(u)  \Delta_{w} (\repM_{w}(\lambda)) (R^{+}_{w})^{-1}=\\
	\\
	\hbar^{\lambda}_{(1)} T_{u}^{-1}  \Delta^{op}_{w} (\repM_{w}(\lambda)) \Rtot^{\slope}(u)(R^{+}_{w})^{-1}  =\\
	\\ \hbar^{\lambda}_{(1)} T_{u}^{-1} (\Rwal^{-}_w)^{-1}\Rwal^{-}_w \Delta^{op}_{w} (\repM_{w}(\lambda))  \Rtot^{\slope}(u)(R^{+}_{w})^{-1}= \\
	\\
	\hbar^{\lambda}_{(1)} T_{u}^{-1} (\Rwal^{-}_w)^{-1} \Delta_{w} (\repM_{w}(\lambda)) \hbar^{\Omega} R^{-}_{w} \Rtot^{\slope}(u)(R^{+}_{w})^{-1}\stackrel{(\ref{Mint})}{=}\\
	\\
	\Delta_{w} (\repM_{w}(\lambda))(R^{+}_{w})^{-1}  \hbar^{\lambda}_{(1)} T_{u}^{-1}  R^{-}_{w} \Rtot^{\slope}(u)(R^{+}_{w})^{-1}=W \qKZ^{\slope'}
	\end{array}
	$$
	where the last equality uses $ \Rtot^{\slope'}(u)= R^{-}_{w}  \Rtot^{\slope}(u) (R^{+}_{w})^{-1}$ because by assumption $\slope$ and $\slope^{\prime}$
	are separated by a single wall $w$.

	Let $s$ and $s'$ be two slopes separated by a single wall $w_0$.
	We choose a path from slope $\slope$ to $\slope+\cL$ crossing some sequence of walls
	$w_0,w_1...,w_{m-1}$. Similarly, the path from $\slope'$ to $\slope'+\cL$ crosses the walls
	$w_1, w_2...,w_{m}$ with $w_m=w_0+\cL$.
	By Proposition \ref{coprodprop} we have:
	$$
	\begin{array}{l}
	\Delta_s(\qA^{\slope}_{\cL}) = W_{w_0}(\lambda)\cdots W_{w_{m-1}}(\lambda) \, \Delta_{\infty}(\cL)\,  T^{-1}_{\cL}\\
	\\
	\Delta_{s'}(\qA^{\slope'}_{\cL}) = W_{w_1}(\lambda)\cdots W_{w_{m}}(\lambda) \, \Delta_{\infty}(\cL)\,  T^{-1}_{\cL}
	\end{array}
	$$
	To finish the proof of the theorem we need to note that
	$$W_{w_0}(\lambda)^{-1} \Delta_{\slope}(\qA^{\slope}_{\cL}) W_{w_0}(\lambda) = \Delta_{\slope'}(\qA^{\slope'}_{\cL}),$$
	{\rr which follows from an identity obtained by applying $\Delta_w$ to (\ref{dva}).}
	 
\end{proof}

\begin{Theorem} \label{thmAK}
	For arbitrary line bundles $\cL,\cL^{'} \in  \textrm{Pic}(X)$ and a slope $\slope$
	 the qKZ operators (\ref{qKZd}) commute with $q$-difference operators (\ref{qA})
	$$
	\ \Delta_{s}(\qA^{\slope}_{\cL}) \qKZ^{\slope}=\qKZ^{\slope}\Delta_{s}(\qA^{\slope}_{\cL}).
	$$
\end{Theorem}
\begin{proof}
Follows from Proposition~\ref{coprodprop}. Indeed, we obtain
$$
\begin{array}{l}
\qKZ^{s} \Delta_{\slope}(\qA^{\slope}_{\cL}) = \qKZ^{s} W_{w_0}( \lambda ) W_{w_1}( \lambda ) \cdots W_{w_{m-1}}( \lambda ) \Delta_{\infty}(\cL) T_{\cL}^{-1}\stackrel{(\ref{wallchange})}{=}\\
\\
W_{w_0}( \lambda ) W_{w_1}( \lambda ) \cdots W_{w_{m-1}}( \lambda ) \qKZ^{s+\cL}  \Delta_{\infty}(\cL) T_{\cL}^{-1}=
\Delta_{\slope}(\qA^{\slope}_{\cL}) \qKZ^{s}
\end{array}
$$
\end{proof}

\subsection{Identification of $\repM^{\slope}_{\cL}(\lambda)$ and $\geomM_{\cL}(u,\lambda)$}
Our main result is the identification of quantum difference operator $\geomM_{\cL}(\lambda)$
with $\repM^{\slope}_{\cL}(\lambda)$\footnote{\rr In this section we often switch from the variables $z$, denoting K\"ahler parameters, to their logarithms $\lambda$ and back. The two are related via (\ref{hlam}).}. Recall that the quantum difference operators $\geomM_{\cL}(u,\lambda)$
for $\cL \in \Pic(X)$ and the shift operator $\shift(u,\lambda)$ form a compatible system of difference equations (\ref{qdeS}). The Theorem \ref{okth} then identifies the shift operator $\shift(u,\lambda)$ with qKZ operator $\qKZ^{s}$ for some canonical choice of the slope $\slope$. We now generalize this theorem to the case of quantum difference operator:

\subsubsection{}
\begin{Theorem}
	\label{mainth}
	Let $\nabla\subset H^{2}(X,\R)$ be the alcove uniquely defined by the conditions:
	
	1) $0 \in H^{2}(X,\R)$ is one of the vertices of $\nabla$

	2) $\nabla \subset -C_{\textrm{ample}}$ ( opposite of the ample cone)
	
	\noindent then for $\slope \in \nabla$ we have:
	$$
	\begin{array}{l}
Stab^{-1}_{+,T^{1/2},\slope}\,  {\rr \calK} \, Stab_{+,T^{1/2},\slope} = {\rr \qKZ^{s}}\\
	\\
	Stab^{-1}_{+,T^{1/2},\slope}\, \calA_{\cL}\, Stab_{+,T^{1/2},\slope} = \qA^{\slope}_{\cL}
	\end{array}
	$$
{\rr 
where $\calK$ and $\calA_{\cL}$ are the quantum difference operators defined by (\ref{qdiffops}), ${\rr \qKZ^{s}}$ 
is qKZ operator (\ref{qKZd}) 
 and 
$$\qA^{\slope}_{\cL}= Const\cdot \T^{-1}_{\cL} \repM^{\slope}_{\cL}(u,z) $$
for some constant $Const$ and $\cL \in \Pic(X)$.
}	
\end{Theorem}
{\rr Equivalently, up to a multiple, the operator $\textbf{M}_{\cL}(u,z)$ from (\ref{qdeS}) 
coincides with  operator (\ref{dMdef}) for the slope $s$ specified in the above theorem.}

\begin{proof}
Let $\bA=\mathbb{C}^{\times}$ be a torus splitting the framing $\bw=u \bw'+ \bw''$.
We denote the components  of $X^{\bA}$ of a Nakajima variety $X=\cM(\bv,\bw)$  by  $F_{\bv'}=\cM(\bv',\bw')\times \cM(\bv'',\bw'')$. Note that we label them by the weight in the first component.
For a line bundle $\cL$ we have two difference operators acting in $K_{\KG}(\bw')\otimes K_{\KG}(\bw'')$ and commuting with the qKZ operator (\ref{qKZd}). First, by Theorem \ref{okth}:
$$
\calA_{\cL} = T^{-1}_{\cL} \textbf{N}^{s}_{\cL}(u,\lambda)
$$
for $\textbf{N}^{s}_{\cL}(u,\lambda)=\Stab^{-1}_{+,T^{1/2},\slope}\, \geomM_{\cL}(u,\lambda)\, \Stab_{+,T^{1/2},\slope}$ commutes with qKZ operator at the slope $\slope$. Second, by  Theorem \ref{thmAK} the operator:
$$
\qA_{\cL}=T_{\cL}^{-1} \repM^{\slope}_{\cL }(u,\lambda)
$$
commutes with the same qKZ operator {\rr (here by $ \repM^{\slope}_{\cL }(u,\lambda)$ we mean the action of the coproduct $\Delta_s(\repM^{\slope}_{\cL }(u,\lambda))$ in $K_{\KG}(\bw')\otimes K_{\KG}(\bw'')$)}. We want to prove that they coincide up to a constant multiple:
$$
\repM^{\slope}_{\cL}(u,\lambda) =\textbf{N}^{s}_{\cL}(u,\lambda) Const
$$

Both $\textbf{N}_{\cL}(u,\lambda)$ and $\repM^{\slope}_{\cL}(u,\lambda)$ are defined in integral $K$-theory, in particular they and their inverses  are Laurent polynomials in $u$. It follows that the operator:
$$
U(u)=\repM^{\slope}_{\cL}(u,\lambda) \textbf{N}^{-1}_{\cL}(u,\lambda)
$$
is a Laurent polynomial in $u$. By construction, this operator commutes with qKZ  at a slope $\slope$ which means that:
$$
U(u q) = \hbar^{\lambda}_{(1)} \Rtot^{\slope}(u) U(u) \Big(\hbar^{\lambda}_{(1)} \Rtot^{\slope}(u) \Big)^{-1}
$$
From Khoroshkin-Tolstoy factorization for the slope $\slope$ $R$-matrix we obtain:
$$
\Rtot^{\slope}(\infty) = \hbar^{\Omega} \prod\limits_{ 0 \in w  }^{\leftarrow} R^{+}_{w}  \ \ \ \Rtot^{\slope}(0) = \prod\limits_{ 0 \in w  }^{\rightarrow} (R^{-}_{w})^{-1} \, \hbar^{-\Omega} \ \
$$
where $R^{+}_{w}$ and $R^{-}_{w}$ are strictly upper and lower triangular wall $R$-matrices.
The products run over walls passing through $0\in H^{2}(X,\R)$. Therefore, the eigenvalues of conjugation by
$\hbar^{\lambda}_{(1)} \Rtot^{\slope}(u)$ at $u=0,\infty$ are either $1$ or $z^{m}\hbar^{m'}$ with $m\neq0$. Solutions in Laurent series in $u$ thus necessarily correspond to eigenvalue $1$. In particular, they are regular at
$u=0$ and $u=\infty$. It follows that $U$ is a constant matrix in $u$.

The constant matrix $U$ commutes with $\hbar^{\lambda}_{(1)} \Rtot^{\slope}(u)$. Diagonalizing the matrix $\hbar^{\lambda}_{(1)} \Rtot^{\slope}(0)$ we find that $U$ is block upper triangular. Similarly diagonalizing
$\hbar^{\lambda}_{(1)} \Rtot^{\slope}(\infty)$ we find that $U$ is block lower triangular. We conclude that
$U$ is block diagonal.

Let us consider the diagonal block $U_{0,0}$ of the matrix $U$ corresponding to the lowest component of the fixed point set:
$$
U_{0,0}: K_{\KG}(F_{0}) \rightarrow K_{\KG}(F_{0})
$$
Since $U$ commutes with qKZ, the block $U_{0,0}$ commutes with the corresponding block
of the $R$-matrix $\Rtot^{\slope}_{0,0}(u)$. From the definition of the $R$-matrix the matrix element $\Rtot^{\slope}_{0,0}(u)$ is the generating function for operators of classical multiplication by tautological classes on $F_{0}$. {\rr As $K_{\bT}(F_0)$ is generated by tautological classes \cite{McgNev}} the operator
$U_{0,0}$ is itself an operator of multiplication by a K-theory class in $K_{\KG}(F_{0})$. To finish the proof it remains to note that:
\be
\label{dimred}
U_{0,0}= U_{F_{0}}
\ee
Where $U_{F_{0}}$ denotes the same operator $U$ for quiver variety $F_{0}$. Indeed, applying (\ref{dimred})
to $X$ in place of $F_{0}$ we conclude that $U$ is an operator of multiplication in $K_{\KG}(X)$. However,
no such nonscalar operator can be diagonal in the stable basis. We conclude that $U=Const$.
\end{proof}

\subsubsection{}
To finish the proof of the theorem we need to prove (\ref{dimred}). It follows from Propositions \ref{pro1} and
\ref{pro2} below.

\begin{Proposition} \label{pro1}
	The matrix of quantum difference operator $\geomM_{\cL}(0,\lambda)$ has the following form:
	\be
	\label{auxlem1}
	\geomM_{\cL}(0,\lambda)_{\bv_2,\bv_1}=0  \ \ \textrm{for} \ \ \bv_1 \neq 0, \ \  \geomM_{\cL}(0,\lambda)_{0,0}=\left.\geomM_{\cL}(\lambda -\kappa)\right|_{F_{0}}
	\ee
\end{Proposition}

\begin{proof}
	First, let us consider the limit $u\rightarrow 0$ in the quantum difference equation (\ref{qdeS}):
	$$
	\geomM_{\cL}(u,z) \capping(u,\lambda) = \capping(u,z q^{\cL}) \cL
	$$
	First, we have $\cL_{\bv_2,\bv_1} \sim u^{\langle \cL, \bv_2 \rangle}$.
	Second, the matrix of  fundamental solution $\capping(0,\lambda)$ is block upper triangular,
	moreover, the ``vacuum matrix element'' has the form $$\capping(0,\lambda)_{0,0} = \left.\capping\right|_{F_{0}}(\lambda-\kappa)$$
	Thus, we conclude that the operator $\geomM_{\cL}(u,\lambda)$ has the form (\ref{auxlem1}).
	
	{\rr The limit $\capping(0,\lambda)$ in the stable basis exists by (10.2.19) from \cite{pcmilect}.}  The
	upper-triangularity statement follows by inspection of the breaking nodes. Every one of them has the weight of the form
	$(1-q^{m} a^{k})$ and it has to be the case that $k>0$ for all of them for the limit to be non-vanishing. In particular, the curves which contribute to $\capping(0,\lambda)_{0,0}$ never break, therefore, stay entirely within the component $F_{0}$.
	Thus
	$\capping(0,\lambda)_{0,0}=\left.\capping\right|_{F_{0}}(\lambda+...)$. The
	exact form of the shift indicated by dots
	can be computed as the index limit computation for the vertex Section 7.4 in \cite{pcmilect} and gives exactly $\kappa$.
\end{proof}

Let us denote $\repM(u)=\repM^{s}_{\cL}(\lambda)$ for the slope $\slope$ as in the Theorem \ref{mainth} and tautological line bundle $\cL$.
\begin{Proposition} \label{pro2}
	\be
	\label{auxlem2}
	\repM(0)_{\bv_2,\bv_1}=0  \ \ \textrm{for} \ \ \bv_1 \neq 0, \ \  \repM(0)_{0,0}=\left.\repM(\lambda -\kappa)\right|_{F_{0}}
	\ee
\end{Proposition}
\begin{proof}
	First by Proposition \ref{coprodprop}, in the tensor product of two $\Uq(\fgh_Q)$ modules  we have:
	\be
	\label{prodform}
	\repM(u)=W_{w_0}( \lambda ) W_{w_1}( \lambda ) \cdots W_{w_{n-1}}( \lambda ) \Delta_{\infty}(\cL)
	\ee
	where $W_{w}( \lambda )=\Delta_w (\repM_{w} )(R^{+}_w)^{-1}$ and $w_0,\cdots w_{n-1}$ is the ordered set of walls crossed by a straight-line path from $\slope$ to $\slope +\cL$.

	By Corollary \ref{Mcop} we have:
	$$
	\Delta_w(\repM_w(\lambda)) = \wJ^{-}_{w}(\lambda) \,  \Big( \repM_w(\lambda +\hat{\kappa}_{(2)}  ) \otimes \repM_w(\lambda - \hat{\kappa}_{(1)}) \Big) \, \wJ^{+}_w(\lambda)
	$$
	Recall that the operators $\wJ^{-}_{w}(\lambda)$ and $R^{+}_{w}$ are triangular with the following matrix elements:
	$$
	\wJ^{\pm}_{w}(\lambda)=\bigoplus \limits_{{s=0}, \atop { \pm \langle \alpha, \theta \rangle>0}}^{\infty}\, J_{s \alpha} , \ \ \  R^{\pm}_{w}(\lambda)=\bigoplus \limits_{{s=0}, \atop {\pm \langle \alpha,  \theta \rangle>0}}^{\infty}\, R_{s \alpha}
	$$
	where $\theta$ is the stability parameter of the quiver and $\alpha$ is the root defining the wall $w$:
	$$
	w=\{ x \in H^{2}(X,\mathbb{R})| \langle x, \alpha \rangle =m  \}.
	$$
	The matrix elements are of the form:
	$$
	J_{s \alpha}, R_{s \alpha} : K_{\KG}(F_\bv)\longrightarrow K_{\KG}(F_{\bv+s \alpha})
	$$
	and by Theorem \ref{udepth} they have the following dependence on the equivariant parameter $u$:
	$$
	J_{s \alpha},R_{s \alpha} \sim u^{s \langle \alpha, \cL_w \rangle}.
	$$
	where $\cL_w$ is a line bundle on the wall $w$. We conclude that the matrix elements of $W_{w}(\lambda)$ have the following form:
	\be\label{Wudep}
	W_w(\lambda)_{\bv_2,\bv_1} \sim u^{\langle s \alpha , \cL_w\rangle }, \ \  \textrm{if} \ \  \bv_2=\bv_1+s \alpha.
	\ee
	{}From (\ref{prodform}) we see that the matrix element $\repM_{\bv_2,\bv_1}$ has the form:
	$$
	\repM_{\bv_2,\bv_1}=\sum\limits_{s_0,\cdots,s_{n-1}=0}^{\infty}\, \repM_{\bv_2,\bv_1}(s_0,\cdots,s_{n-1})
	$$
	where $\repM_{\bv_2,\bv_1}(s_0,\cdots,s_{n-1})$ is the contribution of the following combination
	of matrix elements:
	
	$$
	\begin{array}{l}
	\repM_{\bv_2,\bv_1}(s_0,\cdots,s_{n-1}):\,\\
	\\
{\rr	K_{\KG}(F_{\bv_1}) {\stackrel{\Delta_{\infty}(\cL)}{\longrightarrow}}} K_{\KG}(F_{\bv_1}) {\stackrel{W_{w_{n-1}}(\lambda)}{\longrightarrow}} K_{\KG}(F_{\bv_1+s_{n-1} \alpha_{n-1}})
	{\stackrel{W_{w_{n-2}}(\lambda)}{\longrightarrow}} K_{\KG}(F_{\bv_1+s_{n-1} \alpha_{n-1}+s_{n-2} \alpha_{n-2}})  \\
	\\
{\stackrel{W_{w_{n-3}}(\lambda)}{\longrightarrow}} 	\cdots {\stackrel{W_{w_0}(\lambda)}{\longrightarrow}} K_{\KG}(F_{\bv_2})
	\end{array}
	$$
	such that
	\be
	\label{conslaw}
	s_0 \alpha_0+\cdots+s_{n-1} \alpha_{n-1} = \bv_2-\bv_1
	\ee
	From (\ref{Wudep}) we see that this matrix element has the following dependence on the spectral parameter:
	$
	\repM_{\bv_2,\bv_1}(s_0,\cdots,s_{n-1}) \sim u^{d_{\bv_2,\bv_1}(s_0,\cdots,s_{n-1})},
	$
	with exponent:
	\be
	d_{\bv_2,\bv_1}(s_0,...,s_{n-1})=s_0 \langle \alpha_0 , \cL_0 \rangle+...+s_{n-1} \langle \alpha_{n-1} , \cL_{n-1} \rangle +
	\langle \bv_1 , \cL_n \rangle
	\ee
	where we denote by $\cL_i$ the point at which the straight-line path $(\slope, \slope +\cL)$ intersects the wall $w_i$ and $\cL_n=\cL$. The last term $\langle \bv_1 , \cL_n \rangle$ comes from $\Delta_{\infty}(\cL)$ which is a diagonal operator with diagonal matrix elements
	$\Delta_{\infty}(\cL)_{\bv_1,\bv_1}\sim u^{\langle \bv_1, \cL \rangle}$.


	By our choice, we can assume that the slope $\slope$ lies in an arbitrarily small neighborhood
	of $0 \in H^{2}(X,\R)$. Thus we can assume that $\cL_0=0$ and write:
	\be
	\begin{array}{l}
		d_{\bv_2,\bv_1}(s_0,...,s_{n-1})=\\
		\\
		s_1 \langle \alpha_1 , \cL_1 -\cL_0 \rangle+...
		+s_{n-1} \langle \alpha_{n-1} , \cL_{n-1}-\cL_0 \rangle +\langle \bv_1 , \cL_n -\cL_{0}\rangle
	\end{array}
	\ee
	We rewrite this equality in the following form:
	\be
	\label{mneq}
	\begin{array}{l}
		d_{\bv_2,\bv_1}(s_0,...,s_{n-1})=\\
		\\
		\langle \bv_1, \cL_n-\cL_{n-1} \rangle+\\
		\\
		\langle \bv_1+s_{n-1} \alpha_{n-1}, \cL_{n-1}-\cL_{n-2} \rangle+\\
		\\
		+\cdots + \\
		\\
		\langle  \bv_1+s_{n-1} \alpha_{n-1} +\cdots+s_{1} \alpha_{1}, \cL_1-\cL_0 \rangle
	\end{array}
	\ee
	Now, we have the set of inequalities:
	\be
	\begin{array}{l}
		\bv_1+s_{n-1} \alpha_{n-1} \geq 0\\
		\\
		\bv_1+s_{n-1} \alpha_{n-1}+s_{n-2} \alpha_{n-2}  \geq 0\\
		\\
		\cdots \\
		\\
		\bv_1+s_{n-1} \alpha_{n-1}+\cdots +s_{1} \alpha_{1}\geq 0
	\end{array}
	\ee
	where $\bv \geq 0$ means that the inequality holds for all components of the dimension vector:  $\bv_i\geq 0$.
	If they are not satisfied, the matrix element  $\repM_{\bv_2,\bv_1}(s_0,s_1,...,s_{n-1})$ vanishes
	as the corresponding operator annihilates any class supported on component $F_{\bv_1}$.
	
{\rr By construction of $\cL_i$ we have $\langle \bv, \cL_{i}-\cL_{i-1} \rangle\geq 0$ for $\bv\geq 0$ and  $\langle \bv, \cL_{i}-\cL_{i-1} \rangle>0$ for $\bv>0$. We conclude that for $\bv_1>0$}	
$$
	d_{\bv_2,\bv_1}(s_0,...,s_{n-1}) \geq \langle \bv_1, \cL_n-\cL_{n-1} \rangle>0
$$
and therefore
	$$
	\lim\limits_{u\rightarrow 0} \repM_{\bv_2,\bv_1} = 0 \ \  \textrm{for} \ \ \bv_1 \neq 0.
	$$
Next, let us analyze the case $\bv_2=\bv_1=0$. Substituting
	$\bv_1=0$ into (\ref{mneq}) we see that $d_{\bv_2,\bv_1}(s_0,...,s_{n-1})=0$ only when
	$s_1=s_2=\cdots=s_{n-1}=0$. Thus, from (\ref{conslaw}) we conclude: $s_0 \alpha_0 = \bv_2=0$,
	so that $s_0=0$. It means that only the diagonal matrix elements (all $s_i=0$) of
	$W_{w_k}(\lambda)$ contribute to the vacuum matrix element $\repM(u)_{0,0}$. From (\ref{prodform}) we obtain:
	$$
	\repM(0)_{0,0}=\repM_{w_0}(\lambda-\kappa) \cdots \repM_{w_{n-1}}(\lambda-\kappa) \cL  = \left.\repM^{\slope}_{\cL^{-1}}\right|_{F_{0}}(\lambda-\kappa)
	$$
	The proposition is proven.
\end{proof}

\subsubsection{\label{corsec}}

\begin{Corollary} \label{cor3}
	The operator $\repM^{\slope}_{\cL}(\lambda)$ does not depend on the choice of path made in (\ref{dMdef}).
\end{Corollary}
\begin{proof}
	Let $\repM^{\slope}_{\cL}(\lambda)$ and $\repM^{\slope}_{\cL}(\lambda)^{'}$ be two elements given by formula (\ref{dMdef}) corresponding to different choices of
	a path from $\slope$ to $\slope+\cL$. Assume that the slope $\slope$ belongs to the anti-fundamental alcove $\nabla \subset -C_{\textrm{ample}}$ 
	as in the theorem above. By  Theorem \ref{mainth}, $D=\repM^{\slope}_{\cL}(\lambda)^{'} \repM^{\slope}_{\cL}(\lambda)^{-1}$
	is a constant.  Recall that the wall operators $\repM_w(\lambda)$ are normally ordered (see Section \ref{Boper}).  It means that for a component of minimal weight $\gamma$  we have $\repM_w(\lambda) \gamma =\gamma$. Thus $D(\gamma)=\gamma$ and the constant is $1$. Finally, by Theorem \ref{flemma} this statement holds true for arbitrary slope. 
\end{proof}

\begin{Corollary}
	For arbitrary line bundles $\cL,\cL' \in \Pic(X)$ and  slopes $\slope \in H^{2}(X,\R)$ the corresponding
	$q$-difference operators commute:
	$$
	\qA^{\slope}_{\cL} \qA^{\slope}_{\cL'} = \qA^{\slope}_{\cL'} \qA^{\slope}_{\cL}
	$$
\end{Corollary} 
\begin{proof}
	By Proposition \ref{ProM} $\qA^{\slope}_{\cL} \qA^{\slope}_{\cL'}$ and $\qA^{\slope}_{\cL'} \qA^{\slope}_{\cL}$ 
	give an operator $\qA^{\slope}_{\cL+\cL'}$ with two different choices of a path for $\repM^{\slope}_{\cL+\cL'}(\lambda)$.
	The result is independent on the choice of a path  by  Corollary~\ref{cor3}.   
\end{proof}

\section{Cotangent bundles to Grassmannians \label{s_ExamplesG} }
In this section we consider the simplest quiver, which consists of one vertex. 
In this case the dimension vectors are given by a couple of natural numbers $(\bv,\bw)=(k,n) \in \N^2$, and the corresponding varieties
are isomorphic to cotangent bundles to Grassmannians  of $k$ -dimensional subspaces in $n$-dimensional space:
\be
\cM(\bv,\bw)=T^{*} Gr(k,n)
\ee
The framing torus $\bA\simeq({\C^{\times}})^{n}$ acts on $W=\C^n$ in a standard way. This induces an action
of $\bA$ on $T^{*} Gr(k,n)$. Note that this action preserves the symplectic form on $T^{*} Gr(k,n)$.  Let us denote by  $\KG=\bA\times {\C^{\times}} $  where the extra factor acts by scaling the fibers of the cotangent bundle. This torus scales the symplectic form with character which we denote $\hbar$.

\subsection{Algebra $\Uq(\fgh_{Q})$ and wall subalgebras $\Uq(\fg_w)$}

\subsubsection{}
Let us denote
\be
\label{wholeX}
X= \cM(\bw)=\coprod\limits_{\bv}\, \cM(\bv,\bw)=\coprod\limits_{k=0}^{n}\, T^{*} Gr(k,n)
\ee
Note that $\cM(1)$ is a variety consisting of two points, thus $K_{\KG}(\cM(1))$ is two dimensional over $K_{\KG}(pt)$.
Therefore, if the torus $\bA$ splits the framing as $\bw=u_1+\cdots +u_n$ then we have:
\be
\label{tenpr}
K_{\KG}(X) = \C^{2}(u_1) \otimes \cdots \otimes \C^{2}(u_n)
\ee
so that the total dimension is $2^n$. Note that $T^{*} Gr(k,n)^{\bA}$ consists of $n!/k!/(n-k)!$ points, such that $X^{\bA}$ is a set of $2^n$ points $p_i$. The fixed point basis of (localized) $K_{\KG}(X)$ consists of sheaves ${\cal{O}}_{p_i}$.

\subsubsection{}
We start from the case $n=2$. We have:
$$
X=\textrm{pt}\, \cup \,T^*\mathbb{P}^{1} \,\cup\, \textrm{pt}
$$
where $\textrm{pt}$ stands for a Nakajima variety consisting of one point. Therefore, the only nontrivial block of the $R$-matrix corresponds to $T^*\mathbb{P}^{1}$.
The action of torus $\KG=\bA \times \C^{\times}$ is represented in Fig. \ref{picTP}. In this picture $p_1$ and $p_2$ are two fixed points, corresponding to the points $z=0$ and $z=\infty$ of the base $\mathbb{P}^{1}\subset T^*\mathbb{P}^{1}$. We also specify explicitly the characters of the tangent spaces to $T^*\mathbb{P}^{1}$ at the fixed points. For example the tangent space at $p_1$ is spanned by the tangent space to the base with character $u_1/u_2$ and
the tangent space to the cotangent fiber with character $u_2/(u_1\hbar)$.

To compute the stable envelopes of the fixed points we need to choose a polarization $T^{1/2}$ and a chamber $\fC$.
We choose the positive chamber $\fC$ such that $u_1/u_2\to 0$. The arrows in Fig.\ref{picTP} represent the attracting and repelling directions with respect to this chamber. We choose a polarization $T^{1/2}$ given by the cotangent directions.

We have  $H^{2}(T^*\mathbb{P}^{1},\R)=\R$, thus we identify the set of slopes with real numbers
$\slope \in \R$.

\begin{figure}[h]
	\begin{center}
\includegraphics[scale=0.25]{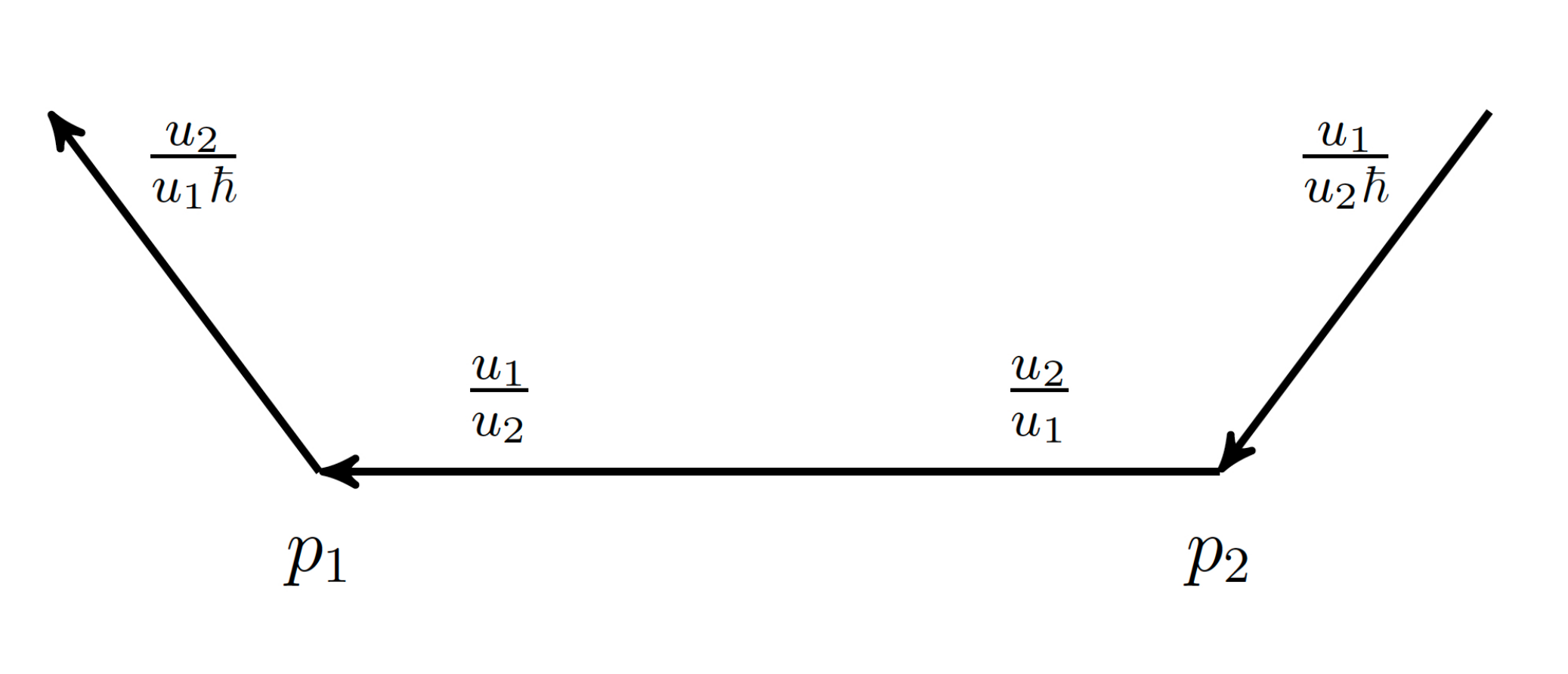}
\end{center}

\caption{Toric representation of $T^{*} {\mathbb{P}^1}$. Arrows represent the repelling and attractive directions
		with respect to the chamber $\fC= u_1/u_2 \to 0$.   \label{picTP}}
\end{figure}

\subsubsection{\label{stabbassec}}
Let us consider the restrictions of the stable envelopes to the fixed components. By (\ref{Stab_norm}) we have:
$$
\left. \textrm{Stab}_{\fC,T^{1/2}, s}(p)\right|_{p}= (-1)^{\textrm{rk} T^{1/2}_{>0}} \Big( \dfrac{\det \cN_{-}}{\det T^{1/2}_{\neq 0}} \Big)^{\frac{1}{2}}\, \Ldm \cN_{-}^\vee
$$
By definition $\cN_{-}$ is the repelling part of the normal bundle to $p$, $T^{1/2}_{>0}$ is the attracting part of the polarization and $T^{1/2}_{\neq 0}$ is the non-stationary part of the polarization.

{\bb 

From the Fig.\ref{picTP} at $p_2$ we have  $\cN_{-}=u_2/u_1$, $\textrm{rk} T^{1/2}_{>0}=1$, $T^{1/2}_{\neq 0}=u_1/(u_2\hbar)$. Thus we find:
\be \label{stabres}
\left. \textrm{Stab}_{\fC,T^{1/2}, s}(p_2)\right|_{p_2}=(1-u_2/u_1) \hbar^{1/2}
\ee
The support condition for the stable envelopes gives $\left. \textrm{Stab}_{\fC,T^{1/2}, s}(p_2)\right|_{p_1}=0$. The unique K-theory class with these restrictions at the fixed points equals
$$
\begin{array}{|c|}
\hline\\
\ \  \textrm{Stab}_{\fC,T^{1/2}, s}(p_2)=(1-{\cal{O}}(1)/u_1) \sqrt{h} \ \ \\ 
\\
\hline
\end{array} 
$$
where ${\cal{O}}(1)$ is the tautological bundle restricting to the fixed points by the rule $\left.{\cal{O}}(1)\right|_{p_i}=u_i$

Next, from Fig.\ref{picTP} at $p_1$ we find:
$\cN_{-}=u_2/u_1/\hbar$, $\textrm{rk} T^{1/2}_{>0}=0$, $T^{1/2}_{\neq 0}=\cN_{-}=u_2/u_1/\hbar$. Thus:
\be \label{p1rest}
\left. \textrm{Stab}_{\fC,T^{1/2}, s}(p_1)\right|_{p_1}=1-\hbar u_1/u_2
\ee
The fractional line bundle corresponding to slope $s$ is ${\cal{O}}(1)^{s}$.
The degree condition (\ref{degK}) for the point $p_1$ gives:
$$
\deg_{\bA}\Big( \left. \textrm{Stab}_{\fC,T^{1/2}, s}(p_1)\right|_{p_2}\Big) \subset
\deg_{\bA}\Big( \left. \textrm{Stab}_{\fC,T^{1/2}, s}(p_2)\right|_{p_2} \times  \dfrac{\left.{\cal{O}}(1)^{s}\right|_{p_2}}{\left.{\cal{O}}(1)^{s}\right|_{p_1}}\Big)
$$
Thus by (\ref{stabres}):
$$
\deg_{\bA}\Big( \left. \textrm{Stab}_{\fC,T^{1/2}, s}(p_1)\right|_{p_2}\Big) \subset
\deg_{\bA}\Big( (1-u_2/u_1) \sqrt{h} \, (u_2/u_1)^{s} \Big)=(s,s+1)
$$
For generic $s$ this condition implies that $\left.\textrm{Stab}_{\fC,T^{1/2}, s}(p_1)\right|_{p_2}$ is a monomial
\be \label{monres}
\left.\textrm{Stab}_{\fC,T^{1/2}, s}(p_1)\right|_{p_2}=c(\hbar) (u_2/u_1)^{\lfloor s+1 \rfloor}
\ee
with unknown coefficient $c(\hbar)$. 

The points $p_1$ and $p_2$ are connected by an equivariant $\mathbb{P}^1$ with weights of the tangent spaces given by $(u_1/u_2)^{\pm 1}$. This means that for any equivariant K-theory class $F$, we have
$
\left.F\right|_{p_1} = \left.F\right|_{p_2}$ at $u_1/u_2=1$.
Applying this to $F=\textrm{Stab}_{\fC,T^{1/2}, s}(p_1)$, from (\ref{p1rest}) and (\ref{monres}) we obtain
$$
c(\hbar)=1-\hbar
$$
We conclude that
\be \label{res2}
\left.\textrm{Stab}_{\fC,T^{1/2}, s}(p_1)\right|_{p_2}=(1-\hbar ) (u_2/u_1)^{\lfloor s+1 \rfloor}
\ee
The unique K-theory class which has  restrictions (\ref{p1rest}) and (\ref{res2}) equals

$$
\begin{array}{|c|}
\hline\\
\textrm{Stab}_{\fC,T^{1/2}, s}(p_1)=(1-\hbar {\cal{O}}(1)/u_2) \left(\dfrac{{\cal{O}}(1)}{u_1}\right)^{\lfloor 1+s\rfloor}\\
\\
\hline
\end{array}
$$
}

\subsubsection{}
For the opposite chamber $-\fC$ we have $u_1/u_2 \to \infty$. It means that in Fig.~\ref{picTP} all arrows are reversed. In particular the stable envelope for $-\fC$ is obtained from the last formula by
permuting the fixed points:
\be
\begin{array}{|l|}
	\hline\\
	\textrm{Stab}_{-\fC,T^{1/2}, s}(p_1)=(1-{\cal{O}}(1)/u_2) \sqrt{h} \\
	\\
	\textrm{Stab}_{-\fC,T^{1/2}, s}(p_2)=(1-\hbar {\cal{O}}(1)/u_1) \left(\dfrac{{\cal{O}}(1)}{u_2}\right)^{\lfloor 1+s\rfloor}\\
	\\
	\hline
\end{array}
\ee
\subsubsection{}
In agreement with our general theory we see that the stable envelopes are locally constant functions of
the parameter $s$. From the last set of formulas we see that it changes only when $s$ crosses an integer point.
We conclude that the set of walls can be identified with $\Z \subset \R$ and thus alcoves are of the form $(w,w+1) \subset \R$.

The alcove specified by Theorem \ref{mainth} has the form $\nabla=(-1,0)$. To compute the $R$-matrix corresponding to
this alcove we choose $\slope \in \nabla$, then in the basis of fixed points ordered as $[p_2,p_1]$, from the above formulas we compute:
\be \label{restr1}
i^{*} \textrm{Stab}_{\fC,T^{1/2}, s}=
\left[ \begin {array}{cc}  \left( 1-{u}^{-1} \right) \sqrt {\hbar}&1-\hbar
\\ \noalign{\medskip}0&1-\hbar u\end {array} \right]\\
i^{*} \textrm{Stab}_{-\fC,T^{1/2}, s}=
\left[ \begin {array}{cc} 1-\hbar u^{-1}&0\\ \noalign{\medskip}1-\hbar&
\left( 1-u \right) \sqrt {\hbar}\end {array} \right] \label{restr2}
\ee
where we denote $u=u_1/u_2$ and $i^{*}$ is the operation of restriction to fixed points. The total $R$-matrix
for slope $\slope$ is defined as follows:
$$
\Rtot^{s}(u)=\textrm{Stab}_{-\fC,T^{1/2}, s}^{-1} \textrm{Stab}_{\fC,T^{1/2}, s}=(i^{*} \textrm{Stab}_{-\fC,T^{1/2}, s})^{-1} \ (i^{*}  \textrm{Stab}_{\fC,T^{1/2}, s})
$$
and we obtain:
\be
\label{RtotTP}
\begin{array}{|l|}
	\hline \\
	\Rtot^{s}(u)=\left[ \begin {array}{cc} {\dfrac { \left( 1-u \right) \hbar^{\frac{1}{2}}}{\hbar-u
	}}&{\dfrac {u \left( \hbar-1 \right) }{\hbar-u}}\\ \noalign{\medskip}{\dfrac {\hbar-1}{\hbar-u}}&{\dfrac { \left( 1-u \right) \hbar^{\frac{1}{2}}}{\hbar-u}}\end {array}
	\right]\\
	\\
	\hline
\end{array}
\ee
\subsubsection{}
The wall $R$-matrices  are defined by (\ref{rootR}) and similarly to what we have above:
$$
R^{\pm}_{w}=(i^{*} \textrm{Stab}_{\pm \fC,T^{1/2}, s'})^{-1} \ (i^{*}  \textrm{Stab}_{\pm \fC,T^{1/2}, s})
$$
where $s$ and $s'$ are two slopes separated by a wall $w$. Let $w$ be an integer representing the wall and $s=w-\epsilon$, $s'=w+\epsilon$
for sufficiently small $\epsilon$ (obviously enough  to take $0<\epsilon<1$). Then from the above formulas we obtain:
\be
\label{walRTP}
\begin{array}{|c|}
	\hline \\
	R^{+}_{w}=\left[ \begin {array}{cc} 1&{\dfrac {1-\hbar}{{u}^{w}\sqrt {\hbar}}}
	\\ \noalign{\medskip}0&1\end {array} \right] \ \ \,R^{-}_{w}= \left[ \begin {array}{cc} 1&0\\ \noalign{\medskip}{\dfrac {{u}^{w}
			\left( 1-\hbar \right) }{\sqrt {\hbar}}}&1\end {array} \right]\\
	\\
	\hline
\end{array}
\ee
{\rr Observe that these matrices are related by transposition as in (\ref{rtransp}).}

\subsubsection{ \label{ktexamtp}}
The KT factorization of $R$-matrix $\slope \in \nabla$ has the form (\ref{Rfac}):
\be
\Rtot^{s}(u)= \prod\limits_{w<0}^{\rightarrow} \, R^{-}_{w} R_{\infty} \prod\limits_{w \geq 0}^{\leftarrow} \, R^{+}_{w}
\ee
This infinite product is convergent in the topology of power series in $u^{-1}$. From (\ref{walRTP}) we obtain:
$$
U=\prod\limits_{w \geq 0}^{\leftarrow} \, R^{+}_{w} =  \cdots R^{+}_{1} R^{+}_{0}=
\left[ \begin {array}{cc} 1&{\dfrac {1-\hbar}{\sqrt {\hbar}}(1+u^{-1}+\cdots)}
\\ \noalign{\medskip}0&1\end {array} \right]=\left[ \begin {array}{cc} 1&{\dfrac {(1-\hbar)u}{\sqrt {\hbar}\,(u-1)}}
\\ \noalign{\medskip}0&1\end {array} \right]
$$
$$
L= \prod\limits_{w<0}^{\rightarrow} \, R^{-}_{w}=R^{-}_{-1} R^{-}_{-2}\cdots=\left[ \begin {array}{cc} 1&0\\ \noalign{\medskip}{\dfrac {
		\left( 1-\hbar \right) }{\sqrt {\hbar}}}(u^{-1}+\cdots)&1\end {array} \right]=
\left[ \begin {array}{cc} 1&0\\ \noalign{\medskip}{\dfrac {
		\left( 1-\hbar \right) }{\sqrt {\hbar}\,(u-1)}}&1\end {array} \right]
$$
Finally, the infinity slope $R$-matrix is given by (\ref{infr}). The attracting and repelling directions are obvious from Fig. \ref{picTP} and we obtain:
$$
R_{\infty}=
-\left[\begin{array}{cc}
\dfrac{u^{-\frac{1}{2}}-u^{\frac{1}{2}}}{u^{\frac{1}{2}}\hbar^{-\frac{1}{2}}-u^{-\frac{1}{2}}\hbar^{\frac{1}{2}}}&0\\
0&\dfrac{u^{-\frac{1}{2}}\hbar^{-\frac{1}{2}}-u^{\frac{1}{2}}\hbar^{\frac{1}{2}}}{u^{\frac{1}{2}}-u^{-\frac{1}{2}}}
\end{array}\right]
$$
One easily checks that in agreement with (\ref{RtotTP}) we have $\Rtot^{\slope}(u)=L \, R_{\infty}\, U$. This gives canonical LU decomposition of the $R$-matrix.

\subsubsection{}
The $R$-matrix for the whole Nakajima variety $X$ given by (\ref{wholeX}) is of the form:
$$
\Rtot^{\slope}(u)=
\left[\begin{array}{cccc}
1&&&\\
&\Rtot^{\slope}_{T^{*}{\mathbb{P}^1}}&&\\
&&&1
\end{array}\right] =
\left[\begin{array}{cccc}
1&0&0&0\\
0&{\dfrac { ( 1-u ) \hbar^{\frac{1}{2}}}{\hbar-u}}&{\dfrac {u ( \hbar-1 ) }{\hbar-u}}&0\\
0&{\dfrac {\hbar-1}{\hbar-u}}&{\dfrac { ( 1-u ) \hbar^{\frac{1}{2}}}{\hbar-u}}&0\\
0&0&0&1
\end{array}\right]
$$
Up to a scalar multiple one recognizes the standard $R$-matrix for ${{\mathscr{U}_{\rr \sqrt{\hbar}}(\widehat{\frak{gl}}_2)}}$ acting in the tensor product
of two fundamental evaluation modules $\C^2(u_1)\otimes \C^2(u_2)$. We conclude that the quiver algebra
corresponding to cotangent bundles to Grassmannians is $\Uq(\fgh_{Q})={{\mathscr{U}_{{\rr \sqrt{\hbar}}}(\widehat{\frak{gl}}_2)}}$.

\subsubsection{}
The codimension function (\ref{codfun}) for $X$ is given, obviously, by the following diagonal matrix:
$$
\hbar^{\Omega}=\textrm{diag}(1,\hbar^{\frac{1}{2}},\hbar^{\frac{1}{2}},1)
$$
We obtain that the wall $R$-matrices defined by the Theorem \ref{thmR} have the following explicit form:
$$
\Rwal^{+}_{w}=\left[\begin{array}{cccc}
1&0&0&0\\
0&\hbar^{\frac{1}{2}}&(1-\hbar)u^{-w}&0\\
0&0&\hbar^{\frac{1}{2}}&0\\
0&0&0&1
\end{array}\right]
$$
In particular all wall $R$-matrices are conjugated to the zeroth one by a line bundle:
\be
\label{cojnR}
\Rwal^{+}_{w} = {\cal{O}}(w)  \Rwal^{+}_{0}   {\cal{O}}(w)^{-1}
\ee
with ${\cal{O}}(w)=\textrm{diag}(1, u_2^{w},u_1^w, 1)$. One recognizes that up to a multiple
$\Rwal^{+}_{0}$ coincides with the standard $R$-matrix for $\mathscr{U}_{{\rr \sqrt{\hbar}}}(\frak{sl}_2)$ in the tensor product of two
fundamental representations. Thus, the wall subalgebra, which is built by FRT procedure from this $R$-matrix
is $\Uq(\fg_0)\simeq \mathscr{U}_{{\rr \sqrt{\hbar}}}(\frak{sl}_2)$. As the $R$-matrices for other walls are conjugates of $\Rwal^{+}_{0}$, we conclude that
$\Uq(\fg_w)\simeq \mathscr{U}_{{\rr \sqrt{\hbar}}}(\frak{sl}_2)$ for arbitrary wall $w$.

\subsubsection{\label{subshsec}}
{\rr To get rid of the square roots it is convenient to change the notations $\hbar\to \hbar^{2}$, which we assume starting from here and to the end of this section.} With this notation we have the algebra $\Uq(\fgh_{Q})=\sldh$ and  a set of subalgebras $\Uq(\fg_w) \simeq \Uq(\frak{sl}_2)$
indexed by walls $w\in \Z$. It is convenient to organize this data as follows: let $E$, $F$ and $K$ be the standard generators of $\Uq(\frak{sl}_2)$ which we understand as $\Uq(\fg_0)$. Then by (\ref{cojnR})
the wall subalgebra $\Uq(\fg_w)$ is generated by $E_{w}$, $F_{w}$ and $K$:
\be
\label{wallsh}
E_{w}={\cal{O}}(w) E {\cal{O}}(w)^{-1}, \ \  F_{w}= {\cal{O}}(w) F {\cal{O}}(w)^{-1}.
\ee
Let us denote $x^{+}(w) =E_{w} , \ \  x^{-}(w)=F_{-w}$. One can check that the relations among these generators can be summarized as the Drinfeld's realization of  $\sldh$:  the algebra $\sldh$  is an associative algebra with $1$ generated over $\C(\hbar)$ by the elements
$x^{\pm}(k), a(l), K^{\pm 1}$ ($ k\in \Z, l\in \Z\setminus \{0\} $) with the following relations:
\be
\begin{array}{l}
	K  K ^{-1}= K^{-1} K =1\\
	\\
	{[}a {(k)},a {(m)}{]}=0, {[}a {(k)},K^{\pm}{]}=0\\
	\\
	K x^{\pm } (k) K^{-1}= \hbar^{\pm 2} x^{\pm } (k)\\
	\\
	{[}x^{+} (k),x^{-} (l){]} = \dfrac{1}{\hbar - \hbar^{-1}} \Big(\psi (k + l) -\varphi (k + l)\Big)\\
	\\
	{[}a (k), x^{\pm} (l){]} =\pm \dfrac{{[}2 k{]}_{\hbar}}{k} x^{\pm} (l+k)
\end{array}
\ee
with
$$
\begin{array}{l}
\sum\limits_{m=0}^{\infty} \,\psi(m) z^{-m}  =K \exp\left( (\hbar-\hbar^{-1}) \sum\limits_{k=1}^{\infty}\, a(k) z^{-k}  \right)\\
\\
\sum\limits_{m=0}^{\infty} \,\varphi(-m) z^{m}  =K^{-1} \exp\left( -(\hbar-\hbar^{-1}) \sum\limits_{k=1}^{\infty}\, a(-k) z^{k} \right)
\end{array}
$$
{\rr and $\hbar$-number $[n]_{\hbar}:=(\hbar^{n}-\hbar^{-n})/(\hbar-\hbar^{-1})$.}

It may be convenient to visualize $\sldh$ and its subalgebras as in the Figure~\ref{pic1} :  the wall $\Uq(\fg_{w})$ corresponds to a line with integer slope $w$.

\begin{figure}[h]
\begin{center}
	\includegraphics[scale=0.35]{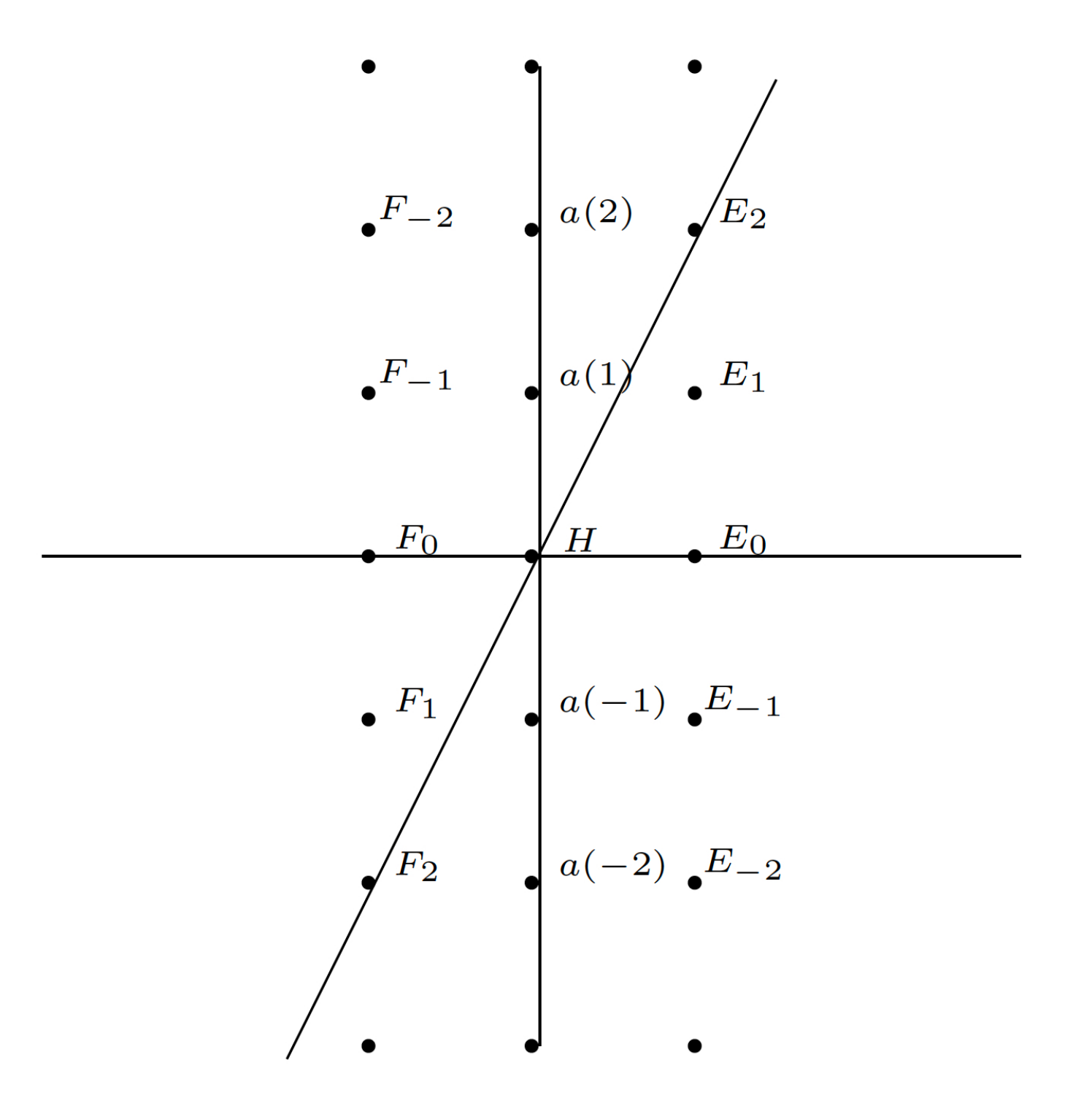}
	\end{center}

	\caption{The structure of $\sldh$. The line through zero corresponds to the slope $2$ subalgebra
		$ U_{q}({\frak{sl}}_2) \subset \sldh $  generated by $E_2, F_{2}, K$.  \label{pic1}}
\end{figure}

\subsection{R -matrices \label{UnivRmatri}}
\subsubsection{}
To write the formulas for $R$-matrices for a general variety (\ref{wholeX}) it is enough to substitute all formulas from the previous section by their ``universal'' versions.

The universal $R$-matrix for $\Uq(\frak{sl}_{2})$ is well known:
\be
\label{sl2R}
\Rwal=\hbar^{-H \otimes H/2} \, \sum\limits_{k=0}^{\infty}  \dfrac{(-1)^k (\hbar-\hbar^{-1})^k \hbar^{-k(k-1)/2}}{[k]_\hbar!}  F^k  \otimes E^k
\ee
with $[k]_{\hbar}! = [1]_\hbar [2]_\hbar \dots [k]_\hbar$ and $H$ related to $K$ as $K=\hbar^{H}$.  Up to a scalar multiple the codimension function is given by\footnote{\rr Note the substitution $\hbar\to \hbar^{2}$ on the left side of this equality which was introduced at the begining of Section \ref{subshsec}. We have $\hbar^{2 \Omega}=\diag(1,\hbar,\hbar,1)$ and $\hbar^{-H \otimes H/2}=\diag(\hbar^{-1/2},\hbar^{1/2},\hbar^{1/2},\hbar^{-1/2})=\hbar^{-1/2}\hbar^{2 \Omega}$.}
$\hbar^{2 \Omega}=\hbar^{-H \otimes H/2}$ thus, we conclude that there is the following universal formula for the wall $R$-matrices:
\be
\label{sl2w}
R^{+}_{w}=\sum\limits_{k=0}^{\infty}  \dfrac{(-1)^k (\hbar-\hbar^{-1})^{k} \hbar^{-k(k-1)/2}}{[k]_\hbar!}  F_w^k  \otimes E_w^k, \ \
\ee
The lower triangular wall $R$-matrix is obtained by transposition $R^{-}_{w} = (R^{+}_{w})_{21}$:

\be
\label{sl2low}
R^{-}_{w}=\sum\limits_{k=0}^{\infty}  \dfrac{(-1)^k (\hbar-\hbar^{-1})^{k} \hbar^{-k(k-1)/2}}{[k]_\hbar!}  E_w^k  \otimes F_w^k, \ \
\ee

\subsubsection{}
The KT factorization (\ref{Rfac}) provides the following universal formula for the total $R$-matrix:
\be
\label{fct}
\Rtot^{\slope}(u)=\prod\limits_{w<\slope}^{\longleftarrow} R^{-}_{w} \,R_{\infty}\, \prod\limits_{w\geq \slope}^{\longleftarrow} R^{+}_{w}
\ee
with $R^{\pm}_{w}$ with given explicitly by $(\ref{sl2w})$. The  $R$ -matrix $R_{\infty}$ is the operator of  multiplication by the class of normal bundles (\ref{infr}). It can be conveniently expressed in terms of generators $a(n)$ corresponding to the infinite slope
in the Fig\ref{pic1}:
$$
R_{\infty}=c \, \hbar^{H \otimes H/2}\, \exp\Big( (\hbar-\hbar^{-1}) \sum\limits_{n=1}^{\infty} \, \frac{n}{[2 n]_\hbar} \,a(-n) \otimes a(n)   \Big)
$$
where $c$ is some scalar multiple depending on normalization.

\subsection{The quantum difference operator $\geomM_{\cL}(z)$}
\subsubsection{}
By definition $\hbar^{\lambda}$ acts on $K$-theory of $\cM(1)=\cM(1,1) \coprod \cM(0,1)$  as:
$$
\hbar^{\lambda} =\left\{
\begin{array}{cc}
z & \textrm{on} \ \ \cM(1,1) \\
1   & \textrm{on}  \ \ \cM(0,1)
\end{array}\right. \ \  \Leftrightarrow \ \
\hbar^{\lambda} =
\left(
\begin{array}{cc}
z & 0 \\
0   & 1
\end{array}\right) = z^{\frac{1}{2}}\, z^{H/2}
$$
From this and (\ref{sl2w}) we see that the ABRR equation  for $\Uq(\frak{sl}_2)$ takes the following form:
$$
J^{+}(z) z^{-H\otimes 1/2}\, \Rwal \, =\, z^{-H\otimes 1/2}\, \hbar^{-H\otimes H/2}  J^{+}(z)
$$
with $\Rwal$ given by (\ref{sl2R}). This is an equation for strictly upper triangular operator $J(z)$, which means that:
$$
J^{+}(z)=1+\sum\limits_{k=1}^{\infty}\, J^{+}_k (z)\, F^k \otimes E^k
$$
The Proposition \ref{Jtheorem} says that the ABRR equation determines the coefficients $J_k (z)$ uniquely. Computation gives:
$$
J^{+}(z)=\sum\limits_{k=0}^{\infty} \,\dfrac{(-1)^k \hbar^{-k(k-1)/2} (\hbar-\hbar^{-1})^k}{[k]_\hbar! \prod\limits_{i=1}^{k}( 1- z^{-1} K \otimes K^{-1} \hbar^{2 i} )} \, F^{k}\otimes E^{k}
$$

\subsubsection{}
By definition (\ref{JJ}) we have $\textbf{J}^{+}_{w}(\lambda)=J^{+}_{w}(\lambda-\tau_w)$. In our case $\tau_w=\textbf{s} w$  and this corresponds to a shift $z\to z\hbar^{-\textbf{s} w} = z q^{-w}$ for integer wall $w$. We conclude that:
\be
\label{Jexp}
\textbf{J}_w^{+}(z)=\sum\limits_{k=0}^{\infty} \,\dfrac{(-1)^k \hbar^{-k(k-1)/2} (\hbar-\hbar^{-1})^k}{[k]_\hbar! \prod\limits_{i=1}^{k}( 1- z^{-1} q^w K \otimes K^{-1} \hbar^{2 i} )} \, F_w^{k}\otimes E_w^{k}
\ee

\subsubsection{}
The operator $\repM_{w}(z)$ is given by (\ref{bdef}). To compute it, we need the formulas for antipode $S_w$
of $\Uq(\fg_w)$. They can be obtained directly from the wall $R$-matrix (\ref{sl2w}). First, from
$1\otimes\Delta (\Rwal)=\Rwal_{13}\Rwal_{12}$ and $\Delta\otimes 1 (\Rwal)=\Rwal_{13}\Rwal_{23}$ we obtain:
\begin{small}
\be \label{sl2copr}
\Delta(E)=K^{-1} \otimes E +E\otimes 1, \ \ \Delta(F)=1 \otimes F +F\otimes K, \ \ \Delta(K)=K\otimes K
\ee
\end{small}
and thus the antipode corresponding to this coproduct has the form:
$$
S(E)=- K E, \ \ S(F)=-F K^{-1}, \ \ S(K)=K^{-1}
$$

\subsubsection{}
The lower triangular solutions of the ABRR equation can be computed from (\ref{Jexp}) by
$\textbf{J}_w^{-}(z) = S_w\otimes S_w ( \textbf{J}_w^{+}(z)_{21} )$,
{\rr which gives:
$$
\textbf{J}_w^{-}(z)=\sum\limits_{k=0}^{\infty} \dfrac{(-1)^k \hbar^{-2 k^2 -k(k-1)/2} (\hbar-\hbar^{-1})^k K^k\otimes K^{-k}}{[k]_{\hbar}! \prod\limits_{i=1}^{k} (1-z^{-1} q^{w} K\otimes K^{-1} \hbar^{2 i -4k})} \, E^k \otimes F^k
$$
To compute the inverse of this operator we write
$$
\textbf{J}_w^{-}(z)^{-1}=1 + \sum\limits_{m=1}^{\infty} \, a_m\, E^m \otimes F^m
$$
and determine the unknown coefficients $a_n$ from the equation $\textbf{J}_w^{-}(z)^{-1} \textbf{J}_w^{-}(z) =1$. Comparing coefficients of $E^n \otimes F^n$ we find the following system of linear equations:
$$
\sum_{k+m=n}\, a_m \dfrac{(-1)^k \hbar^{-2 k^2 -k(k-1)/2 -4km} (\hbar-\hbar^{-1})^k K^k\otimes K^{-k}}{[k]_{\hbar}! \prod\limits_{i=1}^{k} (1-z^{-1} q^w K\otimes K^{-1} \hbar^{2i -4k -4m})}=0, \ \ \ n=1,2,\dots
$$
The coefficients $a_m$ are determined uniquely from this system. For instance, for $n=1$ we obtain 
$$
a_{1}=\dfrac{\hbar^{-2} (\hbar-\hbar^{-1}) K\otimes K^{-1} }{1-z^{-1} q^w K\otimes K^{-1} \hbar^{-2}}
$$
For $n=2$ we have
$$
\begin{aligned}
a_2-a_1 \dfrac{\hbar^{-6}(\hbar-\hbar^{-1}) K\otimes K^{-1} }{(1-z^{-1} q^w K\otimes K^{-1} \hbar^{-6})}+\\
+\dfrac{\hbar^{-9} (\hbar-\hbar^{-1})^2 K^2\otimes K^{-2}}{(\hbar+\hbar^{-1})(1-z^{-1} q^w K\otimes K^{-1}\hbar^{-6})(1-z^{-1} q^w K\otimes K^{-1}\hbar^{-4})}=0
\end{aligned}
$$
which gives
$$
a_2=\dfrac{\hbar^{-7} (\hbar-\hbar^{-1})^2 K^{2} \otimes K^{-2}}{[2]_\hbar! (1-z^{-1} q^w K\otimes K^{-1} \hbar^{-2}) (1-z^{-1} q^w K\otimes K^{-1} \hbar^{-4})}
$$
In general
$$
a_k=\dfrac{\hbar^{-\frac{k(3k+1)}{2}} (\hbar-\hbar^{-1})^k K^{k} \otimes K^{-k}}{[k]_{\hbar}! \prod\limits_{i=1}^{k} (1-z^{-1} q^w K\otimes K^{-1} \hbar^{-2 i})}.
$$
which can be proved by induction on $k$. 
}
Finally, we obtain:
$$
\textbf{m}\Big( 1\otimes S_w (\textbf{J}_w^{-}(z)^{-1})  \Big)=\sum\limits_{k=0}^{\infty}\, \dfrac{(-1)^k (\hbar-\hbar^{-1})^k \hbar^{-k(k+3)/2} }{[k]_\hbar! \prod\limits_{i=1}^{k} (1-z^{-1} q^w  K^2 \hbar^{-2 i} )} K^k E_w^k F_w^k.
$$
\subsubsection{ \label{bslsec}}
To compute the operator $\repM_{w}(z)$ we need to shift parameter $z$ by $\kappa$. By definition
$\kappa=(C \bv-\bw)/2$. Enough to compute the action of $\kappa$ in one evaluation module $\C^2(u)$ of $\sldh$.
This module corresponds to $\bw=1$. The Cartan matrix corresponding to our case is $C=2$. We therefore find:
$$
\kappa=\left\{\begin{array}{ll}
1/2 & \textrm{on}  \ \ \cM(1,1)\\
-1/2 & \textrm{on} \ \ \cM(0,1)
\end{array}\right. \ \ \Leftrightarrow
{\kappa = \left(\begin{array}{cc} 1/2 &0 \\ 0&-1/2 \end{array}\right)} = H/2
$$
Thus, we conclude that the shift $\lambda \rightarrow \lambda + \hat{\kappa}$ is given by\footnote{\rr The factor $2$ in $\hbar^{2\kappa}$ is from our conventions introduced at the beginning of Section \ref{subshsec}.}
{\rr 
$$
z\to z \hbar^{2 \kappa}=z \hbar^{H}=z K^{1}
$$
}
Thus, from the definition (\ref{bdef}) we obtain:
$$
\repM_{w}(z)=\sum\limits_{k=0}^{\infty}\, \dfrac{(-1)^k (\hbar-\hbar^{-1})^k \hbar^{-k(k+3)/2} }{[k]_\hbar! \prod\limits_{i=1}^{k} (1-z^{-1} q^w  K \hbar^{-2 i} )} K^k E_w^k F_w^k
$$
\subsubsection{}
The alcove specified by  Theorem \ref{mainth} corresponds to the interval $\nabla =(-1,0)$.
Let $s \in \nabla$ and $\cL ={\cal{O}}(1)$.
There is only one wall $w=-1$ between $s$ and $s-1$. Thus, the definition (\ref{dMdef})
and  Theorem \ref{mainth} give the following explicit formula for the quantum difference operator:
\be
\label{anssl2}
\begin{array}{|c|}
	\hline\\
	\\
	\ \ \textbf{M}_{{\cal{O}}(1)}(z)={\rr  Const}\,{\cal{O}}(1)\sum\limits_{k=0}^{\infty}\, \dfrac{(-1)^k (\hbar-\hbar^{-1})^k \hbar^{-k(k+3)/2} }{[k]_\hbar! \prod\limits_{i=1}^{k} (1-z^{-1} q^{-1} K \hbar^{-2 i} )} K^k E_{-1}^k F_{-1}^k \ \ \\
	\\
	\hline
\end{array}
\ee
{\rr We expect that the constant factor in Theorem \ref{mainth} is $Const=1$ for the case $k\leq n/2$ and non-trivial for $k>n/2$.\footnote{This expectation is in agreement with explicit computations of capped vertex functions \cite{pcmilect} for the first values of $k$ and $n$.} In the rest of this section we assume that $Const=1$ for simplicity.}
\subsubsection{}
Using (\ref{wallsh}) we can also rewrite this operator as:

\be \label{moper2}
\textbf{M}_{{\cal{O}}(1)}(z)=\Big(\sum\limits_{k=0}^{\infty}\, \dfrac{(-1)^k (\hbar-\hbar^{-1})^k \hbar^{-k(k+3)/2} }{[k]_\hbar! \prod\limits_{i=1}^{k} (1-z^{-1} q^{-1} K \hbar^{-2 i} )} K^k E^k F^k\Big)  {\cal{O}}(1).\ \ 
\ee
This form is particularly convenient for explicit computations as it expresses the difference operator
through the standard $\Uq(\frak{sl}_2)$.
\subsubsection{}
{\rr An important feature of quasimap quantum K-theory of Nakajima varieties is the degeneration formula, see Section 6.5 in \cite{pcmilect}. This formula relates the count of quasimaps from a curve $C$ and from its nodal degeneration $C \to C_1 \cup _p C_2$. The main element of the degeneration formula is the ``glue operator" $\textbf{G}$ defined by (6.5.20) in \cite{pcmilect}.  We have the following result:
\begin{Theorem}[Corollary 8.1.19, \cite{pcmilect}]
\be \label{glumatKth}
\begin{array}{ll}
\lim\limits_{q\to 0} \geomM_{\cL}(z)=\cL \\
\\
\lim\limits_{q\to\infty}\, \geomM_{\cL}(z q^{-1}) \cL^{-1}=\mathbf{G}
\end{array}
\ee
\end{Theorem}
It is elementary to check that the first limit in this Theorem is in agreement with our formula (\ref{moper2}). From the second limit we obtain a formula for the glue operator in terms of representation theory:
$$
\textbf{G}=\sum\limits_{k=0}^{\infty}\, \dfrac{(-1)^k (\hbar-\hbar^{-1})^k \hbar^{-k(k+3)/2} }{[k]_\hbar! \prod\limits_{i=1}^{k} (1-z^{-1} K \hbar^{-2 i} )}K^k E^k F^k.
$$
}

\subsubsection{}
{\rr Let us compute the matrices of the operator $\textbf{M}_{{\cal{O}}(1)}(z/q)$ for the first few cases. Let $e_1$ and $e_2$ be the standard basis of $\mathbb{C}^2$ with standard action of $\sldh$:
$$
E  e_1 =0,  \ \ E e_2= e_1, \ \ F e_1=e_2, \ \ F e_2 =0, \ \ K e_1 = \hbar e_1, \ \ K e_2= \hbar^{-1} e_2
$$
The K-theory of $T^{*} \mathbb{P}^1$ corresponds to the $0$-weight subspace of $\mathbb{C}^2(u_1) \otimes \mathbb{C}^2(u_2)$. We use the stable map to identify the basis  $e_1\otimes e_2$ and $e_2\otimes e_1$ in this space with the basis of stable envelopes for an  anti-canonical slope $s\in (-1,0)$ which we computed in Section \ref{stabbassec} . As $E^2=F^2=0$ and $\Delta(K)=K\otimes K =1$ we have:
$$
\repM_{0}(z)=1-\frac{(\hbar-\hbar^{-1}) \hbar^{-2}}{ 1-z^{-1} \hbar^{-2}} \Delta(E)  \Delta(F)
$$
where $\Delta$ is coproduct (\ref{sl2copr}). In the basis $e_1\otimes e_2,e_2\otimes e_1$  we compute
$$
\Delta(E)  \Delta(F)= \left[\begin{array}{ll}
\hbar^{-1}&1\\
1 & \hbar
\end{array}\right]
$$
Thus, in the stable basis we have:
$$
\repM_{0}(z)_{stab}=\left[ \begin {array}{cc} {\dfrac {{\hbar}^{4}z-z{\hbar}^{2}-{\hbar}^{2}+z}{{\hbar}^{2
		} \left( z{\hbar}^{2}-1 \right) }}&{\dfrac { \left( \hbar-\hbar^{-1} \right) z}{ 1-z{\hbar}^{2} }}\\ \noalign{\medskip}{\dfrac { \left( \hbar-\hbar^{-1} \right) z}{ 1-z{\hbar}^{2} }}&{\dfrac {1-z}{1-z{\hbar}^{2}}}\end {array} \right] 
$$
Next, the matrix of the operator of multiplication by ${{\cal{O}}(1)}$ in the basis of fixed points equals:
$$
{{\cal{O}}(1)}_{fp}=\left[\begin{array}{cc}
u_1 & 0 \\
0 & u_2
\end{array}\right]
$$
To compute the action of this operator in the stable basis we use explicit formulas from Section \ref{stabbassec}. The transition matrix between the basis of fixed points and the stable basis for $s\in (-1,0)$ is computed by:\footnote{{\rr Note that we need to substitute $\hbar\to \hbar^{2}$ in the geometric formulas to relate them to the action of $\sldh$, as we explain in Section \ref{subshsec}. }}
\be \label{Tmat}
T_{i,j} = \dfrac{\left.\textrm{Stab}_{\fC,T^{1/2}, s}(p_j)\right|_{p_i}}{ \Lambda^{\bullet}(T_{p_i} X^{\vee})} = \left[ \begin {array}{cc} {\dfrac {u_{{1}}}{-u_{{2}}+u_{{1}}}}&0
\\ \noalign{\medskip}-{\dfrac { \left( \hbar-1 \right)  \left( \hbar+1 \right) 
		u_{{2}}u_{{1}}}{ \left( -u_{{2}}+u_{{1}} \right)  \left( u_{{2}}{\hbar}^{2
		}-u_{{1}} \right) }}&{\dfrac {u_{{2}}\hbar}{u_{{2}}{\hbar}^{2}-u_{{1}}}}
\end {array} \right]
\ee
Thus, the action of ${{\cal{O}}}(1)$ in the stable basis is given by
$$
{{\cal{O}}}(1)_{stab}=T^{-1}\, {{\cal{O}}(1)}_{fp}\, T =  \left[ \begin {array}{cc} u_{{1}}&0\\ \noalign{\medskip}{\dfrac {
		\left( \hbar-1 \right)  \left( \hbar+1 \right) u_{{1}}}{\hbar}}&u_{{2}}
\end {array} \right]
$$
Finally, we compute:
$$
\textbf{M}_{{\cal{O}}(1)}(z/q)_{stab}=\repM_{0}(z)_{stab} {{\cal{O}}}(1)_{stab}= \left[ \begin {array}{cc} {\dfrac {u_{{1}} \left( z-1 \right) }{z{h}^{
			2}-1}}&{\dfrac { \left( \hbar-\hbar^{-1} \right)  zu_{{2}}}{
		\left(1- z{\hbar}^{2} \right) }}\\ \noalign{\medskip} {\dfrac { \left( \hbar-\hbar^{-1}
		\right)  u_{{1}}}{ \left( 1-z{\hbar}^{2} \right) }}&{
	\dfrac { \left( z-1 \right) u_{{2}}}{z{\hbar}^{2}-1}}\end {array} \right] 
$$

For $T^{*}\mathbb{P}^2$ the computation is the same: we consider the subspace spanned by 
$e_1 \otimes e_2 \otimes e_2, e_2\otimes e_1 \otimes e_2, e_2 \otimes e_2 \otimes e_1$ in  $\mathbb{C}^2(u_1) \otimes \mathbb{C}^2(u_2) \otimes \mathbb{C}^2(u_3)$. The element $K$ acts on this subspace via 
$\Delta(K)=K\otimes K\otimes K$, i.e.,  as multiplication by $\hbar^{-1}$. As $F^2=0$ we find: 
$$
\repM_{0}(z)=1-\frac{(\hbar-\hbar^{-1}) \hbar^{-3}}{ 1-z^{-1} \hbar^{-3}} \Delta^{2}(E) \Delta^{2}(F)
$$
The computation gives: 
$$
\Delta^{2}(E) \Delta^{2}(F)=\left[\begin{array}{lll}
\hbar^{-2} & \hbar^{-1} &1 \\
\hbar^{-1} & 1 & \hbar \\
1 & \hbar & \hbar^{2}
\end{array}\right]
$$
Next, computing the stable envelopes for $T^{*}\mathbb{P}^2$ as in Section  \ref{stabbassec} for $s\in (-1,0)$ would give
$$
\begin{array}{l}
\textrm{Stab}_{\fC,T^{1/2}, s}(p_1)=(1-{\cal{O}}(1) \hbar^{2}/u_2)(1-{\cal{O}}(1) \hbar^{2}/u_3)\\
\textrm{Stab}_{\fC,T^{1/2}, s}(p_2)=\hbar (1-{\cal{O}}(1) /u_1)(1-{\cal{O}}(1) \hbar^{2}/u_3)\\
\textrm{Stab}_{\fC,T^{1/2}, s}(p_3)= \hbar^{2} (1-{\cal{O}}(1) /u_1)(1-{\cal{O}}(1)/u_2)
\end{array}
$$
Using these formulas we find:
$$
{{\cal{O}}}(1)_{stab}=T^{-1}\, {{\cal{O}}(1)}_{fp}\, T = \left[ \begin {array}{ccc} u_{{1}}&0&0\\ \noalign{\medskip}{\dfrac {
		\left( \hbar^2-1 \right) u_{{1}}}{\hbar}}&u_{{2}}&0
\\ \noalign{\medskip} \left( \hbar^2-1 \right)   u_{{1}}&{
	\dfrac {u_{{2}} \left( \hbar^2-1 \right)  }{\hbar}}&u_{{3}}
\end {array} \right]
$$
where 
$$
{{\cal{O}}(1)}_{fp}=\left[\begin{array}{lll}
u_1&0&0\\
0&u_2&0\\
0&0&u_3
\end{array}\right]
$$
and $T$ is the transition matrix computed as in (\ref{Tmat}). Combining all this together for $\textbf{M}_{{\cal{O}}(1)}(z/q)_{stab}=\repM_{0}(z)_{stab} {{\cal{O}}}(1)_{stab}$ we find:
$$
\textbf{M}_{{\cal{O}}(1)}(z/q)_{stab}= \left[ \begin {array}{ccc} {\dfrac { \left( 1-\hbar z \right) u_{{1}} }{1-z{\hbar}
		^{3}}}&{\dfrac { \left( \hbar-\hbar^{-1} \right) \hbar zu_{{2}}}{1-z{
			\hbar}^{3}}}&{\dfrac { \left( \hbar-\hbar^{-1} \right)  zu_{{3}}}{
		 1-z{\hbar}^{3} }}\\ \noalign{\medskip}{\dfrac {
		\left( \hbar-\hbar^{-1} \right) u_{{1}} }{ \left( 1-z{\hbar}^{3} \right) 
}}&{\dfrac { \left( 1-\hbar z \right) u_{{2}} }{1-z{\hbar}^{3}}}&
{\dfrac {\left( \hbar-\hbar^{-1} \right) \hbar zu_{{3}}}{1-z{\hbar}^{3}}}
\\ \noalign{\medskip}{\dfrac { \left( \hbar-\hbar^{-1} \right) \hbar 
		u_{{1}}}{1-z{\hbar}^{3}}}&{\dfrac { \left( \hbar-\hbar^{-1} \right) u_{{2}}  }{1- z{\hbar}^{3} }}&{\dfrac { \left( 1-\hbar z
		\right) u_{{3}}}{1-z {\hbar}^{3}}}\end {array} \right].
$$}

\section{Instanton moduli spaces  \label{apa}}

In this section we consider the example of Jordan quiver: the quiver consisting of one vertex and a single loop.
The dimension vectors are given by two non-negative integer numbers $\bv=m$, $\bw=r$. The corresponding variety
$\cM(m,r)$ is the moduli space of framed rank $r$ torsion-free sheaves ${\cal{F}}$ on $\mathbb{P}^2$ with fixed second Chern class $c_{2}({\cal{F}})=m$. A framing of a sheaf ${\cal{F}}$ is a choice of an isomorphism:
\be
\label{fram}
\phi: \left.{\cal{F}}\right|_{L_{\infty}} \to {\mathscr{O}}^{\oplus r }_{L_{\infty}}
\ee
where $L_{\infty}$ is the line at infinity of $\C^{2} \subset \mathbb{P}^2$. This moduli space is usually referred to as instanton moduli space.

Let $\bA \simeq (\C^{\times})^r$ be the framing torus acting on $\cM(m,r)$ by changing the isomorphism (\ref{fram}).
This torus acts on the instanton moduli space preserving the symplectic form.

Let us denote by $\KG=\bA\times (\C^{\times})^2$ where the second factor acts on $\C^2 \subset \mathbb{P}^2 $ by scaling the coordinates. This induces an action of $\KG$ on~$\cM(m,r)$. The action of this torus scales the symplectic form with a character which we denote by $\hbar$.

We denote the equivariant parameters corresponding to $\bA$ by $u_1,\cdots,u_r$, and to torus $\KG/\bA$ by $t_1, t_2$ such that the weight of the symplectic form is:
$$
\hbar = t_1 t_2
$$
\subsection{Algebra $\Uq(\fgh_{Q})$ and wall subalgebras $\Uq(\fg_w)$}

\subsubsection{}
In the special case $r=1$ the instanton moduli space is isomorphic to the Hilbert scheme of $m$ points on the complex plane $\cM(m,1)=\textrm{Hilb}^{m}(\C^2)$.  As a vector space, the $K$-theory of Hilbert schemes can be identified with polynomials in an infinite number of variables.
\be\label{fc}
\bigoplus\limits_{m=0}^{\infty}\, K_{\KG}(\textrm{Hilb}^{m}(\C^2)  )= \textsf{F}(u_1) \stackrel{def}{=} \Q[p_1,p_2,\cdots] \otimes \Q[u_1^{\pm 1}, t_1^{\pm 1},t_2^{\pm 1}]
\ee
If we introduce a grading in the polynomial ring $\Q[p_1,p_2,\cdots]$ by $\deg(p_k)=k$. Then the $m$-th term
on the left  side of (\ref{fc}) corresponds to degree $m$.
\subsubsection{}
The fixed point set $\textrm{Hilb}^{m}(\C^2)^{\KG}$ is discrete.
Its elements are labeled by partitions $\nu$ with  $|\nu|=m$. The structure sheaves of the fixed points ${\cal{O}}_{\nu}$ form a basis of the localized $K$-theory. The polynomials representing the elements of this basis
under isomorphism (\ref{fc}) are the Macdonald polynomials $P_{\nu}$ in Haiman normalization \cite{Haim}.
To fix the norms we write the first several Macdonald polynomials here:
\begin{align*}
P_{[1]}=p_1,  \ \ P_{[2]}=\frac{1+ t_1}{2}\, p_1^2 +\frac{1-t_1}{2}\, p_2, \ \
P_{[1,1]}=\frac{1+ t_2}{2} \,p_1^2 +\frac{1-t_2}{2} \,p_2\\
\\
P_{[3]}=\frac{(1+t_1)(1+t_1+t_1^2)}{6} \, p_1^3 +\frac{(1-t_1)(1+t_1+t_1^2)}{2} \, p_1 p_2 +\frac{(1+t_1)(1-t_1)^2}{3} \, p_3\\
\\
P_{[1,1,1]}=\frac{(1+t_2)(1+t_2+t_2^2)}{6} \, p_1^3 +\frac{(1-t_2)(1+t_2+t_2^2)}{2} \, p_1 p_2 +\frac{(1+t_2)(1-t_2)^2}{3}\,  p_3\\
\\
P_{[2,1]}=\frac{1+ t_1 t_2 +2 t_1 +2 t_2}{6} \, p_1^3 +\frac{1-t_1 t_2}{ 2}\,  p_2 p_1 +\frac{(1-t_1)(1-t_2)}{3}\,  p_3
\end{align*}

\subsubsection{}
Assume, that the torus $\bA$ splits the framing by $\bw=u_1+\cdots+u_r$ then in the notations of Section \ref{rmatdef}
we obtain:
\be
\label{Kfock}
\bigoplus\limits_{n=0}^{\infty}\, K_{\KG}( \cM(n,r)  ) =  \textsf{F}(u_1)\otimes \cdots \otimes  \textsf{F}(u_r)
\ee
\subsubsection{}
Let us set $\hZ =\matZ^2$, $\hZ^{*}=\hZ\setminus \{(0,0)\}$ and:
$$
\hZ^{+}=\{ (i,j)\in \hZ; i>0 \ \ \textrm{or} \ \ i=0,\ \ j>0 \}, \ \ \hZ^{-}=-\hZ^{+}
$$
Set
$$n_{k}=\dfrac{(t_1^{\frac{k}{2}}-t_1^{-\frac{k}{2}})(t_2^{\frac{k}{2}}-t_2^{-\frac{k}{2}})
	(\hbar^{-\frac{k}{2}}-\hbar^{\frac{k}{2}})}{k}$$
and for vector $\textbf{a}=(a_1,a_2) \in \hZ$ denote by $\textrm{deg}(\textbf{a})$ the greatest common divisor of $a_1$ and $a_2$, {\rr in particular $\textrm{deg}((m,0))=\textrm{deg}((0,m))=m$}. We set $\epsilon_\textbf{a}=\pm 1$ for $\textbf{a}\in \hZ^{\pm}$. For a pair of non-collinear vectors we set $\epsilon_{\textbf{a},\textbf{b}}=\textrm{sign}(\det(\textbf{a},\textbf{b}))$.

\noindent

The ``toroidal'' algebra $\gt$ is an associative algebra with $1$ generated by elements $e_{\textbf{a}}$ and $K_\textbf{a}$ with $\textbf{a} \in \hZ^{\rr  \ast}$, subject to the following relations \cite{SchifVas}:
\begin{itemize}
	\item  elements $K_\textbf{a}$ are central and
	$$K_0=1, \ \ \ K_\textbf{a} K_\textbf{b}=K_{\textbf{a}+\textbf{b}}$$
	\item if  $\textbf{a}$, $\textbf{b}$ are two collinear vectors then:
	\be
	\label{colin}
	[e_\textbf{a},e_\textbf{b}]=\delta_{\textbf{a}+\textbf{b}} \dfrac{K^{-1}_\textbf{a}-K_\textbf{a}}{n_{\textrm{deg}(\textbf{a})}}
	\ee
	\item if $\textbf{a}$ and $\textbf{b}$ are such that $\textrm{deg}(\textbf{a})=1$ and the triangle $\{(0,0), \textbf{a},\textbf{a}+ \textbf{b}\}$
	has no interior lattice points then
	$$
	[e_\textbf{a},e_\textbf{b}]=\epsilon_{\textbf{b},\textbf{a}} K_{\alpha(\textbf{b},\textbf{a})} \, \dfrac{\Psi_{\textbf{a}+\textbf{b}}}{n_1}
	$$
	where
	$$
	\alpha(\textbf{a},\textbf{b})=\left\{\begin{array}{ll}
	\epsilon_{\textbf{a}}(\epsilon_{\textbf{a}} \textbf{a} +\epsilon_\textbf{b} \textbf{b} -\epsilon_{\textbf{a}+\textbf{b}}(\textbf{a}+\textbf{b})  )/2 & \textrm{if} \ \ \epsilon_{\textbf{a},\textbf{b}}=1\\
	\epsilon_{\textbf{b}}(\epsilon_{\textbf{a}} \textbf{a} +\epsilon_\textbf{b} \textbf{b} -\epsilon_{\textbf{a}+\textbf{b}}(\textbf{a}+\textbf{b})  )/2 & \textrm{if} \ \ \epsilon_{\textbf{a},\textbf{b}}=-1
	\end{array}\right.
	$$
	and elements $\Psi_\textbf{a}$ are defined by:
	$$
	\sum\limits_{k=0}^{\infty} \, \Psi_{k \textbf{a}} z^k = \exp\Big( \sum\limits_{i=1}^{\infty}\, n_i \, e_{i\, \textbf{a}}\, z^i \Big)
	$$
	for $\textbf{a}\in \hZ$ such that $\textrm{deg}(\textbf{a})=1$.
\end{itemize}

\subsubsection{}
For $w\in \Q \cup \{\infty\}$ we denote by $d(w)$ and $n(w)$ the denominator and numerator of $w$.
We set $d(\infty)=0$ and $n(\infty)=1$. From (\ref{colin}) we see that
$$
\alpha^{w}_{k}= e_{(d(w)k, n(w)k)}, \ \ k \in \Z {\rr \setminus \{0 \}}
$$
generate a Heisenberg subalgebra of $H_w\subset \gt$ with the following relations:
$$
[\alpha^{w}_{-k},\alpha^{w}_{k}]=\dfrac{ K_{(1,0)}^{k d(w) } -K_{(1,0)}^{-k d(w) } }{n_k}
$$
It is convenient to visualize the algebra $\gt$ as in the Figure \ref{pic2}. The Heisenberg subalgebras of $\gt$
are labeled by $w \in \mathbb{Q}$ and correspond to lines with slope $w$ in this picture.
\hspace{2mm}
\begin{figure}[ht]
	\begin{center}

		\begin{center}
		\includegraphics[scale=0.4]{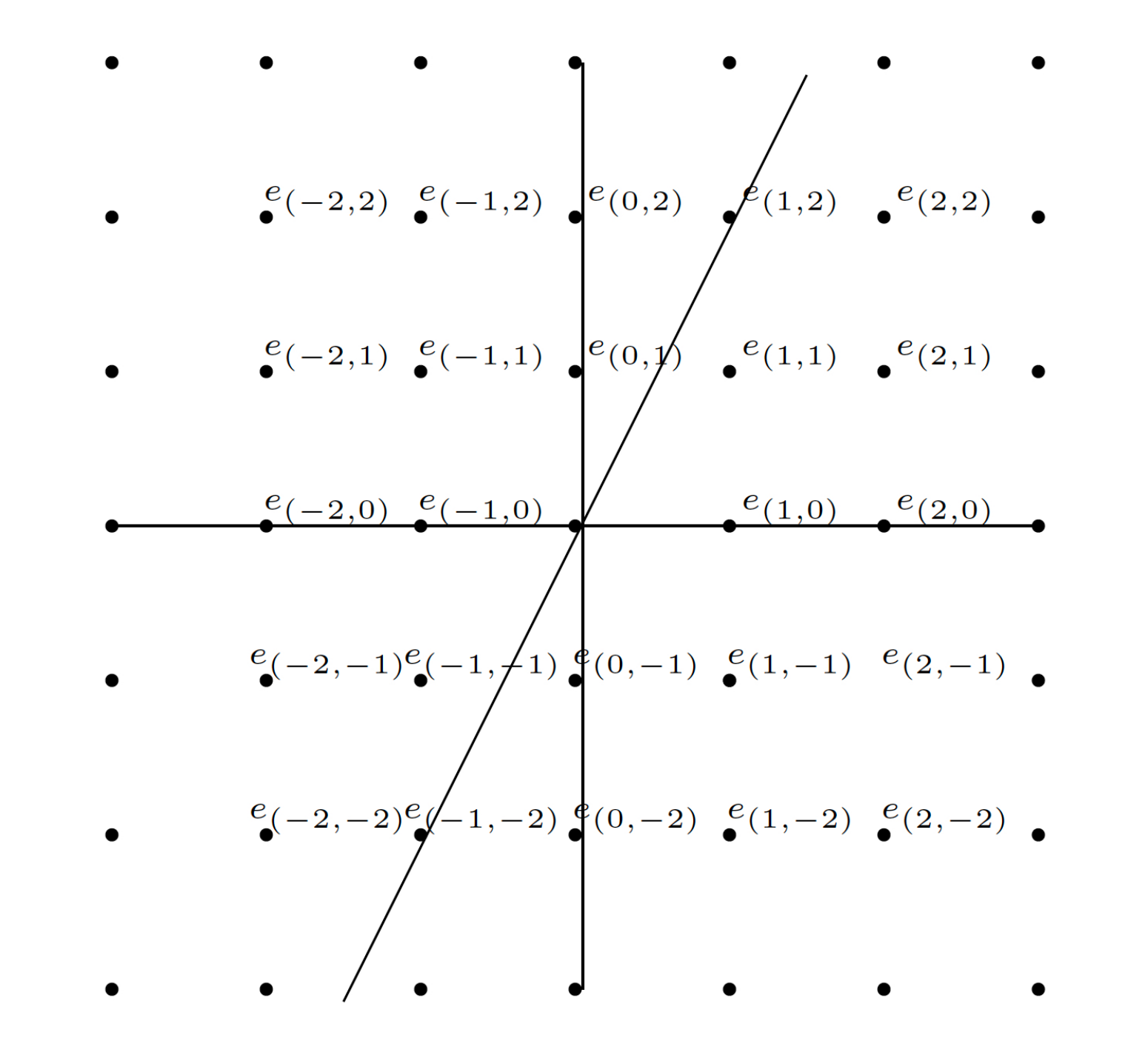}	
		\end{center}

		\caption{The line with slope $2$ corresponds to the Heisenberg subalgebra
			generated by $e_{ k,2 k}$ for $k\in \Z \setminus \{0\}$.  \label{pic2}}
	\end{center}
\end{figure}

\subsubsection{}
The action of $\gt$ on the $K$-theory (\ref{fc}) was constructed
in \cite{SchifVas}. The central elements  act in this representation by:
\be
K_{(1,0)}=t_1^{-\frac{1}{2}} t_2^{-\frac{1}{2}}, \ \ \ K_{(0,1)}=1
\ee
In particular, the ``vertical'' generators commute in this representation:
$$
[ e_{(0,m)}, e_{(0,n)}]=0
$$
The ``horizontal'' Heisenberg subalgebra:
$$
[ e_{(m,0)}, e_{(n,0)}]=  \dfrac{-m}{(t_1^{m/2} - t_1^{-m/2} )(t_2^{m/2} - t_2^{-m/2} ) } \, \delta_{n+m}
$$
acts explicitly as follows:
\be
\label{f1}
e_{(m,0)}=\left\{\begin{array}{rl}
	\dfrac{1}{(t_1^{m/2}-t_1^{-m/2})(t_2^{m/2}-t_2^{-m/2})}\,  p_{-m} & m<0 \\
	\\
	-m \dfrac{\partial}{\partial p_m} &   m>0
\end{array}\right.
\ee
The action of vertical subalgebra is diagonal in Macdonald polynomials:
\be
\label{f2}
e_{(0,l)} (P_\lambda) =  u_1^{-l} \textrm{sign}(l) \left( \dfrac{1}{ 1-t_1^{-l}} \sum\limits_{i=1}^{\infty} \, t_1^{-l \lambda_i} t_2^{-l(i-1)}   \right) P_\lambda
\ee
The infinite sum here should be understood as the series expansion of a rational function:
$$
\sum\limits_{i=1}^{\infty} \, t_1^{-l \lambda_i} t_2^{-l(i-1)}=\sum\limits_{i=1}^{length(\lambda)} \, t_1^{-l \lambda_i} t_2^{-l(i-1)}+\dfrac{t_2^{- length(\lambda) l }}{1-t_2^{-l}}.
$$
It is clear that $e_{(0,l)}$ and $e_{(l,0)}$ generate the whole $\gt$. Thus, the last two formulas
determine the action of $\gt$ on the Fock space.

\subsubsection{}
{\rr It is expected that the geometric algebra $ \Uq(\fgh_{Q})$ is isomorphic to $\gt$, see \cite{Neg} for discussion. Among other things, this isomorphism implies that the $R$-matrix of $ \Uq(\fgh_{Q})$ evaluated in the tensor product of the Fock modules coincides with the geometric $R$-matrix for the instanton moduli spaces. In particular, comparing the ``universal formula''  for the $R$-matrix of $\gt$ obtained in \cite{NegRmatrix} with the KT-factorization (\ref{Rfac}), we find that the wall $R$-matrices of $ \Uq(\fgh_{Q})$ coincide with the $R$-matrices of the slope Heisenberg algebras $H_w$ (to see that it  is enough to compare the limits (\ref{limRU2}) of the $R$-matrices). This way, this leads to an isomorphism of the wall subalgebras $\Uq(\fg_w) \subset  \Uq(\fgh_{Q})$ and Heisenberg subalgebras $H_w \subset \gt$.


In the remaining part of this section we derive formulas for the quantum difference equation for the instanton moduli spaces assuming the above isomorphism exists.
}

\subsection{R-matrices}
\subsubsection{\label{heisenbdef}}
Recall that the quantum Heisenberg algebra is an algebra generated by elements $e$, $f$ and a central element $K$ modulo the following relations:
\be
\label{hei}
[e,e]=[f,f]=0, \ \  [e,f]=\dfrac{K-K^{-1}}{c-c^{-1}}
\ee
The Fock space  $\textsf{F}=\Q[x]\otimes \Q[c^{\pm 1}]$ is a natural module over the Heisenberg algebra
with the following action:
$$
e(p) = x p, \ \ f (p) = -\dfrac{d p}{d x }, \ \  K (p) = c p
$$
so that $c$ is a formal parameter fixing the value of central element $K$ in $\textsf{F}$. The Heisenberg algebra is a Hopf algebra with the following coproduct:
$$
\begin{array}{l}
\Delta(e)=e \otimes 1 + K^{-1} \otimes e \\
\\
\Delta(f)=f \otimes K + 1 \otimes f\\
\\
\Delta(K)=K\otimes K
\end{array}
$$
antipode:
$$
S(e)=-K e, \ \  S(f)= -K^{-1} f, \ \ S(K)=K^{-1}
$$
and counit:
$$
\varepsilon(e)=\varepsilon(f)=0, \ \ \varepsilon(K)=1
$$
We consider the tensor product $\textsf{F}\otimes \textsf{F} =\Q[x,y]\otimes \Q[c^{\pm 1}]$, and define codimension function by $c^\Omega (x^i y^j) =c^{i+j} x^i y^j $.
We consider the following upper and lower triangular R-matrices.
$$
\Rwal^{+}=c^{-\Omega} \exp(-(c-c^{-1})\,f \otimes e), \ \ \ \Rwal^{-}=c^{-\Omega} \exp(-(c-c^{-1})\,e \otimes f)
$$
\begin{Proposition}
	The R-matrices satisfy the QYBE in $\textsf{F}^{\otimes 3}$:
	$$
	\Rwal^{\pm}_{23}\Rwal^{\pm}_{13}\Rwal^{\pm}_{12} =  \Rwal^{\pm}_{12}\Rwal^{\pm}_{13}\Rwal^{\pm}_{23}
	$$
	and have the following properties:
	$$
	\Rwal^{+} \Delta =\Delta_{21} \Rwal^{+}, \ \ \  \Rwal^{-} \Delta_{21} =\Delta \Rwal^{-}
	$$
	where $\Delta_{21}$ is the opposite coproduct, and
	$$
	\begin{array}{l}
	1 \otimes \Delta ( \Rwal^{+} )=  \Rwal^{+}_{13}\Rwal^{+}_{12}, \ \
	\Delta  \otimes 1 ( \Rwal^{+} )= \Rwal^{+}_{13}\Rwal^{+}_{23}\\
	\\
	1 \otimes \Delta ( \Rwal^{-} )=\Rwal^{-}_{12}\Rwal^{-}_{13}, \ \
	\Delta  \otimes 1 ( \Rwal^{-} )= \Rwal^{-}_{23}\Rwal^{-}_{13}
	\end{array}
	$$
\end{Proposition}
\subsubsection{}
The Picard group $\Pic(X)=\Z$ is generated by ${\cal{O}}(1)$. It acts on $H^{2}(X,\R)=\R$ by shifts.
The explicit computation of stable map for $\cM(m,r)$ {\rr \cite{SmirnovEllipticHilbert,DinkinsAffineQuivers}}   shows that $\Stab^{\slope}$ is a locally constant function which changes only at the walls:
$$
\textrm{walls} = \{ w=\frac{a}{b} \in \R \, :\, a \in \Z, \ \  \,b\in \{1,2,...,m\}  \}
$$
Therefore, the set of walls for $\cM(r) = \coprod\limits_{m=0}^{\infty}  \cM(m,r) $  is identified
with rational numbers $\Q \subset \R$.

\subsubsection{}
We conclude that the $R$-matrix $R^{+}_w$ for the wall $w\in \Q$ corresponding to the Heisenberg subalgebra $\Uq(\fg_w)$ takes the form:
\be \label{rplusbos}
R^{+}_{w}=\prod\limits_{k={\rr 1}}^{\infty} \,\exp(- n_k\, \alpha^{w}_{k} \otimes \alpha^{w}_{-k} ) = \exp\Big(- \sum\limits_{k={\rr 1}}^{\infty}\, n_k\, \alpha^{w}_{k} \otimes \alpha^{w}_{-k} \Big)
\ee
The lower triangular $R$-matrix is obtained by the transposition:
\be \label{rminbos}
R^{-}_{w}=\prod\limits_{k={\rr 1}}^{\infty} \,\exp( - n_k\, \alpha^{w}_{-k} \otimes \alpha^{w}_{k} ) = \exp\Big(- \sum\limits_{k={\rr 1}}^{\infty}\, n_k\, \alpha^{w}_{-k} \otimes \alpha^{w}_{k} \Big)
\ee
As the central element of the elliptic Hall algebra acts in the Fock space by $K_{(1,0)}=\hbar^{-1/2}$, the central parameter $c$ of the quantum Heisenberg algebra generated by $e=\alpha^{w}_{-k}$ and $f=\alpha^{w}_{k}$
is given by $c=\hbar^{-k d(w)/2}=( t_1 t_2)^{-k d(w)/2}$.

\subsubsection{}
Let us fix a slope $\slope\in H^{2}(X,\R)=\R$. The Khoroshkin-Tolstoy factorization (\ref{Rfac})
provides the following universal formula for the total $R$-matrix:
\be
\Rtot^{\slope}(u)=\prod\limits_{{w\in \Q}\atop {w<\slope}}^{\rightarrow} R^{-}_{w} \,R_{\infty}\, \prod\limits_{{w\in \Q}\atop {w>\slope}}^{\leftarrow} R^{+}_{w}
\ee
The infinite slope $R$ -matrix $R_{\infty}$ is the operator of multiplication by normal bundles (\ref{infr}). From explicit formula for action of $\alpha^{\infty}_k$ (\ref{f2}) we can obtain:
$$
R_{\infty}=\exp\Big( -\sum\limits_{k=1}^{\infty}\, n_k\, \alpha^{\infty}_{-k} \otimes \alpha^{\infty}_{k} \Big)
$$
This, together with formulas from the previous section give the following universal expression
for a slope $\slope$ $R$-matrix:
$$
\Rtot^{\slope}(u)=\prod\limits_{w \in \Q  \cup\{\infty\}}^{\leftarrow \slope} \, \exp\Big( -\sum\limits_{k=1}^{\infty}\, n_k\, \alpha^{w}_{-k} \otimes \alpha^{w}_{k} \Big).
$$
As we mentioned above, this universal factorization of toroidal $R$-matrix is expected to coincide with one obtained in \cite{NegRmatrix}. 
{\rr 
\begin{Remark}
The geometric $R$-matrices associated to a Nakajima variety with a quiver $Q$ can be expressed as infinite products of $R$-matrices associated with the universal cover quiver $\widehat{Q}$, see Section 4.3 in \cite{MO}. This leads to infinite product formulas for $R$-matrices different from the KT-factorization described above. For the Jordan quiver $Q$, the universal cover $\widehat{Q}$ is the $A_{\infty}$-type quiver, and thus the $R$-matrices for the instanton moduli factor to infinite products of the $\cU_{\hbar}(\widehat{\frak{gl}}_{\infty})$ $R$-matrices. In equivariant cohomology an example of such factorization is considered in \cite{InstR}. A similar formula holds in equivariant $K$-theory. 
\end{Remark}}

\subsection{The quantum difference operator $\geomM_{\cL}(z)$}
\subsubsection{}
{\rr In this section we derive the solution of the ABRR equation. We assume that $\bA$ splits the framing by $r=r_1 u_1+ r_2 u_2$ so that
$$
K_{\KG}(\cM(r)^{\bA})=\textsf{F}^{\otimes r_1}(u_1) \otimes \textsf{F}^{\otimes r_2}(u_2).
$$	
Let $F = \cM(m_1, r_1) \times \cM(m_2, r_2)$ be a component of $\cM(m, r)^{\bA}$. As $\dim \cM(m, r) =2 r m $  we obtain that the corresponding eigenvalue of $\Omega$ equals:
\be \label{omrr}
\Omega = \frac{\codim(F)}{4}=\frac{2 r m -2 r_1 m_1 -2 r_2 m_2  }{4}= \frac{m_1 r_2 + m_2 r_1}{2}.
\ee	
The ABRR equation (\ref{wallKZ}) for a wall $w\in \Q$ takes the form:
\be
\hbar^{\Omega} \label{abrrinst}
R^{-}_{w}  \hbar^{-\lambda}_{(1)} J^{-}_{w}(z) = J^{-}_w(z) \hbar^{\Omega} \hbar^{-\lambda}_{(1)}
\ee
We are looking for a strictly lower-triangular solution $J^{-}_{w}(z) \in \Uq(\fg_w)^{\otimes 2}$ which means that
$J^{-}_{w}(z)$ is of the form:
$$
J^{-}_{w}(z)=\exp\Big( \sum\limits_{k=1}^{\infty}\, J_k(z)\, \alpha^{w}_{-k} \otimes \alpha^{w}_{k} \Big)
$$
We have:
\be \label{abrrtrans}
R^{-}_{w}  \hbar^{-\lambda}_{(1)} J^{-}_{w}(z) \hbar^{\lambda}_{(1)} =\hbar^{-\Omega} J^{-}_w(z) \hbar^{\Omega} 
\ee
and
\be \label{conj1}
 \hbar^{-\lambda}_{(1)} J^{-}_{w}(z) \hbar^{\lambda}_{(1)}=\exp\Big( \sum\limits_{k=1}^{\infty}\, J_k(z) z^{-k d(w)}\, \alpha^{w}_{-k} \otimes \alpha^{w}_{k} \Big).
\ee
We note that $\alpha^{w}_{-k} \otimes \alpha^{w}_{k}$ acts by
$$
 K_{\bT}(\cM(m_1, r_1) \times \cM(m_2, r_2)) \longrightarrow K_{\bT}(\cM( m_1+ k d(w),r_1) \times \cM( m_2-k d(w),r_2))
$$
Thus, (\ref{omrr}) for the corresponding matrix element  we have:
$$
\hbar^{-\Omega} \alpha^{w}_{-k} \otimes \alpha^{w}_{k} \hbar^{\Omega} =\hbar^{\frac{k d(w) r_1 - k d(w) r_2}{2}}\, \alpha^{w}_{-k} \otimes \alpha^{w}_{k}.
$$

We note that $K_{(1,0)}$ acts on $\textsf{F}$ by the scalar $\hbar^{-1/2}$ and thus it acts on $\textsf{F}^{\otimes r}$ via $\Delta^r(K_{1,0})=K_{(1,0)}^{\otimes r}$, i.e., by the scalar $\hbar^{-r/2}$. In this view, we can write the last equation in universal form
$$
\hbar^{-\Omega} \alpha^{w}_{-k} \otimes \alpha^{w}_{k} \hbar^{\Omega} =K_{(1,0)}^{- k d(w)} \otimes K^{k d(w)}_{(1,0)}\,  \alpha^{w}_{-k} \otimes \alpha^{w}_{k}.
$$
We conclude
\be \label{conj2}
\hbar^{-\Omega} J^{-}_w(z) \hbar^{\Omega} = \exp\Big( \sum\limits_{k=1}^{\infty}\, J_k(z)\,K_{(1,0)}^{-k d(w)} \otimes K^{k d(w)}_{(1,0)}\, \alpha^{w}_{-k} \otimes \alpha^{w}_{k} \Big).
\ee
Substituting (\ref{conj1}) , (\ref{conj2}) and (\ref{rminbos}) to the ABRR equation (\ref{abrrtrans}) gives the linear system  for the coefficients $J_k(z)$:
$$
-n_k + J_{k}(z) z^{-k d(w)} =J_{k}(z) K^{-kd(w)}_{(1,0)} \otimes K^{k d(w)}_{(1,0)} 
$$
which gives 
$$
J^{-}_{w}(z)=\exp\Big(- \sum\limits_{k=1}^{\infty}\, \dfrac{n_k\,K_{(1,0)}^{{\rr kd(w)}}\otimes K_{(1,0)}^{{\rr -kd(w)}} }{1-z^{-k d(w)} K_{(1,0)}^{{\rr kd(w)}}\otimes K_{(1,0)}^{{\rr -kd(w)}}  }\, \alpha^{w}_{-k} \otimes \alpha^{w}_{k} \Big).
$$
}
\subsubsection{}
The shift $\lambda \rightarrow \lambda -\tau_w$ corresponds to substitution $z\rightarrow z  q^{-w}$.
Thus by definition (\ref{JJ}) we obtain:
$$
\textbf{J}^{-}_{w}(z)=\exp\Big(- \sum\limits_{k=1}^{\infty}\, \dfrac{n_k\,K^{k d(w)}_{(1,0)}\otimes K^{-k d(w)}_{(1,0)} }{1-z^{-k d(w)} q^{k n(w)} K^{k d(w)}_{(1,0)}\otimes K^{-k d(w)}_{(1,0)}  }\, \alpha^{w}_{-k} \otimes \alpha^{w}_{k} \Big)
$$
{\rr and 
$$
\textbf{J}^{-}_{w}(z)^{-1}=\exp\Big( \sum\limits_{k=1}^{\infty}\, \dfrac{n_k\,K^{k d(w)}_{(1,0)}\otimes K^{-k d(w)}_{(1,0)} }{1-z^{-k d(w)} q^{k n(w)} K^{k d(w)}_{(1,0)}\otimes K^{-k d(w)}_{(1,0)}  }\, \alpha^{w}_{-k} \otimes \alpha^{w}_{k} \Big).
$$}
\subsubsection{}
From Section \ref{heisenbdef} it is clear that the antipode of $\Uq(\fg_w)$ has the following form:
$$
S_w (\alpha^{w}_{k})=- K_{( 1,0)}^{-k d(w)} \alpha^{w}_{k}\\
$$
From this we obtain:
$$
\textbf{m}\big( 1\otimes S_{w} (  \textbf{J}^{-}_w(z)^{-1}   ) \big)=\, : \exp\Big(-\sum\limits_{k=1}^{\infty}\, \dfrac{n_k \,K_{( 1,0)}^{k d(w)}}{1-z^{-k d(w)}q^{k n(w)} K_{( 1,0)}^{2 k d(w)} }\,  \alpha^{w}_{-k} \alpha^{w}_{k} \Big) :
$$
The symbol $::$ stands for the normal ordering meaning that all ``annihilation''  operators $\alpha^{w}_{k}$ with $k>0$ act first.

\subsubsection{}
The Cartan matrix of the Jordan quiver is $C=0$  and therefore
$\kappa=(C \bv -\bw)/2 =-r/2 $. Thus the shift $\lambda \rightarrow \lambda +\kappa $ corresponds to
$z \to z \hbar^{-r/2} =z\, K_{(1,0)}$.
From (\ref{bdef}) we obtain:
$$
\textbf{B}_{w}(z)=: \exp\Big(- \sum\limits_{k=1}^{\infty}\, \dfrac{n_k \,K_{( 1,0)}^{k d(w)}}{1-z^{-k d(w)}q^{k n(w)} K_{( 1,0)}^{k d(w)} }\,  \alpha^{w}_{-k} \alpha^{w}_{k} \Big):
$$

\subsubsection{}
Let $\cL={\cal{O}}(1)$ be the generator of the Picard group.  Let $\nabla \subset \R$ be the alcove specified by Theorem \ref{mainth}. If $\slope \in \nabla$, then the interval ${\rr (\slope-\cL,\slope)}$ contains all walls $w\in \Q$ such that
$-1\leq w<0$. {\rr We assume that $Const$  in Theorem \ref{mainth} for the case of $\cM(n,r)$ is trivial for all values of $n$ and $r$.}\footnote{This expectation is in agreement with explicit computations of the capped vertex functions \cite{pcmilect} for the first several values of $n$ and $r$.}
Therefore, by definition (\ref{dMdef}) we obtain the following explicit formula for quantum difference operator:
\be
\label{qdeins}
\begin{array}{|c|}
	\hline\\
	\geomM_{{\cal{O}}(1)}(z)={\cal{O}}(1)\!\!\!\!\prod\limits_{{w \in \Q} \atop { -1\leq w<0 }}^{\leftarrow}\!\!\!:\!\exp\Big(- \sum\limits_{k=1}^{\infty} \dfrac{n_k \,\hbar^{- k r d(w)/2}}{1-z^{-k d(w)}q^{k n(w)} \hbar^{-k r d(w)/2} }\,  \alpha^{w}_{-k} \alpha^{w}_{k} \Big)\!\!:\\  
	\\
	\hline
\end{array} \ \ \ 
\ee
where we used that in the K-theory of instanton moduli space $\cM(m,r)$ the central element acts by the scalar
$K_{(1,0)}=\hbar^{-r/2}$.
\subsubsection{}
Let us consider some limits of the difference operator. First, for all terms in
in the previous formula $d(w)>0$ and $n(w)<0$. Thus we have:
$$
\lim\limits_{q\to 0}\, \geomM_{{\cal{O}}(1)}(z) =\lim\limits_{z\to 0}\, \geomM_{{\cal{O}}(1)}(z) = {\cal{O}}(1)
$$
Second, to compute the limit of $\geomM_{{\cal{O}}(1)}(z q^{-1})$ as $q\rightarrow \infty$ we note that for all terms in (\ref{qdeins}) $d(w)+n(w)\geq 0$. Moreover $d(w)+n(w)=0$ only for $w=-1$.
We conclude that:
$$
\lim\limits_{q\to \infty}\, \geomM_{{\cal{O}}(1)}(z q^{-1}) ={\cal{O}}(1) : \exp\Big(- \sum\limits_{k=1}^{\infty}\, \dfrac{n_k \,\hbar^{ -k r /2}}{1-z^{-k} \hbar^{-k r/2 } }\,  \alpha^{-1}_{-k} \alpha^{-1}_{k} \Big):
$$
{\rr From (\ref{glumatKth}) and ${\cal{O}}(1) \alpha^{w}_{k} {\cal{O}}(1) ^{-1}=\alpha^{w+1}_{k}$ we find a formula for the glue operator in this case:}
$$
\textbf{G}=: \exp\Big( -\sum\limits_{k=1}^{\infty}\, \dfrac{n_k \,\hbar^{ -k r /2}}{1-z^{-k} \hbar^{-k r/2 } }\,  \alpha^{0}_{-k} \alpha^{0}_{k} \Big):
$$
The action of ``horizontal'' Heisenberg algebra $\alpha^{0}_{k}$ on the $K$-theory is given by
(\ref{f1}). Using these formula glue operator can be easily computed explicitly.
\subsubsection{}
Let us consider the example of $X=\textrm{Hilb}^{2}(\C^2)$. The walls which contribute to (\ref{qdeins}) are $w=-1$ and $w=-1/2$. The quantum difference operator takes the form:
$$
\calA_{{{\cal{O}}}(1)}=  T^{-1}_z {{\cal{O}}}(1) \textbf{B}_{-1}(z)\textbf{B}_{-\frac{1}{2}}(z)
$$
Using the identity (\ref{dva}) we can also write it in the form:
$$
\calA_{{{\cal{O}}}(1)}=  \textbf{B}_{0}(z)\textbf{B}_{\frac{1}{2}}(z)  {{\cal{O}}}(1) T^{-1}_z
$$
which means that:
$$
\geomM_{{{\cal{O}}}(1)}(z q^{-1})=\textbf{B}_{0}(z)\textbf{B}_{\frac{1}{2}}(z)  {{\cal{O}}}(1)
$$
Similarly, for $X=\textrm{Hilb}^{3}(\C^2)$ we have:
$$
\geomM_{{{\cal{O}}}(1)}(z q^{-1})=\textbf{B}_{0}(z)\textbf{B}_{\frac{1}{3}}(z) \textbf{B}_{\frac{1}{2}}(z)\textbf{B}_{\frac{2}{3}}(z) {{\cal{O}}}(1)
$$

\subsubsection{}
{\rr The torus acting on $X$ is two-dimensional. The corresponding coordinates are $t_1$ and $t_2$. The framing torus does not act on $X$ since $r=1$. We consider the one-dimensional torus corresponding to $\ker(\hbar)$. The coordinate on this torus is given by $t_1/t_2$. 
For this torus let $\Stab_{\pm}(\lambda)$ be the stable envelope of a fixed point $\lambda$ with a slope from the anti-canonical alcove, chambers $(t_1/t_2)^{\pm}\to 0$ and the standard polarization. Up to a multiple, as the elements of the Fock space, $\Stab_{+}(\lambda)$  and $\Stab_{-}(\lambda)$ coincide with the so called plethystic Schur polynomials:
$$
s_{\lambda}\Big(\frac{p_1}{1-t_1}, \frac{p_2}{1-t_1^2},\dots\Big),\ \ \ 	s_{\lambda}\Big(\frac{p_1}{1-t_2^{-1}}, \frac{p_2}{1-t_2^{-2}},\dots\Big)
$$
respectively. Here $s_{\lambda}(p_1,p_2,\dots)$ denotes the standard Schur polynomial associated with a partition $\lambda$. See Proposition 3.3 in \cite{KononoOkounkov} for a proof. 

Using a computer we find the following explicit examples in the basis of plethystic Schur polynomials corresponding to the chamber $(t_1/t_2)\to 0$.\footnote{We use a Maple package, implemented by the second author, which computes the action of $\gt$ on the Fock space. 
The package is available from the author upon request. } If the basis of partitions of $2$ is ordered as $[1,1],[2]$ we compute:
} 

$$
\begin{array}{l}
\textbf{B}_{0}(z \hbar^{1/2} )={\dfrac {z-1}{ \left( {z}^{2}{t_{{1}}}^{2}{t_{{2}}}^{2}-1 \right)
		\left( zt_{{1}}t_{{2}}-1 \right) }}
\left[ \begin {array}{cc} {z}^{2}t_{{1}}t_{{2}}-1&- \left( t_{{1}}t_{
	{2}}-1 \right) z\\ \noalign{\medskip}- \left( t_{{1}}t_{{2}}-1
\right) z&{z}^{2}t_{{1}}t_{{2}}-1\end {array} \right]\\
\\
\\
\textbf{B}_{\frac{1}{2}}(z \hbar^{1/2} )=1+{\dfrac {{z}^{2} \left( t_{{1}}t_{{2}}-1 \right) }{{z}^{2}{t_{{1}}}^{2}
		{t_{{2}}}^{2}-q}}
\left[ \begin {array}{cc} -1&t_{{2}}\\ \noalign{\medskip}t_{{1}}&-t_{
	{1}}t_{{2}}\end {array} \right]\\
\\
\\
{{\cal{O}}}(1)=\left[ \begin {array}{cc} t_{{2}}&0\\ \noalign{\medskip}-t_{{1}}t_{{2
}}+1&t_{{1}}\end {array} \right]
\end{array}
$$
If the basis of partitions of $3$ is ordered as $[1,1,1],[2,1],[3]$ we compute:
\begin{small}
	
	$$
	\begin{array}{l}
	\textbf{B}_{0}(z \hbar^{1/2})={\frac { \left( z-1 \right)  \left( zt_{{1}}t_{{2}}+1 \right)  \left(
			t_{{1}}t_{{2}}-1 \right) }{ \left( {z}^{3}{t_{{1}}}^{3}{t_{{2}}}^{3}-1
			\right)  \left( {z}^{2}{t_{{1}}}^{2}{t_{{2}}}^{2}-1 \right)  \left( z
			t_{{1}}t_{{2}}-1 \right) }}
	\times \\
	\\
	\left[ \begin {array}{ccc} {\frac { \left( {z}^{2}t_{{1}}t_{{2}}-1
			\right)  \left( {z}^{3}{t_{{1}}}^{2}{t_{{2}}}^{2}-1 \right) }{
			\left( t_{{1}}t_{{2}}-1 \right)  \left( zt_{{1}}t_{{2}}+1 \right) }}&
	-z \left( {z}^{2}t_{{1}}t_{{2}}-1 \right) &{\frac {{z}^{2} \left( z{t_
				{{1}}}^{2}{t_{{2}}}^{2}-1 \right) }{zt_{{1}}t_{{2}}+1}}
	\\ \noalign{\medskip}-z \left( {z}^{2}t_{{1}}t_{{2}}-1 \right) &{
		\frac {{z}^{4}{t_{{1}}}^{2}{t_{{2}}}^{2}-{z}^{3}{t_{{1}}}^{2}{t_{{2}}}
			^{2}+{z}^{2}{t_{{1}}}^{2}{t_{{2}}}^{2}-2\,{z}^{2}t_{{1}}t_{{2}}+{z}^{2
			}-z+1}{t_{{1}}t_{{2}}-1}}&-z \left( {z}^{2}t_{{1}}t_{{2}}-1 \right)
	\\ \noalign{\medskip}{\frac {{z}^{2} \left( z{t_{{1}}}^{2}{t_{{2}}}^{2
			}-1 \right) }{zt_{{1}}t_{{2}}+1}}&-z \left( {z}^{2}t_{{1}}t_{{2}}-1
	\right) &{\frac { \left( {z}^{2}t_{{1}}t_{{2}}-1 \right)  \left( {z}^
			{3}{t_{{1}}}^{2}{t_{{2}}}^{2}-1 \right) }{ \left( t_{{1}}t_{{2}}-1
			\right)  \left( zt_{{1}}t_{{2}}+1 \right) }}\end {array} \right]
	\\
	\\
	\\
	\textbf{B}_{1/3}(z \hbar^{1/2})=1+{\frac {{z}^{3} \left( t_{{1}}t_{{2}}-1 \right) }{{z}^{3}{t_{{1}}}^{3}
			{t_{{2}}}^{3}-q}}
	\left[ \begin {array}{ccc} -1&t_{{2}}&-{t_{{2}}}^{2}
	\\ \noalign{\medskip}t_{{1}}&-t_{{1}}t_{{2}}&t_{{1}}{t_{{2}}}^{2}
	\\ \noalign{\medskip}-{t_{{1}}}^{2}&{t_{{1}}}^{2}t_{{2}}&-{t_{{2}}}^{2
	}{t_{{1}}}^{2}\end {array} \right]
	\\
	\\
	\\
	\textbf{B}_{1/2}(z \hbar^{1/2})=1+{\frac {{z}^{2}t_{{1}}t_{{2}} \left( t_{{1}}t_{{2}}-1 \right) }{{z}^{2
			}{t_{{1}}}^{2}{t_{{2}}}^{2}-q}}\times\\
	\\
	\left[ \begin {array}{ccc} -{\frac {t_{{1}}t_{{2}}+t_{{1}}-1}{{t_{{1}
			}}^{2}t_{{2}}}}&{\frac {t_{{1}}t_{{2}}-1}{{t_{{1}}}^{2}}}&{\frac {t_{{
					2}}}{{t_{{1}}}^{2}}}\\ \noalign{\medskip}{\frac { \left( t_{{1}}t_{{2}
			}-1 \right)  \left( t_{{1}}t_{{2}}+t_{{1}}-1 \right) }{{t_{{2}}}^{2}{t
				_{{1}}}^{2}}}&-{\frac { \left( t_{{1}}t_{{2}}-1 \right) ^{2}}{{t_{{1}}
			}^{2}t_{{2}}}}&-{\frac {t_{{1}}t_{{2}}-1}{{t_{{1}}}^{2}}}
	\\ \noalign{\medskip}-{\frac { \left( t_{{1}}t_{{2}}+t_{{1}}-1
			\right)  \left( t_{{1}}t_{{2}}-t_{{1}}-1 \right) }{t_{{1}}{t_{{2}}}^{
				2}}}&{\frac { \left( t_{{1}}t_{{2}}-1 \right)  \left( t_{{1}}t_{{2}}-t
			_{{1}}-1 \right) }{t_{{1}}t_{{2}}}}&{\frac {t_{{1}}t_{{2}}-t_{{1}}-1}{
			t_{{1}}}}\end {array} \right]
	\\
	\\
	\\
	\textbf{B}_{2/3}(z \hbar^{1/2})=1+ {\frac {{z}^{3} \left( t_{{1}}t_{{2}}-1 \right) t_{{1}}t_{{2}}}{{z}^{3
			}{t_{{1}}}^{3}{t_{{2}}}^{3}-{q}^{2}}}\times \\
	\\
	\left[ \begin {array}{ccc} {\frac {t_{{1}}{t_{{2}}}^{2}-t_{{1}}-t_{{2
		}}}{{t_{{1}}}^{2}t_{{2}}}}&-{\frac {t_{{2}} \left( t_{{1}}t_{{2}}-t_{{
					1}}-1 \right) }{{t_{{1}}}^{2}}}&-{\frac {{t_{{2}}}^{2}}{{t_{{1}}}^{2}}
	}\\ \noalign{\medskip}-{\frac { \left( t_{{1}}t_{{2}}+t_{{1}}-1
			\right)  \left( t_{{1}}{t_{{2}}}^{2}-t_{{1}}-t_{{2}} \right) }{{t_{{2
			}}}^{2}{t_{{1}}}^{2}}}&{\frac { \left( t_{{1}}t_{{2}}+t_{{1}}-1
			\right)  \left( t_{{1}}t_{{2}}-t_{{1}}-1 \right) }{{t_{{1}}}^{2}}}&{
		\frac {t_{{2}} \left( t_{{1}}t_{{2}}+t_{{1}}-1 \right) }{{t_{{1}}}^{2}
	}}\\ \noalign{\medskip}{\frac { \left( {t_{{1}}}^{2}+t_{{1}}t_{{2}}-1
			\right)  \left( t_{{1}}{t_{{2}}}^{2}-t_{{1}}-t_{{2}} \right) }{t_{{1}
			}{t_{{2}}}^{2}}}&-{\frac { \left( t_{{1}}t_{{2}}-t_{{1}}-1 \right)
			\left( {t_{{1}}}^{2}+t_{{1}}t_{{2}}-1 \right) }{t_{{1}}}}&-{\frac {t_
			{{2}} \left( {t_{{1}}}^{2}+t_{{1}}t_{{2}}-1 \right) }{t_{{1}}}}
	\end {array} \right]\\
	\\
	\\
	{{\cal{O}}}(1)=\left[ \begin {array}{ccc} {t_{{2}}}^{3}&0&0\\ \noalign{\medskip}-
	\left( t_{{1}}t_{{2}}-1 \right)  \left( t_{{2}}+1 \right) t_{{2}}&t_{
		{1}}t_{{2}}&0\\ \noalign{\medskip} \left( t_{{1}}t_{{2}}-1 \right)
	\left( t_{{2}}{t_{{1}}}^{2}+{t_{{2}}}^{2}t_{{1}}-1 \right) &- \left(
	t_{{1}}t_{{2}}-1 \right)  \left( t_{{1}}+1 \right) t_{{1}}&{t_{{1}}}^{
		3}\end {array} \right]
	\end{array}
	$$
	
\end{small}
{\rr
\subsubsection{}
The operators $\textbf{B}_{w}(z)$ have remarkable symmetries and applications which are far from obvious. The explicit formulas for matrices of $\textbf{B}_{w}(z)$ simplify drastically if computed in the ``mixed" stable basis.
Let us denote by $\overset{\curvearrowright}{\textbf{B}}_{w}(z)$  the matrix of the operator $\textbf{B}_{w}(z)$ in the  mixed stable basis: the input is the stable basis before a wall~$w$, 
\be 
s^{w-\epsilon}_{\lambda}:=\Stab_{+, w-\epsilon}(\lambda)
\ee
and the output in the stable basis after $w$: 
\be \label{staft}
s^{w+\epsilon}_{\lambda}:=\Stab_{+,w+\epsilon}(\lambda).
\ee
for small enough $\epsilon$.
Explicitly, we have
\be \label{mixedbasdef}
\textbf{B}_{w}(z)(  s^{ w-\epsilon}_{\lambda} )  = \sum\limits_{\mu} \, \overset{\curvearrowright}{\textbf{B}}_{w}(z)_{\mu,\lambda} \, s^{ w+\epsilon}_{\mu}
\ee
One can show that the  matrix elements of $\overset{\curvearrowright}{\textbf{B}}_{w}(z)$ only depend on
parameters $z$ and $\hbar=t_1 t_2$ but are independent on the equivariant parameter $a=t_1/t_2$. Moreover, the matrix $\overset{\curvearrowright}{\textbf{B}}_{w}(z)$ coincides with the $K$-theoretic $R$-matrix of the {\it cyclic quiver variety} with $d(w)$ vertices.  This variety appears as a subvariety in the ``symplectic dual'' Hilbert scheme. Both $z$ and $\hbar$ play a role of equivariant parameters of a certain torus acting on the dual side. We refer to Theorem 12 in \cite{SmirnovQDE2} for a proof. The examples of the corresponding K-theoretic $R$-matrices for cyclic quiver varieties can be found in Appendix $D$ of \cite{SmirnovQDE2}.

We note also that the operator 
$$
\textbf{B}_{w}:=\lim_{z\to \infty} \textbf{B}_{w}(z)
$$
describes the monodromy of the quantum differential equation for $\textrm{Hilb}^{n}(\matC^2)$ \cite{OP}, around a loop containing the singularity $z_w=\exp(2 \pi i w)$. This means, in particular, that the operators $\textbf{B}_{w}$ provide a representation of the fundamental group
$$
\pi_1(\mathbb{P}^1 \setminus \{\textrm{singularities of qde for }  \textrm{Hilb}^{n}(\matC^2)\},0 )
$$
in the Fock space. We refer to Theorem 17 of \cite{SmirnovQDE2} for details a proof. A categorical version of these results is a topic of ongoing research \cite{BO}. }

\bibliographystyle{abbrv}

\bibliography{bib}

\newpage

\noindent
Andrei Okounkov\\
Department of Mathematics, Columbia University\\
New York, NY 10027, U.S.A.\\

\vspace{-12 pt}

\noindent
Institute for Problems of Information Transmission\\
Bolshoy Karetny 19, Moscow 127994, Russia\\

\vspace{-12 pt}

\noindent
Laboratory of Representation
Theory and Mathematical Physics \\
Higher School of Economics \\
Myasnitskaya 20, Moscow 101000, Russia

\vspace{12 mm}

\noindent
Andrey Smirnov\\
Department of Mathematics, University of North Carolina\\
Chapel Hill, NC, 27599 U.S.A.\\

\vspace{-12 pt}

\noindent
Steklov Mathematical Institute of Russian Academy of Sciences\\
Moscow, Gubkina 8, Russia, 117966.

\end{document}